\documentclass[a4paper, accepted=2020-08-02]{quantumarticle}
\pdfoutput=1

\usepackage{amssymb,amsmath,amsthm,amsfonts,dsfont, mathrsfs, multicol, soul}
\usepackage[numbers,comma,sort&compress]{natbib}
\usepackage[colorlinks,backref]{hyperref}
\usepackage{hypcap, float}
\usepackage{bbm, soul}
\usepackage{slashed}
\usepackage{braket}
\usepackage{url}
\usepackage{xkeyval, etoolbox, fancyhdr, tikz, ltxgrid, ltxcmds, lmodern, natbib, bbm}
\usepackage{pgfplotstable}
\usepackage{array}

\usepackage{graphicx}
\usepackage[font=small]{caption}
\usepackage[labelformat=simple]{subcaption}
\usepackage{float}
\usepackage{algorithmic}
\usepackage[ruled,vlined,linesnumbered]{algorithm2e}
\usepackage{footnote}
\usepackage{enumitem, pgfplotstable, booktabs, longtable}
\usepackage[margin=1in]{geometry}
\usepackage{authblk}

\setlist{listparindent=\parindent,
  parsep=0pt}
  
\usepackage{xcolor}
\usepackage{tikz, pgfplots}

\usepackage{mathtools}
\mathtoolsset{centercolon}

\newcommand{\floor}[1]{\lfloor{#1}\rfloor}
\newcommand{\ceil}[1]{\lceil{#1}\rceil}

\newcommand{\eq}[1]{(\ref{eq:#1})}
\newcommand{\thm}[1]{\hyperref[thm:#1]{Theorem~\ref*{thm:#1}}}
\newcommand{\defn}[1]{\hyperref[defn:#1]{Definition~\ref*{defn:#1}}}
\newcommand{\lem}[1]{\hyperref[lem:#1]{Lemma~\ref*{lem:#1}}}
\newcommand{\prop}[1]{\hyperref[prop:#1]{Proposition~\ref*{prop:#1}}}
\newcommand{\fig}[1]{\hyperref[fig:#1]{Figure~\ref*{fig:#1}}}
\newcommand{\tab}[1]{\hyperref[tab:#1]{Table~\ref*{tab:#1}}}
\renewcommand{\sec}[1]{\hyperref[sec:#1]{Section~\ref*{sec:#1}}}
\newcommand{\append}[1]{\hyperref[append:#1]{Appendix~\ref*{append:#1}}}
\newcommand{\cor}[1]{\hyperref[cor:#1]{Corollary~\ref*{cor:#1}}}

\newcommand{\N}{\mathbb{N}}

\newcommand{\comment}[1]{}

\newcommand{\ketbra}[2]{\ket{#1}\!\bra{#2}}

\newenvironment{proof-sketch}{%
	\proof}{\endproof}

\usepackage{thmtools}
\usepackage{thm-restate}

\newtheorem{theorem}{Theorem}

\newtheorem{lemma}[theorem]{Lemma}
\newtheorem{proposition}[theorem]{Proposition}

\newtheorem{corollary}[theorem]{Corollary}

\input{Qcircuit.tex}

\begin{document}

\title{Quantum Algorithms for Simulating the Lattice Schwinger Model}

\author[1,5]{Alexander F. Shaw}
\author[1]{Pavel Lougovski}
\author[2]{Jesse R. Stryker }
\author[3,4]{Nathan Wiebe}

\affil[1]{Quantum Information Science Group, Computational Sciences and Engineering Division, Oak Ridge National Laboratory, Oak Ridge, TN 37831, U.S.A.}
\affil[2]{Institute for Nuclear Theory, University of Washington, Seattle, WA 98195-1550, U.S.A.}
\affil[3]{Department of Physics, University of Washington, Seattle, WA 98195, U.S.A.}
\affil[4]{Pacific Northwest National Laboratory, Richland, WA 99354, U.S.A.}
\affil[5]{Department of Physics, University of Maryland, College Park, Maryland 20742, U.S.A.}

\date{August 5, 2020}

\begin{abstract}
\onecolumn
The Schwinger model (quantum electrodynamics in 1+1 dimensions) is a testbed for the study of quantum gauge field theories.
We give scalable, explicit digital quantum algorithms to simulate the lattice Schwinger model in both NISQ and fault-tolerant settings.
In particular, we perform a tight analysis of low-order Trotter formula simulations of the Schwinger model, using recently derived commutator bounds, and give upper bounds on the resources needed for simulations in both scenarios.
In lattice units, we find a Schwinger model on $N/2$ physical sites with coupling constant $x^{-1/2}$ and electric field cutoff $x^{-1/2}\Lambda$ can be simulated on a quantum computer for time $2xT$ using a number of $T$-gates or CNOTs in $\widetilde{O}( N^{3/2} T^{3/2} \sqrt{x} \Lambda )$ for fixed operator error.
This scaling with the truncation $\Lambda$ is better than that expected from algorithms such as qubitization or QDRIFT.
Furthermore, we give scalable measurement schemes and algorithms to estimate observables which we cost in both the NISQ and fault-tolerant settings by assuming a simple target observable---the mean pair density.
Finally, we bound the root-mean-square error in estimating this observable via simulation as a function of the diamond distance between the ideal and actual CNOT channels.
This work provides a rigorous analysis of simulating the Schwinger model, while also providing benchmarks against which subsequent simulation algorithms can be tested.
\end{abstract}

\section{Introduction}
The 20th century saw tremendous success in discovering and explaining the properties of Nature at the smallest accessible length scales, with just one of the crowning achievements being the formulation of the Standard Model.
Though incomplete, the Standard Model explains most familiar and observable phenomena as arising from the interactions of fundamental particles---leptons, quarks, and gauge bosons.
Effective field theories similarly explain emergent phenomena that involve composite particles, such as the binding of nucleons via meson exchange.
Many impressive long-term accelerator experiments and theoretical programs have supported the notion that these quantum field theories are adequate descriptions of how Nature operates at subatomic scales.

Given an elementary or high-energy quantum field theory, the best demonstration of its efficacy is showing that it correctly accounts for physics observed at longer length scales.
Progress on extracting long-wavelength properties has indeed followed. For example, in quantum chromodynamics (QCD), the proton is a complex hadronic bound state that is three orders of magnitude heavier than any of its constituent quarks.
However, methods have been found to derive its static properties directly from the theory of quarks and gluons;
using lattice QCD, calculations of the static properties of hadrons and nuclei with controlled uncertainties at precisions relevant to current experimental programs have been achieved~\cite{Yamazaki:2009ua,Beane:2012ey,Beane:2012vq,Yamazaki:2012hi,Inoue:2014ipa,Yamazaki:2015asa,Kirscher:2015yda,Contessi:2017rww,Iritani:2018zbt}.
It is even possible with lattice QCD to study reactions of nucleons~\cite{nplqcdcollaboration.beane.eaInitioCalculation15,nplqcdcollaboration.savage.eaProtonProtonFusion17}.
While scalable solutions for studies like these have been found, they are still massive undertakings that require powerful supercomputers to carry them out.

Generally, however, scalable calculational methods are not known for predicting non-perturbative emergent or dynamical properties of field theories.
Particle physics, like other disciplines of physics, generically suffers from a quantum many-body problem.
The Hilbert space grows exponentially with the system size to accommodate the entanglement structure of states and the vast amount of information that can be stored in them. This easily overwhelms any realistic classical-computing architecture.
Those problems that are being successfully studied with classical computers tend to benefit from a path integral formulation of the problem;
this formulation trades exponentially-large Hilbert spaces for exponentially-large configuration spaces, but in these particular situations it suffices to coarsely importance-sample from the high-dimensional space with Monte Carlo methods~\cite{PhysRevD.10.2445,creutz1980monte}.
Analytic continuation to imaginary time (Euclidean space) is key to this formulation.
Real-time processes are intrinsically ill-suited for importance-sampling in Euclidean space, and they (along with finite-chemical-potential systems) are usually plagued by exponentially hard ``sign problems.''%~\cite{%
%  mak.chandlerSolvingSign90,
%  samsonClassicalEffective95,
%  bergkvist.henelius.eaReductionSign03,
%  splittorff.svetitskySignProblem07,
%  fukushima.hidakaModelStudy07,
%  deforcrandSimulatingQCD09,
%  aartsCanStochastic09,
%  endres.kaplan.eaNoiseSign11,
%  unsalThetaDependence12,
%  greensite.myers.eaQCDSign13}.

Quantum computing offers an alternative approach to simulation that should avoid such sign problems by returning to a Hamiltonian point of view and leveraging controllable quantum systems to simulate the degrees of freedom we are interested in~\cite{feynman1982, lloyd1996universal,lanyon2010towards,Jordan_Lee_Preskill_2012}.
The idea of simulation is to map the time-evolution operator of the target quantum system into a sequence of unitary operations that can be executed on the controllable quantum device---the quantum computer.
The quantum computer additionally gives a natural platform for coherently simulating the entanglement of states that is obscured by sampling classical field configurations.
We will specifically consider digital quantum computers, i.e., the desired operations are to be decomposed into a set of fundamental quantum gates that can be realized in the machine.
Inevitable errors in the simulation can be controlled by increasing the computing time used in the simulation (often by constructing a better approximation of the target dynamics on the quantum computer) or by increasing the memory used by the simulation (e.g. increasing the number of qubits to better approximate the continuous wave function of the field theory with a fixed qubit register size, through latticization and field digitization).

Beyond (and correlated with) the qubit representation of the field, it is important to understand the unitary dynamics simulation cost in terms of the number of digital quantum gates required to reproduce the dynamics to specified accuracy on a fault tolerant quantum computer.
Furthermore, since we are currently in the so-called Noisy Intermediate Scale Quantum (NISQ) era, it is also important to understand the accuracy that non-fault tolerant hardware can achieve as a function of gate error rates and the available number of qubits.
For many locally-interacting quantum systems, the simulation gate count is polynomial in the number of qubits~\cite{lloyd1996universal,berry2007efficient}, but for field theories, even research into the scaling of these costs is still in its infancy.

While strides are being made in developing the theory for quantum simulation of gauge field theories \cite{%
byrnes.yamamotoSimulatingLattice06,
banerjee.bogli.eaAtomicQuantum13,
zohar.cirac.eaColdAtomQuantum13,
tagliacozzo.celi.eaSimulationNonAbelian13,
wieseUltracoldQuantum13,
wieseQuantumSimulating14,
rico.pichler.eaTensorNetworks14,
mezzacapo.rico.eaNonAbelianSU15,
martinez.muschik.eaRealtimeDynamics16,
muschik.heyl.eaWilsonLattice17,
zohar.farace.eaDigitalQuantum17,
Klco:2018kyo,
davoudi.hafezi.eaAnalogQuantum19,
magnifico2019z,
nuqscollaboration.lamm.eaGeneralMethods19,
alexandru.bedaque.eaGluonField19,
klco.stryker.eaSUNonAbelian19,
raychowdhury.strykerLoopString19,
banuls.blatt.eaSimulatingLattice19,
schweizer.grusdt.eaFloquetApproach19,
avkhadiev.shanahan.eaAcceleratingLattice19,
harmalkar.lamm.eaQuantumSimulation20},
much remains to be explored in terms of developing explicit gate implementations that are scalable and can ultimately surpass what is possible with classical machines.
This means that there is much left to do before we will even know whether or not quantum computers suffice to simulate broad classes of field theories.

In this work, we report explicit algorithms that could be used for time evolution in the (massive) lattice Schwinger model, a gauge theory in one spatial dimension that shares qualitative features with QCD~\cite{schwinger1962model}---making it a testbed for the gauge theories that are so central to the Standard Model. 
Due to the need to work with hardware constraints that will change over time, we specifically provide algorithms that could be realized in the NISQ era, followed by efficient algorithms suitable for fault-tolerant computers where the most costly operation has transitioned from entangling gates (CNOT gates), to non-Clifford operations ($T$-gates).
We prove upper bounds on the resources necessary for quantum simulation of time evolution under the Schwinger model, finding sufficient conditions on the number and/or accuracy of gates needed to guarantee a given error between the output state and ideal time-evolved state.
Additionally, we give measurement schemes with respect to an example observable (mean pair density) and quantify the cost of simulation, including measurement, within a given rms error of the expectation value of the observable.
%{We also provide numerical analyses of these costs, which may help understand quantum-computational costs in parameter regimes inaccessible to competing classical-computation methods, once these regimes are more sharply identified. We do not sharply identify these regimes here.}
{We also provide numerical analyses of these costs, which can be used when determining the quantum-computational costs involved with accessing parameter regimes beyond the reach of classical computers.}

{Outside the scope of this article, but necessary for a comprehensive analysis of QFT simulation, are initial state preparation, accounting for parameter variation in extrapolating to the continuum, and rigorous comparison to competing classical algorithms;
we leave these tasks to subsequent work.}

%{This paper does not contain further features one may expect of a comprehensive how-to on simulating the Schwinger Model on a digital quantum computer; namely, a discussion of sufficient selections of lattice parameters (i.e. field cutoff and lattice spacing) to accurately approximate continuum dynamics, and a discussion of initial state preparation. We leave these to future work.}

The following is an outline of our analysis.
In \sec{schwinger}, we describe the Schwinger Model Hamiltonian, determine how to represent our wavefunction in qubits (\sec{hamrep}), then describe how to implement time evolution using product formulas (\sec{trottime}, \sec{trotschwinger}) and why we choose this simulation method (\sec{compare}).
Additionally, we define the computational models we use to cost our algorithms in both the near- and far-term settings (\sec{modeldef}).
In \sec{nisqimplement} and \sec{circuiterr} we give our quantum circuit implementations of time evolution under the Schwinger model for the near- and far-term settings, respectively.
Then, in \sec{faulttolerant}, we determine a sufficient amount of resources to simulate time evolution as a function of simulation parameters in the far-term settings and discuss how our results would apply in near-term analyses.
With our implementations of time evolution understood, we then turn to the problem of measurement in \sec{measurement}, where we analyze two methods: straightforward sampling, and a proposed measurement scheme that uses amplitude estimation.
We apply our schemes to a specific observable -- the mean pair density -- and cost out full simulations (time evolution and measurement) in \sec{fullsim}.
We report numerical analysis of the costing in the far-term setting in \append{costnumeric}.
We then discuss how to take advantage of the geometric locality of the lattice Schwinger model Hamiltonian may be used in simulating time evolution or in estimating expectation values of observables in \sec{locality} before concluding.

\section{Schwinger Model Hamiltonian}\label{sec:schwinger}
Quantum electrodynamics in one spatial dimension, the Schwinger model \cite{schwinger1962model}, is perhaps the simplest gauge theory that  shares key features of QCD such as confinement and spontaneous breaking of chiral symmetry.
The continuum Lagrangian density of the Schwinger model is given by
\begin{eqnarray}
{\cal L} & = & 
- {1\over {4g^2}}\sum_{\mu,\nu =0}^{1} F_{\mu\nu}F^{\mu\nu} + \sum_{\alpha,\beta=1}^{2} \overline{\psi}_\alpha \left(i \sum_{\mu} \gamma^\mu D_\mu - m\right)_{\alpha \beta} \psi_{\beta}
\ \ \ .
\end{eqnarray}
Here, $F_{\mu\nu}=\partial_\mu A_\nu - \partial_\nu A_\mu$ is the field strength tensor, $\psi$ is a two-component Dirac field with $\bar{\psi}=\psi^\dagger \gamma^0 $ its relativistic adjoint, and $D_\mu = \partial_\mu -i A_\mu $ is the covariant derivative acting on $\psi$.

In their seminal paper~\cite{KS1975} Kogut and Susskind introduced a Hamiltonian formulation of Wilson's action-based lattice gauge theory~\cite{PhysRevD.10.2445} in the context of SU(2) lattice gauge theory.
The analogous Hamiltonian for compact U(1) electrodynamics can be written as~\cite{banksKS1976strong}
\begin{equation}\label{eq:ham}
    H = H_E + H_I + H_M \ ,
\end{equation}
with
\begin{align}
    H_E &=  \sum_{r} E_{r}^2 \label{eq:SchwingerHE} \\
    H_I &= x \sum_r \left[ U_r \psi^\dagger_r \psi_{r+1} -  U_r^\dagger  \psi_r \psi_{r+1}^\dagger \right] \label{eq:SchwingerHI} \\
    H_M &= \mu \sum_r (-)^r \psi^\dagger_r \psi_r, \label{eq:SchwingerHM}
\end{align} 
where $\mu=2m/(ag^2)$ and $x=1/(ag)^2$, with $a$ the lattice spacing, $m$ the fermion mass and $g$ the coupling constant, as described in \cite{banksKS1976strong}.
($H$ has been non-dimensionalized by rescaling with a factor $2 a^{-1} g^{-2}$.)
The electric energy $H_E$ is given in terms of $E_{r}$, the dimensionless integer electric fields living on the links $r$ of the lattice.
For convenience, the most important definitions of symbols used in this paper are tabulated in \sec{notation}.
The interaction or ``hopping'' Hamiltonian $H_I$ corresponds to minimal coupling of the Dirac field to the gauge field.
On the lattice, the (temporal-gauge) gauge field $A_{\mu}=(0,A_1)$ is encoded in the unitary ``link operators'' $U_{r}$, which are complex exponentials of the vector potential,
\begin{align}
    U_{r} &= e^{i a A_{r}}.
\end{align}
These are defined to run ``from site $r$ to site $r+1$.''
The continuum commutation relations of $E$ and $A$ translate to the lattice commutation relations
\begin{align}\label{eq:linkcom}
    [E_{r}, U_{s}] &= +U_{r} \delta_{rs} \\
    [E_{r}, U_{s}^\dagger] &= -U_{r}^\dagger \delta_{rs} .
\end{align}
The fermionic operators $\psi_r,\psi_r^\dagger$ coupled to the link operators live on lattice vertices $r$ and satisfy the anticommutation relations
\begin{align}
  \{ \psi_r , \psi_s \} = \{ \psi_r^\dagger , \psi_s^\dagger \} &= 0 \ ,\\
    \{\psi_r , \psi_s^\dagger \} & = \delta_{rs}.
\end{align}
Finally, $H_M$ is the mass energy of the fermions;
the alternating sign $(-)^r$ corresponds to the use of staggered fermions~\cite{KS1975}.

The full Hamiltonian acts on a Hilbert space characterized by the on-link bosonic Hilbert spaces and the on-site fermionic excitations.
Moreover, the physical interpretation of the fermionic Hilbert spaces is
\begin{align*}
    \text{occupied even site } &\sim \text{ presence of a positron} \\
    \text{empty odd site } &\sim \text{ presence of an electron},
\end{align*}
so the operator that counts electric charge at site $r$ is given by
\begin{equation}
    \rho_r = \psi_r^\dagger \psi_r + \frac{(-)^r - 1}{2}.
\end{equation}
Gauge invariance leads to a local constraint, Gauss's law, relating the states of the links to the states of the sites:
\begin{align}
    \mathcal{G}_r &= E_r - E_{r-1} - \rho_r \\
    \mathcal{G}_r \ket{\text{physical state}} &= 0.
\end{align}
These $\mathcal{G}_r$ operators generate gauge transformations and must commute with the full Hamiltonian.
It is well-known that in $1$d one can gauge-fix all the links (except for one in a periodic volume), essentially removing the gauge degree of freedom.
Correspondingly, the constraint can be solved exactly to have purely fermionic dynamical variables at the cost of having long-range interactions.
This is not considered in this work since it is not representative of the situations encountered in higher dimensions.

\subsection{Qubit Representation of the Hamiltonian}\label{sec:hamrep}

A link Hilbert space is equivalent to the infinite-dimensional Hilbert space of a particle on a circle, and it is convenient to use a discrete momentum-like electric eigenbasis $\ket{\varepsilon}_r$, in which $E_r$ takes the form
\begin{equation}\label{Eq:E_r_squared}
    E_r = \sum_\varepsilon \varepsilon \ket{\varepsilon}_r \bra{\varepsilon}_r \quad \Rightarrow E_r^2 = \sum_\varepsilon \varepsilon^2 \ket{\varepsilon}_r \bra{\varepsilon}_r 
\end{equation}
In this basis, the link operators $U_r$ and $U_r^\dagger$ have simple raising and lowering actions,
\begin{equation}
    U_r = \sum_\varepsilon \ket{\varepsilon+1}_r \bra{\varepsilon}_r \qquad U_r^\dagger = \sum_\varepsilon \ket{\varepsilon-1}_r \bra{\varepsilon}_r.
\end{equation}
The microscopic degrees of freedom  are depicted in \fig{1dlattice}.

In order to represent the system on a quantum computer, we first truncate the Hilbert space.
On the link Hilbert spaces, we do this by periodically wrapping the electric field at a cutoff $\Lambda$:
\begin{equation}\label{eq:ecutoff}
  E_r = \sum_{\varepsilon = - \Lambda}^{\Lambda - 1} \varepsilon \ket{\varepsilon}_r \bra{\varepsilon}_r \ , \qquad U_r \ket{\Lambda-1}_r = \ket{-\Lambda}_r \ , \qquad U_r^\dagger \ket{-\Lambda}_r =\ket{\Lambda-1}_r.
\end{equation}
This modifies the on-link commutation relation \eq{linkcom} at the cutoff:
\begin{equation}
    [E_r, U_r] = U_r - 2\Lambda \ket{-\Lambda}_r \bra{\Lambda - 1}_r,
\end{equation}
\begin{equation}
    \left[ E_{r} , U_{r}^{\dagger} \right] = - U_{r}^{\dagger} + 2 \Lambda \ket{\Lambda - 1}_{r} \bra{- \Lambda}_{r}. 
\end{equation}
This altered theory (which is distinct from simulating a $Z(2\Lambda)$ theory due to $H_E$ remaining quadratic in the electric field) allows the $U_r$ operator to be easily simulated using a cyclic incrementer quantum circuit.
However, the periodic identification requires one to minimize simulating states that approach the cutoff, in order to avoid mapping $\ket{\Lambda - 1} \leftrightarrow \ket{-\Lambda}$ (which we expect to not be an issue in practice).
The asymmetry between the positive and negative bounds is introduced to make the link Hilbert space even-dimensional, so the electric field is more naturally encoded with a register of qubits.
The number of qubits on the link registers is given by 
\begin{equation}
    \eta = \log(2\Lambda),
\end{equation}
where we will implicitly assume that $\Lambda$ is a non-negative power of $2$ and all logarithms are base 2.

We map the electric field states into a computational basis with $\eta$ qubits, which is in a tensor product space of $\eta$ identical qubits.
Here, the Hilbert space of any single qubit is spanned by computational basis states $\ket{0}$ and $\ket{1}$ and the Pauli operators take the form
\begin{align*}
  X = \ket{0}\bra{1} + \ket{1}\bra{0} \ , \quad Y &= -i\ket{0}\bra{1} + i\ket{1}\bra{0} \ , \quad Z = \ket{0}\bra{0} - \ket{1}\bra{1} \ , \\
  \sigma^{\pm} &= \frac{1}{2} (X \pm i Y) \ .
\end{align*}
A non-negative integer $0\leq j < 2^\eta$ is represented on the binary $\eta$-qubit register as
\begin{equation}
    \ket{j} = \Ket{ \sum_{n=0}^{\eta-1} j_n 2^n } = \bigotimes_{n=0}^{\eta-1} \ket{j_n} \ .
\end{equation}
Using this unsigned binary computational basis, the electric eigenbasis $\ket{\varepsilon}$ is encoded via
\begin{equation}\label{eq:emap}
\varepsilon = j - \Lambda.
\end{equation}
The last ingredient to make the theory representable by a quantum computer is a volume cutoff; the lattice is taken to have an even number of sites $N$, corresponding to $N/2$ physical sites.

In addition to the above truncations, we will take ``open'' boundary conditions (the system does not source electric flux beyond its boundary sites) and a vanishing background field.
Note that the Gauss law constraint limits the maximum electric flux saturation possible in this setup:
If $\Lambda > \lceil N/4 \rceil $ then the quantum computer can actually represent the entire (truncated) physical Hilbert space.

\begin{figure}
    \centering
    \includegraphics[]{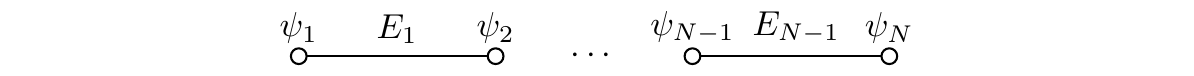}
    \caption{Labeling of the 1d lattice degrees of freedom on a finite, non-periodic lattice.}
    \label{fig:1dlattice}
\end{figure}

The Jordan-Wigner transformation is a simple way to represent fermionic creation/annihilation operators with qubits.
In the transformation, the state $\ket{0}$ encodes the vacuum state for a fermionic mode and $\ket{1}$ represents an occupied state.
For the creation operator $\psi_r^\dagger$, we have
\begin{equation}
    \psi_r^\dagger = \left(X_r - i Y_r \right)/2 \prod_{j=1}^{r-1} Z_j
\end{equation}
(We associate $\psi^\dagger$ with $\sigma^-$ because in the computational Pauli-$Z$ basis $\sigma^-=\ket{1}\bra{0}$.)

Therefore the interaction Hamiltonian \eq{SchwingerHI} can be expressed as
\begin{equation}
H_I= x\sum_{r=1}^{N-1}\left[ U_r \sigma_r^- \sigma_{r+1}^+  +  U_r^\dagger \sigma_r^+ \sigma_{r+1}^- \right]
\end{equation}
\cite{banksKS1976strong} with the Hamiltonian now acting on a tensor product space of link Hilbert spaces and on-site qubit Hilbert spaces.
Tensor product notation for operators will often be supressed;
at times a superscript ${}^{b}$ (${}^{f}$) may be added to an operator $\mathcal{O}_r$ to emphasize that $\mathcal{O}_r$ acts on the $r^{\mathrm{th}}$ bosonic (spin) degree of freedom.
In terms of the Pauli matrices, $H_I$ can be reexpressed as
\begin{align}
H_I &=  \frac{1}{4} x \sum_{r} \left[ \left(U_r + U_r^\dagger  \right)\left(X_r X_{r+1} +Y_r Y_{r+1} \right) + i \left(U_r - U_r^\dagger  \right)\left(X_r Y_{r+1} - Y_r X_{r+1} \right) \right] .
  \label{eq:JWHam}
\end{align}
$H_I$ in this one-dimensional model is seen to involve only nearest-neighbor fermionic interactions after Jordan-Wigner transformation, but in higher dimensions one will generally have to deal with non-local Pauli-$Z$ strings.
\noindent
Note that in~\eqref{eq:JWHam} the Hamiltonian manifestly expands to a sum of unitary operators.
This is true in spite of truncation due to the periodic wrapping of $U_r$ that preserves its unitarity.
This observation will be crucial to our development of simulation circuits.

\subsection{Trotterized Time Evolution}\label{sec:trottime}
Perhaps the central task in simulating dynamics on quantum computers is to compile the evolution generated by the Hamiltonian into a sequence of implementable quantum gates.
Trotter-Suzuki decompositions were the first methods that were proposed to efficiently simulate quantum dynamics on a quantum computer~\cite{lloyd1996universal,zalka1998simulating,boghosian1998simulating}.
The idea behind these methods is that one first begins by decomposing the Hamiltonian into a sum of the form $H=\sum_j H_j$ where each $H_j$ is a simple Hamiltonian for which one can design a brute force circuit to simulate its time evolution.
Examples of elementary $H_j$ that can be directly implemented include Pauli-operators~\cite{lloyd1996universal}, discretized position and momentum operators~\cite{zalka1998simulating,kivlichan2017bounding} and one-sparse Hamiltonians~\cite{aharonov2003adiabatic,berry2007efficient,childs2010simulating}.

The simplest Trotter-Suzuki formula used for simulating quantum dynamics is
\begin{align}\label{eq:trot}
    U_1(t)&:=\prod_{j=1}^m e^{-iH_jt}
\end{align}
Here we implicitly take all products to be lexicographically ordered in increasing order of the index $j$ of $H_j$ from right to left in the product.
Similarly we use $\prod_{j=m}^1$ to represent the opposite ordering of the product.
The error in the approximation can be bounded above by~\cite{huyghebaert1990product,Yuan}
\begin{align}
    \left\|e^{-iHt}-U_1(t)\right\| &\leq \frac{1}{2}\sum_{x>y}\|[H_x,H_y] \|t^2.
    \label{eq:u1bound}
\end{align}
where $\|\cdot\|$ is the spectral norm.
This shows that not only is the error in the approximation at second order in $t$, but it is also zero when the operators in question commute.
In general, the error in a sequence of such approximations can be bounded using Box 4.1 of~\cite{nielsen2002quantum} which states that for any unitary $U$ and $V$ 
\begin{equation}
\label{eq:nielsenChuang4.1}
    \| U^s - V^s\| \le s \|U-V\|.
\end{equation}
This implies that if we break a total simulation interval $T$ into $s$ time steps, then the error can be bounded by
\begin{equation}
    \left\|e^{-iHT}-U_1^s(T/s)\right\|  \leq\frac{1}{2}\sum_{x>y}\|[H_x,H_y] \|\frac{T^2}{s}.
\end{equation} 
Therefore if the propagators for each of the terms in the Hamiltonian can be individually simulated, then $e^{-i HT}$ can be also simulated within arbitrarily small error by choosing a value of $s$ that is sufficiently large.

The simple product formula in~\eqref{eq:trot} is seldom optimal for dynamical simulation because the value of $s$ needed to perform the simulation within $\epsilon$ spectral-norm error grows as $O(T^2/\epsilon)$.
This can be improved by switching to the symmetric Trotter-Suzuki formula, also known as the Strang splitting, which is found by symmetrizing the first-order Trotter-Suzuki formula to cancel out the even-order errors.
This formula reads (see~\cite{wecker2014gate} and~\cite{Yuan})
\begin{equation}
 \label{eq:2ndordertrota}
     U_2(t) :=\prod_{j=1}^m e^{-iH_jt/2}\prod_{j=m}^1 e^{-iH_jt/2} 
\end{equation}
\begin{equation}
 \label{eq:2ndordertrot}
    \left\|e^{-iHt}-U_2(t)\right\|\le \frac{1}{12}\sum_{x,y>x} \|[[H_x,H_y],H_x]\|t^3 + \frac{1}{24}\sum_{x,y>x,z>x} \|[[H_x,H_y],H_z]\|  t^3.
\end{equation}
This expression can be seen to be tight (up to a single application of the triangle inequality) in the limit of $t\ll 1$ from the Baker-Campbell-Hausdorff formula~\cite{wecker2015solving, Yuan}.
The improved error scaling similarly leads to a number of time-steps that scales as $O(T^{3/2}/\sqrt{\epsilon})$.
Higher-order variants can be constructed recursively \cite{doi:10.1063/1.529425} to achieve still better scaling~\cite{wiebe2010higher}, via
\begin{align}
    U_{2k+2}(t) &= U_{2k}^2(s_kt)U_{2k}([1-4s_k]t)U_{2k}^2(s_kt)\nonumber\\
    \|e^{-iHt} - U_{2k+2}(t) \| &\le 2\left(2(k+1)m(5/3)^{k}\max_x \|H_x\|t\right)^{2k+3}.
    \label{eq:trotorder}
\end{align}
\noindent
where
\begin{equation}
    s_k = (4-4^{1/(2k+1)})^{-1}.
\end{equation}
Until very recently, the scaling of the error in these formulas as a function of the nested commutators of the Hamiltonian terms was not known beyond the fourth-order product formula.
Recent work, however, has shown the following tighter bound for Trotter-Suzuki formulas:

\begin{align}
    \|e^{-iHt} - U_{2k+2}(t) \| &\le 4 \frac{5^{k}\cdot (2\cdot 5^k)^{2k+2} t^{2k+3}}{(2k+3)} \sum_{\gamma_{2k+3}=1}^M\cdots \sum_{\gamma_{2}=1}^M \|[H_{\gamma_{2k+3}},\ldots, [H_{\gamma_2},H_{\gamma}]]\|\nonumber\\
    &= \frac{5^{2k^2+3k}4^{k+2} t^{2k+3}}{(2k+3)} \sum_{\gamma_{2k+3}=1}^M\cdots \sum_{\gamma_{2}=1}^M \|[H_{\gamma_{2k+3}},\ldots, [H_{\gamma_2},H_{\gamma}]]\|.
\end{align}
This specifically follows from Eqns. (189-191) of~\cite{Yuan}.

A major challenge that arises in making maximal use of this expression is due to the fact that there are many layers of nested commutators.
This means that, although this expression closely predicts the error in Trotter-Suzuki expansions~\cite{Yuan}, precise knowledge of the commutators is needed to best leverage this formula.
As the number of commutators in the series grows as $M^{2k+3}$ this makes explicit calculation of the commutators costly for all but the smallest values of $M$ and $k$~\cite{babbush2015chemical,kivlichan2018quantum}.
This can be ameliorated to some extent by estimating the sum via a Monte-Carlo average~\cite{reiher2017elucidating}, but millions of norm evaluations are often still needed to reduce the variance of the estimate to tolerable levels.
Such sampling further prevents the results from being used as an upper bound without invoking probabilistic arguments and loose-tail probability bounds such as the Markov inequality.
Even once these bounds are evaluated, it should be noted that the prefactors in the analysis are unlikely to be tight~\cite{Yuan} which means that in order to understand the true performance of high-order Trotter formulas for lattice discretizations of the Schwinger model, it may be necessary to attempt to extrapolate the empirical performance from small-scale numerical simulations of the error.

These higher-order Trotter-Suzuki formulas  also are seldom preferable for simulating quantum dynamics for modest-sized quantum systems.
Indeed, many studies have verified that low-order formulas (such as the second-order Trotter formula) actually can cost fewer computational resources than their high-order brethren~\cite{childs2018toward,wecker2014gate,reiher2017elucidating} for realistic simulations.
Furthermore, the second-order formula can also outperform asymptotically-superior simulation methods such as qubitization~\cite{childs2018toward,low2019hamiltonian} and linear-combinations of unitaries (LCU) simulation methods~\cite{childs2012hamiltonian,berry2015simulating}.
This improved performance is anticipated to be a consequence of the fact that the Trotter-Suzuki errors depend on the commutators, rather than the magnitude of the Hamiltonian terms' coefficients as per~\cite{low2019hamiltonian,childs2012hamiltonian,berry2015simulating}.

For the above reasons (as well as the fact that tight and easily-evaluatable error bounds exist for the second-order formula), we choose to use the second-order Trotter formula in the following discussion and explicitly evaluate the commutator bound for the error given in~\eqref{eq:2ndordertrot}, which is the tightest known bound for the second-order Trotter formula.

\subsection{Comparison to Qubitization and Linear Combination of Unitaries Methods}\label{sec:compare}
One aspect in which the second-order Trotter formulas that we use perform well relative to popular methods, such as qubitization~\cite{low2019hamiltonian}, linear combinations of unitaries~\cite{childs2012hamiltonian,low2019hamiltonian,berry2015simulating} or their classically-controlled analogue QDRIFT~\cite{campbell2019random,berry2019time}, is that the complexity of the Hamiltonian simulation scales better with the size of the electric cutoff, $\Lambda$.
The complexities of qubitization, LCU and QDRIFT all scale linearly with the sum of the Hamiltonian terms' coefficients when expressed as a sum of unitary matrices.
This leads to a scaling  of the Trotter step number with $\Lambda$ of $\widetilde{O}(\Lambda^2)$.
Instead, it is straightforward to note that the norm of the commutators $[E_r^2,[E_r^2,H_I]]$ scale as $O(\Lambda^2)$ and that all other terms scale at most as $O(\Lambda^2)$ (see~\lem{combound} for more details).
This leads to a number of Trotter steps needed for the second-order Trotter-Suzuki formula that is in $O(\Lambda)$ from~\eqref{eq:2ndordertrot}.
Thus, accounting for the logarithmic-sized circuits that we will use to implement these terms, the total cost for the simulation is in $\widetilde{O}(\Lambda)$, which is (up to polylogarithmic factors) quadratically better than LCU or qubitization methods and without the additional spatial overheads they require~\cite{childs2018toward}.

The fact that Trotter-Suzuki methods can be the preferred method for simulating discretizations of unbounded Hamiltonians is also noted in~\cite{somma2016trotter} where commutator bounds are used to show that quartic oscillators can be simulated more efficiently with Trotter-series decompositions than would be expected from a standard implementation of qubitization or LCU methods.
This is why we focus on Trotter methods for this study here.
It is left for subsequent work to examine the performance of post-Trotter methods discussed above.

\subsection{Trotter-Suzuki Decomposition for the Schwinger Model}\label{sec:trotschwinger}

In order to form a Trotter-Suzuki formula, we decompose the Schwinger model Hamiltonian into a sum of simulatable terms.
Ideally, each term would also be efficiently simulatable.
Examples of such terms are Pauli operators or operators diagonal in the computational basis (with efficiently computable diagonal elements)~\cite{nielsen2002quantum,berry2007efficient}.

First we define $T_r$ to be the interaction term coupling sites $r$ and $r+1$ and define $D_r$ to be the (diagonalized) mass plus electric energy associated to site $r$:
\begin{align}\label{eq:terms}
    T_r &:= x\left(  \frac{1}{4} (U_r + U_r^\dagger)(X_r X_{r+1} + Y_r Y_{r+1}) + \frac{i}{4} (U_r - U_r^\dagger)(X_r Y_{r+1} - Y_r X_{r+1})\right),  \nonumber \\
    D_r^{(M)} &:= -\frac{\mu}{2}(-1)^r Z_r \quad \text{and} \quad D_r^{(E)} := E_r^2 (1- \delta_{r,N}),
\end{align}
where we further define
\begin{equation}\label{eq:Sumterms}
    H = \sum_{r=1}^{N} \ ( T_r + D_r) \ , \quad \text{with} \quad D_r := D_r^{(M)} + D_r^{(E)}.
\end{equation}
The $D_r$ are each of a sum of two terms that commute so that
\begin{equation}
    e^{-iD_r t} = e^{-iD_r^{(M)} t}e^{-iD_r^{(E)}t}.
\end{equation}
The $T_r$ are ``hopping'' terms, which we further decompose in \sec{kineticimplementRC} into four non-commuting Hermitian terms $T_r^{(i)}$:
\begin{equation}
    T_r = T_r^{(1)} + T_r^{(2)} + T_r^{(3)} + T_r^{(4)}.
\end{equation}
Using these terms, we determine the corresponding resource scaling of second-order simulation \eq{2ndordertrota}.
Explicitly, we choose to approximate the time-evolution operator with
\begin{equation}
\label{eq:timestep}
    V(t) := \prod_{r=1}^{N-1} \left ( \prod_{k\in \{M,E\}} e^{-i D_r^{(k)} t/2}\prod_{j=1}^4e^{-i T_r^{(j)} t/2}  \right )e^{-i D_N^{(M)} t} \prod_{r=N-1}^1  \left ( \prod_{j=4}^1 e^{-i T_r^{(j)} t/2}\prod_{k\in \{E,M\}} e^{-i D_r^{(k)} t/2}\right ).
\end{equation}
The following sections discuss the computational models we will use to analyze the cost of implementing time evolution and measurement, using the above approximation.

\subsection{Definitions of Noisy Entangling Gate and Fault-Tolerant Models}\label{sec:modeldef}

Throughout the rest of the paper, our analysis will usually operate under one of the two computational models discussed here.

\subsubsection{Noisy Entangling Gate Model}\label{sec:NEG}
With the possible exception of topological qubits, physical realizations of quantum computers will usually be able to perform arbitrary single-qubit rotations accurately and inexpensively.
Instead, two-qubit interactions will often be more time consuming and/or less accurate.
Thus, counting the number of two-qubit gates needed to implement a protocol is a significant metric of the cost of implementing an algorithm in most NISQ platforms.
We call this computational model the NEG (Noisy Entangling Gate) model for brevity.

In particular, here we take a more specific model than the one-qubit model.
We assume that the user has single qubit gates and measurements that are computationally free and also error free.
They also have access to CNOT channels, $\widetilde{C}_x$, that act between any two qubits (meaning that the qubits are on the complete graph) and if we define $C_x$ to be the ideal CNOT channel then $\|C_x - \widetilde{C}_x\|_{\diamond}\le \delta_g$.
Here the diamond norm is defined to be the maximum trace distance between the outputs of the channels over all input states~\cite{watrous}.
This error is assumed to not be directly controllable and thus places a limit on our ability to accurately simulate the quantum dynamics.

Additionally, because fault tolerance is not assumed to hold, we cannot make the simulation error arbitrarily small in the NEG model.
Therefore, the analysis needed in this model is qualitatively different than that needed in the case of fault tolerance.

\subsubsection{Fault-Tolerant Model}
In this model, we choose to minimize the number of non-Clifford operations, specifically the number of $T$-gates, because they are by far the most expensive gates to implement in most approaches to fault tolerance.
This is because the Eastin-Knill theorem prohibits a universal set of fault tolerant gates to be implemented transversally~\cite{eastin2009restrictions}.
Magic state distillation is commonly used to sidestep this restriction, but it is still by far the most costly aspect of this approach to fault tolerance~\cite{bravyi2012magic}.
For this reason, we will consider the cost of circuits to be the number of $T$-gates they require.

\section{Trotter Step Implementation for Noisy Entangling Gate Model}\label{sec:nisqimplement}
Here we lay out the design primitives that we use to simulate the Schwinger model using quantum computers.
The focus here will be on those elements that would be appropriate to use in a resource constrained setting, wherein the number of qubits available is minimal and two-qubit gates are not assumed to be perfect.
Later we will adapt some of the components to provide better asymptotic scaling, which is important for addressing the intrinsic complexity of simulating Schwinger model dynamics.
\subsection{Implementing (Off-Diagonal) Interaction Terms $T$}\label{sec:kineticimplementRC}
The hopping terms in the Hamiltonian govern the local creation and annihilation of electron-positron pairs.
Rewriting the associated Hamiltonian contribution that acts on each site $1\leq r \leq N-1$,
\begin{align}\label{eq:offdiag}
  T_r = x \left(  \frac{1}{4} (U_r + U_r^\dagger)(X_r X_{r+1} + Y_r Y_{r+1}) + \frac{i}{4} (U_r - U_{r+1}^\dagger)(X_r Y_{r+1} - Y_r X_{r+1})\right)\ \ \ .
\end{align}
\noindent
To fully translate this into qubit operators, we use the linear mapping of electric fields \eq{emap} onto a binary register of $\eta$ qubits: $\ket{-\Lambda} \to \ket{00\cdots0}$, $\ket{-\Lambda+1} \to \ket{00\cdots 1}$, and so on up to $\ket{\Lambda-1} \to \ket{11\cdots 1}$.

As discussed before, if the electric field cutoff is chosen such that the dynamical population of states at $\pm \Lambda$ becomes negligible, then the digitized link operator may be wrapped periodically, simplifying the following discussion\footnote{If the truncation is severe, two $C^{\eta}(X)$ gates per Trotter step can be introduced to remove the interaction between $\pm \Lambda$ field values}.
The combinations of link raising and lowering operators appearing in the Hamiltonian then take the form
\begin{equation} \label{eq:gaugeOps}
    U + U^\dagger = \begin{bmatrix} \ddots& &  &  &  & 1 \\  & 0 &1 &0&0& \\ &1&0&1&0& \\ &0&1&0&1& \\ &0&0&1&0&  \\ 1& & & & & \ddots \end{bmatrix}  \quad \text{and} \quad U - U^\dagger = \begin{bmatrix} \ddots& &  &  &  & 1 \\  & 0 &-1 &0&0& \\ &1&0&-1&0& \\ &0&1&0&-1& \\ &0&0&1&0&  \\ -1& & & & & \ddots \end{bmatrix} \ \ \ .
\end{equation}

As written, the operators $U \pm U^\dagger $ couple all possible nearest-neighbor pairs in the ordered binary basis.
Matters are vastly simplified by splitting these operators into terms that are block-diagonal, with each block only mixing a two-dimensional subspace.
This is possible by the decompositions
\begin{align}
   U + U^\dagger &= A+\tilde{A} \\
   A &:= I\otimes \cdots \otimes I\otimes X \text{ or simply } X_0^b \\
   \tilde{A} &:= \mathcal{S}_E^\dagger(I \otimes \cdots \otimes I \otimes X)\mathcal{S}_E = \mathcal{S}_E^\dagger A \mathcal{S}_E
\end{align}
and
\begin{align}
    i(U - U^\dagger) &= B + \tilde{B} \\
    B &:= I \otimes \cdots \otimes I \otimes Y \text{ or simply } Y_0^b \\
    \tilde{B} &:= \mathcal{S}_E^\dagger(I \otimes \cdots \otimes I \otimes Y)\mathcal{S}_E = \mathcal{S}_E^\dagger B \mathcal{S}_E
\end{align}
Above, $\mathcal{S}_E = \sum_{j=0}^{2^\eta -1} \ket{j+1 \text{ (mod $2^\eta$)}} \bra{j} $ is a unit shift operator on the electric basis, numerically identical to the periodically-wrapped $U$, which transforms $\tilde{A}$ ($\tilde{B}$) into $A$ ($B$).
 The decompositions above therefore show how to map the action of the link operator combinations into incrementers, decrementers and actions on an individual qubit.
 
 To implement the incrementer of the link basis $\mathcal{S}_E$, one can use a quantum Fourier transform and single-qubit rotations in Fourier space, as shown in \fig{shiftcircuit}.
An asymptotically advantageous alternative to this implementation using $\eta-1$ ancillary qubits is presented in \sec{kineticimplementFT}.
\begin{figure}
\centering
\begin{equation}
\includegraphics[]{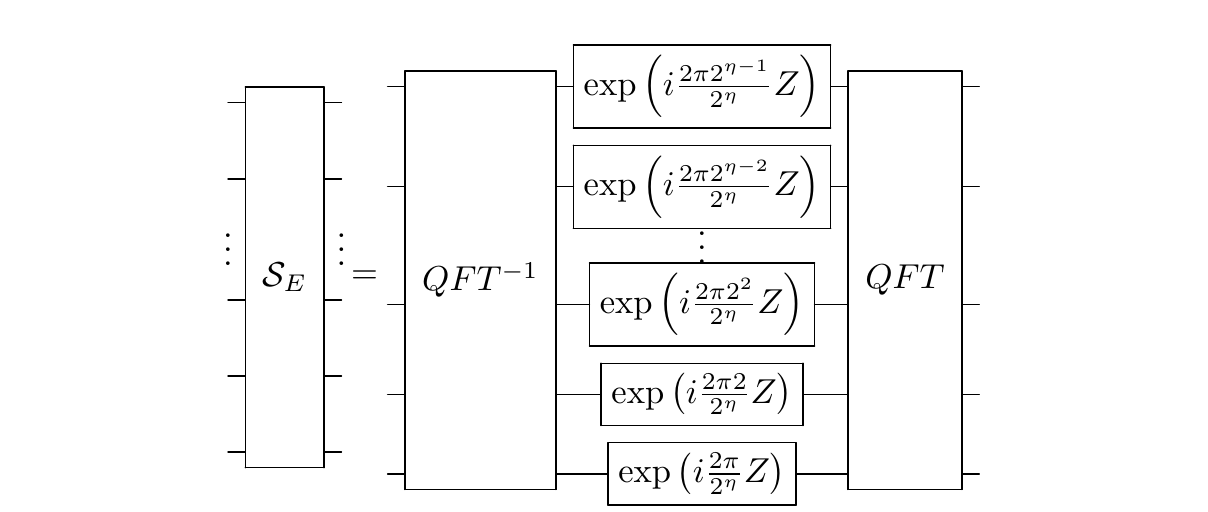}
\end{equation}
\caption{Circuit implementation of the binary increment operator through single-qubit operators in Fourier space.}
\label{fig:shiftcircuit}
\end{figure}

Turning to the factors from the fermionic fields, the pertinent operators are $X_r X_{r+1} + Y_r Y_{r+1}$ and $X_r Y_{r+1} - Y_r X_{r+1}$.
These operators are related by 
\begin{equation}
    X \otimes Y - Y \otimes X = - (S \otimes I) (X\otimes X + Y\otimes Y) (S^\dagger \otimes I),
\end{equation}
with $S$ the ``phase gate,'' $\ket{0}\bra{0}+i\ket{1}\bra{1}$.
To reduce clutter, these composite operators are denoted by
\begin{equation}
  G_r:= X_r X_{r+1} + Y_r Y_{r+1} \quad \text{and} \quad \tilde{G}_r := X_r Y_{r+1} - Y_r X_{r+1}.
\end{equation}

To simulate a hopping term in the Trotter step $V(t)$, we will employ the approximation
\begin{equation}
\label{eq:kinetictrot}
     e^{-i\frac{xt}{8}((A+\tilde{A})\otimes G + (B+\tilde{B})\otimes \tilde{G})} \approx e^{-it T^{(4)}/2}e^{-itT^{(3)}/2}e^{-itT^{(2)}/2}e^{-itT^{(1)}/2},
\end{equation}
where
\begin{align}\label{eq:hopdefs}
    T^{(1)} &:= x(A\otimes G)/4,\\
    T^{(2)} &:= x(\tilde{A}\otimes G)/4, \\
    T^{(3)} &:= x(\tilde{B}\otimes \tilde{G})/4, \\
    T^{(4)} &:= x(B\otimes \tilde{G})/4. \label{eq:hopdeflast}
\end{align}

\begin{figure}
\includegraphics[]{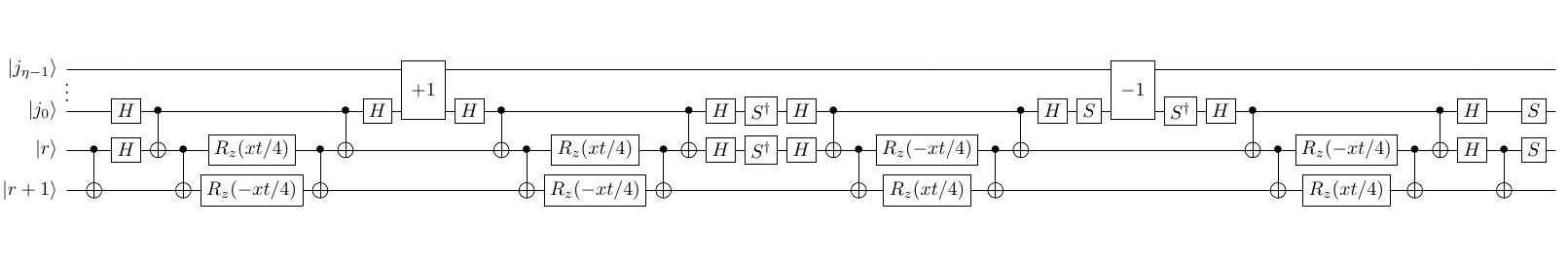}
\caption{A circuit to simulate the Schwinger model hopping terms, $\prod_{j=4}^1 e^{-i T^{(j)} t/2}$, in the order corresponding to \eq{kinetictrot}.
The locality of the presented operator will be expanded to include $\eta$-distance CNOTs between qubits representing fermionic degrees of freedom in quantum registers with one-dimensional connectivity.
The gates labeled $+1$ and $-1$ are the incrementer and decrementer circuits.}
\label{fig:kinetic}
\end{figure}
A circuit representation of the right-hand side of \eq{kinetictrot} is given in \fig{kinetic}.
This routine can be understood in a simple way by first noting the similarity of the four $T^{(i)}$ operators:
\begin{align}
  T_r^{(2)} &= \mathcal{S}_{E,r}^\dagger T_r^{(1)} \mathcal{S}_{E,r} \label{eq:TBG}\\
  T_r^{(3)} &= \mathcal{S}_{E,r}^\dagger ( S^b_{0,r} \ S^f_r ) \left( - T_r^{(1)} \right) ( S^b_{0,r} \ S^f_r )^\dagger \mathcal{S}_{E,r} \label{eq:TBprimeGprime} \\
  T_r^{(4)} &= ( S^b_{0,r} \ S^f_r ) \left(-T^{(1)} \right) ( S^b_{0,r} \ S^f_r )^\dagger  \label{eq:TAprimeGprime}
\end{align}
Consequently, the whole circuit is essentially just four applications of $e^{- i tT^{(1)}/2}$ along with appropriately inserted basis transformations and rotation angle negations.
The specific ordering of the $T^{(i)}$ chosen yields cancellations that reduce the number of internal basis transformations that must be individually executed.
A few single- and two-qubit gates are also spared by additional cancellations.
The remainder of this section addresses the implementation of $e^{\mp it T^{(1)}/2}$.

To effect an application of $e^{- it T^{(1)}/2}$, one can first transform to a basis in which $X \otimes G$ is diagonal.
(Recall $A$ is just $X_0$ -- a bit flip on the last bit of the bosonic register.)
$G$ is diagonalized by the so-called Bell states,
\begin{align}
    \ket{\beta_{ab}} &= \frac{\ket{0\ b} + (-1)^a \ket{1\ \bar b}}{\sqrt{2}} \\
    G \ket{\beta_{ab}} &= 2 b(-1)^a  \ket{\beta_{ab}}
\end{align}
with $\bar b$ indicating the binary negation of $b$, while $X$ is diagonalized by $\ket{\pm} = (\ket{0} \pm \ket{1})/\sqrt{2}$.
From this, we have that
\begin{align}
    e^{-\frac{ixt}{8} X \otimes G} \ket{\pm}\ket{\beta_{00}} & = \ket{\pm}\ket{\beta_{00}} \\
    e^{-\frac{ixt}{8} X \otimes G} \ket{\pm}\ket{\beta_{01}} & = e^{\mp \frac{ixt}{4}}\ket{\pm}\ket{\beta_{01}} \\
    e^{-\frac{ixt}{8} X \otimes G} \ket{\pm}\ket{\beta_{10}} & = \ket{\pm}\ket{\beta_{10}} \\
    e^{-\frac{ixt}{8} X \otimes G} \ket{\pm}\ket{\beta_{11}} & = e^{\pm \frac{ixt}{4}}\ket{\pm}\ket{\beta_{11}}.
\end{align}
Thus, in the Bell basis, we implement rotations conditioned on $a$ and $b$.

The first two time slices of the circuit serve to change to the $X \otimes G$ eigenbasis.
The subsequent parallel $R_z$ rotations flanked by \textsc{CNOT}s implement the controlled rotations in the computational basis, taking
\begin{align}
    \ket{z}\ket{00} & \to \ket{z}\ket{00} ,\\
    \ket{z}\ket{01} & \to e^{(-1)^{\bar z} \frac{ixt}{4}} \ket{z}\ket{01}, \\
    \ket{z}\ket{10} & \to \ket{z}\ket{10}, \\
    \ket{z}\ket{11} & \to e^{(-1)^{z} \frac{ix t}{4}} \ket{z}\ket{11};
\end{align}
this is equivalent to acting with $e^{-\frac{ixt}{4} Z\otimes Z}$.
After undoing the basis transformation, we will have effected $e^{-\frac{ixt}{8}A\otimes G}$.
Three similar operations are executed in the remainder of the circuit;
an incrementer $S_E$ (denoted by ``+1''), the phase gates, and the overall minus sign on the rotations in the latter half of the circuit all stem directly from the relations given in (\ref{eq:TBG},\ref{eq:TBprimeGprime},\ref{eq:TAprimeGprime}).

\begin{figure}
\centering
\begin{equation}
\includegraphics[]{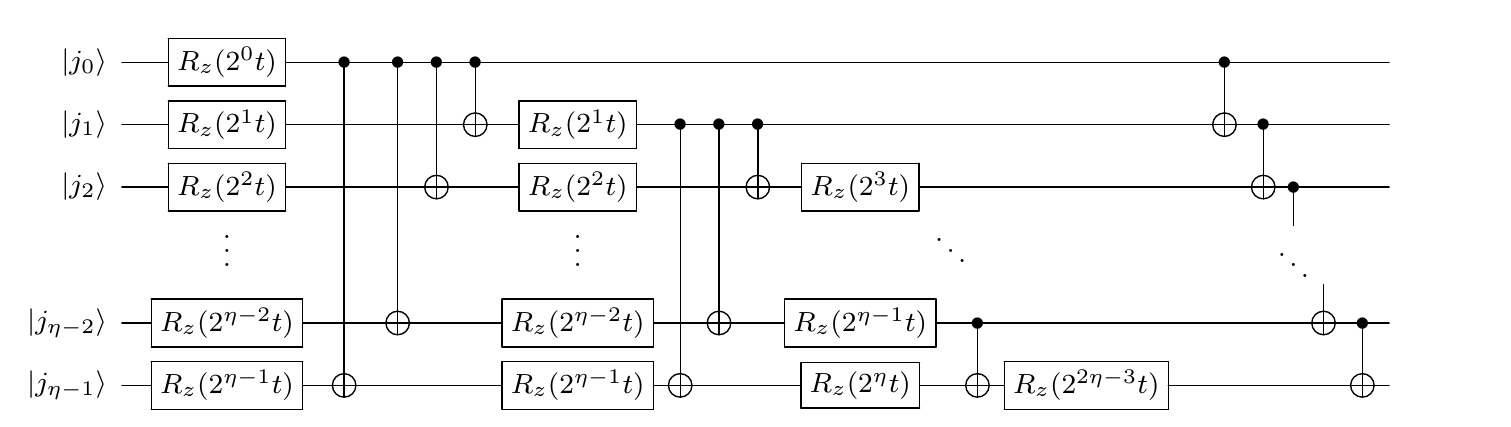}
\end{equation}
\caption{Simplified circuit for simulating $e^{-iE^{2}_{r} t}$ in qubit limited setting.
The circuit is shown acting on the product state $\otimes_{k=0}^{\eta-1} \ket{j_k}$ to clearly mark which qubit each gate is intended to act upon although the circuit is valid for arbitrary inputs.}\label{fig:limittedSquare2}
\end{figure}

The above discussion is summarized below as a lemma for convenience.
\begin{lemma}\label{lem:DrCNOT}
  For any (evolution time) $t\in \mathbb{R}$ the operation
$$e^{-it T^{(4)}/2}e^{-itT^{(3)}/2}e^{-itT^{(2)}/2}e^{-itT^{(1)}/2}$$ can be performed using at most $8+2\eta$ single-qubit rotations, $4$ $\eta$-qubit quantum Fourier transform circuits, $18$ CNOT  gates and no ancillary qubits.
\end{lemma}

\subsection{Implementing (Diagonalized) Mass and Electric Energy Terms ($D$)}\label{sec:diagtermsRC}

\begin{lemma}\label{lem:eneg}
The circuit provided in~\fig{limittedSquare2} implements $e^{-iE^2t}$ on $\eta$ qubits exactly, up to an (efficiently computable) global phase, using $\frac{(\eta+2)(\eta-1)}{2}$ CNOT operations and $\frac{\eta(\eta+1)}{2}$ single-qubit rotations.
\end{lemma}
\begin{proof}
The time evolution associated with the electric energy can be exactly implemented utilizing the structure of the operator.
As defined in (\ref{Eq:E_r_squared}), $E^2={\rm diag}[\Lambda^2,(\Lambda-1)^2,\cdots, 1,0,1,\cdots, (\Lambda-1)^2]$, where $\Lambda$ is the electric field cutoff.
Note that the diagonal elements are not distributed symmetrically---the first diagonal entry is $\Lambda^2$ while the last entry is $(\Lambda-1)^2$.
This lack of symmetry is required to incorporate the gauge configuration with zero electric field.
However, symmetry can be leveraged 
by using the following operator identity:
\begin{equation}\label{eq:klcoIdentity}
    {E}^2 = \left(E+\frac{1}{2}I\right)^2 - \left(E  +\frac{I}{2}\right) + \frac{I}{4}\end{equation}
The operator $E+\frac{1}{2}I = \frac{1}{2}{\rm diag}[-2\Lambda+1,\cdots,-1,1,\cdots,2\Lambda-1]$ is skew persymmetric---containing positive-negative pairs along the diagonal.
We then have from~\eqref{eq:klcoIdentity} and since $[E_r, E_r^2]=0$ that
\begin{align}
    e^{-iE^2 t} &= e^{-i\left(E+\frac{1}{2}I\right)^2 t} e^{i\left(E+\frac{1}{2}I\right) t} e^{-it/4}.
\label{eq:expEr2}
\end{align}
Since unitaries are equivalent in quantum mechanics up to a global phase, we can ignore the last phase in the computation (even if we didn't want to ignore it, it can be efficiently computed as $t$ is a known quantity).

Assuming that each lattice link is given by an $\eta$-qubit register to represent the electric fields ($\Lambda=2^{\eta-1}$), the goal is to find a compact representation of $E+\frac{1}{2}I$  as sums of Pauli operators.
This can be done by performing a decomposition of positive submatrices using the binary representation of integers.
Alternatively, symmetric averages across matrix sub-blocks with dimensions scaling by powers of two may be used to identify 
\begin{equation}
    E+\frac{1}{2}I = -\frac{1}{2}\sum\limits_{j = 0}^{\eta-1} 2^{j}Z_j ,
\end{equation}
where the $(\eta-1)^{\text{th}}$ qubit in the register represents the most significant digit in the binary representation.
Intuitively, the least significant qubit in the binary representation is seen to have a coefficient of $\frac{1}{2}$, necessary to split the unit integer spacing in the value of the electric field of neighboring quantum states.

The leftmost column of $R_z$ operations in~\fig{limittedSquare2}, where
\begin{equation}
  e^{-i \theta Z/2}\equiv \Qcircuit @C=0.5em @R=.5em { &\gate{R_z(\theta )}&\qw  } ,
\end{equation}
can then be seen to implement the linear exponential in \ref{eq:expEr2},
\begin{equation}
    e^{i(E +I/2)t} = \prod_{j=0}^{\eta-1} e^{-i2^{j-1}t Z_j}.
\end{equation}
This requires $\eta$ single-qubit rotations and no CNOT gates.

The quadratic contribution in \ref{eq:expEr2} is given by
\begin{equation}
\left(E + \frac{1}{2} I \right)^2 = \sum\limits_{j = 0}^{\eta-1} \sum\limits_{k=0}^{\eta-1} 2^{j+k-2} Z_{j}Z_{k},
\end{equation}
a sum of entangling operations involving all qubit pairs in the register.
We can manipulate this equation so it is easier to realize its exponential as a quantum circuit.
Pulling out summands proportional to the identity, we see that
\begin{align}\label{eq:esquarecirc}
\left(E + \frac{1}{2} I \right)^2 &= \sum_{j=k}\sum_{k=0}^{\eta-1} 2^{j+k-2} I + \sum\limits_{j = 0}^{\eta-1} \sum\limits_{k\neq j}^{\eta-1} 2^{j+k-2} Z_{j}Z_{k} \nonumber\\
&= \frac{1}{12} (4^{\eta}-1) + 2\sum\limits_{j = 0}^{\eta-2} \sum\limits_{k > j}^{\eta-1} 2^{j+k-2} Z_{j}Z_{k}  \nonumber\\
 &= \frac{1}{12} (4^{\eta}-1) + \sum\limits_{j = 0}^{\eta-2} \sum\limits_{k>j}^{\eta-1} 2^{j+k-1} Z_{j}Z_{k},
\end{align}
where exponentiating the constant term results in a global phase that we can drop.
We give the corresponding quantum circuit in~\fig{limittedSquare}.

Note that the network of $\mathrm{CNOT}_{i,j}$ operations (with control qubit $i$ and target qubit $j$) used in implementing~\eq{esquarecirc} may be described by 
\begin{equation}
    F_{j} := {\rm CNOT}_{j,j+1}{\rm CNOT}_{j,j+2}\cdots {\rm CNOT}_{j,\eta-1}.
\end{equation}
\noindent
We can reduce the number of CNOTs in~\fig{limittedSquare} by using the following identity:
\begin{equation}\label{eq:nisqident}
    F_{j+1} \ F_{j} = {\rm CNOT}_{j,j+1} \ F_{j+1} \ .
\end{equation}
This can be seen by comparing the action of the left- and right-hand side of this equation on the substring $\ket{x_j , x_{j+1} , \ldots, x_{\eta-1}}$ of a computational basis vector $\ket{x}$.
Acting with the left-hand side:
\begin{align}
    F_{j+1} \ F_{j} \ket{x} &= F_{j+1} \ \ket{x_{j}, (x_j \oplus x_{j+1}),(x_{j} \oplus x_{j+2}), \ldots, (x_{j} \oplus x_{\eta-1})} \nonumber \\
    &= \ket{x_{j}, (x_j \oplus x_{j+1}), (x_j \oplus x_{j+1}) \oplus (x_{j} \oplus x_{j+2}), \ldots, (x_j \oplus x_{j+1}) \oplus (x_{j} \oplus x_{\eta-1})} \nonumber \\
    &= \ket{x_{j}, (x_j \oplus x_{j+1}), (x_{j+1} \oplus x_{j+2}), \ldots, (x_{j+1} \oplus x_{\eta-1})} \nonumber \\
    &= {\rm CNOT}_{j,j+1} \ \ket{x_{j}, x_{j+1}, (x_{j+1} \oplus x_{j+2}), \ldots, (x_{j+1} \oplus x_{\eta-1})} \nonumber \\
    &= {\rm CNOT}_{j,j+1} \ F_{j+1} \ \ket{x}.
\end{align}

The right-hand side of \eq{nisqident} requires fewer CNOTs to implement than the left-hand side.
Reducing the circuit of \fig{limittedSquare} to \fig{limittedSquare2} employs an application of \eq{nisqident} $\eta-2$ times between the columns of $z$-rotations, along with a minor CNOT simplification at the end of the circuit.
Counting the CNOTs in \fig{limittedSquare2}, we get

\begin{equation}
    \left(\sum_{j=1}^{\eta-1} j\right) + \eta - 1 = \frac{\eta (\eta-1)}{2} + \eta - 1 = \frac{(\eta +2)(\eta-1)}{2}.
\end{equation}

Similarly, the number of single-qubit rotations needed for the circuit is
\begin{equation}
    \frac{\eta(\eta-1)}{2} + \eta = \frac{\eta(\eta+1)}{2}
\end{equation}

Combining this information, the electric time evolution operator, $e^{-i {E}^2t}$, can be implemented exactly (up to a $t$-dependent global phase) by $\frac{\eta(\eta+1)}{2}$ single-qubit $Z$ rotations and $\frac{(\eta+2)(\eta-1)}{2}$ CNOTs.
\end{proof}
\begin{figure}
\begin{equation}
\includegraphics[]{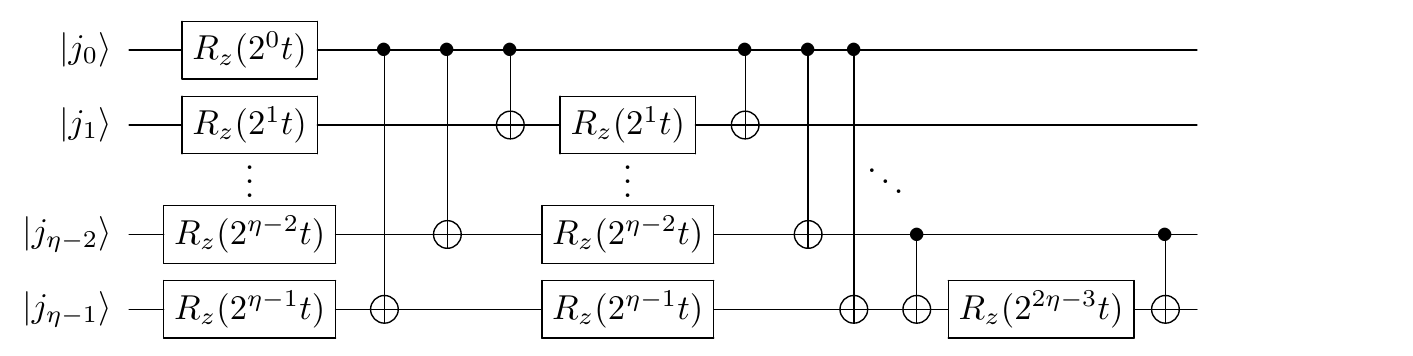}
\end{equation}
\caption{Circuit for simulating $e^{-iE^{2} t}$ in qubit limited setting.
The circuit is shown acting on the product state $\otimes_{k=0}^{\eta-1} \ket{j_k}$ to clearly mark which qubit each gate is intended to act upon although the circuit is valid for arbitrary inputs.}\label{fig:limittedSquare}
\end{figure}

While being conducive to implementation on qubit-limited hardware, this implementation strategy is attractive as a decomposition into mutually-commuting operators---contributing no additional systematic errors to the Trotterized time evolution operator.
It is interesting to note that this set of all two-qubit $Z$ operators has also been found sufficient to implement time evolution of the scalar field mass term~\cite{jordan2011quantum,Jordan_Lee_Preskill_2012,rol2015quantum,Macridin:2018gdw,Klco:2018zqz}.
Similar resource constrained results can be found using singular value transformations or phase arithmetic~\cite{gilyen2019quantum,wiebe2016quantum,low2019hamiltonian} but these approaches require at least one additional qubit.
For fault-tolerent implementation, a quadratically-improved method utilizing arithmetic with ancillary qubits is presented in \sec{diagtermsFT}.

\subsection{Cost to Implement Approximate Time Step in Noisy Entangling Gate Model}

In the following statement, we summarize the cost of implementing a single time step $V(t)$ \eq{timestep}.
Note that we may implement these exactly in the NEG computational model, where single-qubit operations are free.

\begin{lemma}[Schwinger Time Step Cost in NEG Model]
\label{lem:nisqtimestep}
Consider any $t \in \mathbb{R}$.
The unitary operation $V(t)$ as defined in \eq{timestep} may be implemented on a quantum computer in the Noisy Entangling Gate model using a number of CNOTs that is at most
$$(N-1)(9\eta^2 -7\eta +34).$$
\end{lemma}
\begin{proof}
Proof follows by considering the symmetric Trotter-Suzuki expansion of the time-evolution operator.
In particular, we have that $e^{-iHt} = V(t)+O(t^3)$, where $V$ is defined as in \eq{timestep} and restated here for convenience:
\begin{equation}
    V(t) = \prod_{r=1}^{N-1} \left ( \prod_{k\in \{M,E\}} e^{-i D_r^{(k)} t/2}\prod_{j=1}^4e^{-i T_r^{(j)} t/2}  \right )e^{-i D_N^{(M)} t} \prod_{r=N-1}^1  \left ( \prod_{j=4}^1 e^{-i T_r^{(j)} t/2}\prod_{k\in \{E,M\}} e^{-i D_r^{(k)} t/2}\right ). \nonumber
\end{equation}
The proof now proceeds by a counting of the number of gates needed to implement $V$.

The mass terms, given in $D_r^{(M)}$, are single-qubit operations and require no CNOTs to implement.
They are therefore free within the cost model considered here.

The hopping terms are implemented according to the discussion of \sec{kineticimplementRC} and their total cost is determined as follows: 18 explicit CNOTs appear in \fig{kinetic} and the rest are embedded in the two shifters.
Each shifter can be implemented as in \fig{shiftcircuit}, namely using two quantum Fourier transforms and single-qubit rotations.
A single quantum Fourier transform can be done using Kitaev's approach~\cite{nielsen2002quantum} in $\sum_{k=1}^{\eta-1} (\eta-k)=\eta(\eta-1)/2$ controlled $Z$-rotations, each implementable using two CNOTs.
Putting this all together gives the total $4\eta(\eta-1)+18$ CNOTs required to implement $\prod_{j=1}^4e^{-i T_r^{(j)} T/(2s)}$.
This term appears $2(N-1)$ times in the Trotter-suzuki expansion and so the total number of CNOTs due to $H_I$, per second-order Trotter step, is $(N-1)(8\eta(\eta-1)+36)$.

The cost of the electric terms, given in $D_r^{(E)}$, follows from the fact that the procedure given in \lem{eneg} can be implemented with no more than $(\eta+2)(\eta-1)/2$ CNOTS.
Since we use the second-order Trotter formula, each Trotter step consists of $2(N-1)$ such evolutions and the total cost is at most $(N-1)((\eta+2)(\eta-1))$ CNOT gates.
\end{proof}

\section{Trotter Step Implementation in Fault-Tolerant Model}\label{sec:circuiterr}

In this section, we adapt the above implementations for fault-tolerant architectures, recalling that we wish to minimize the number of $T$-gates.
Due to the fact that singe-qubit rotations are not free under the computational model of this section, we will have to implement them approximately in general.
We denote our circuit construction of a time step $V(t)$ by $\widetilde V(t)$, where we introduce the tilde to emphasize the fact that there are approximation errors that arise due to circuit synthesis.

\subsection{Implementing (Off-Diagonal) Interaction Terms $T$}\label{sec:kineticimplementFT}

First, we consider the hopping interactions.

In the scenario of fault-tolerance, the incrementers used in \fig{kinetic} would be an optimization of one that trades $T$-gates for more ancillas.
An example of such an incrementer, for $\eta=3$, is depicted in \fig{increment}.
Cascading the two-Toffoli construction between ancillas gives the generalization.
Using a computation and uncomputation trick for the logical ANDs \cite{gidney2018halving}, the total cost for each incrementer is $\eta-2$ Toffolis and $\eta-1$ ancilla.

\begin{figure}
\[
\includegraphics[]{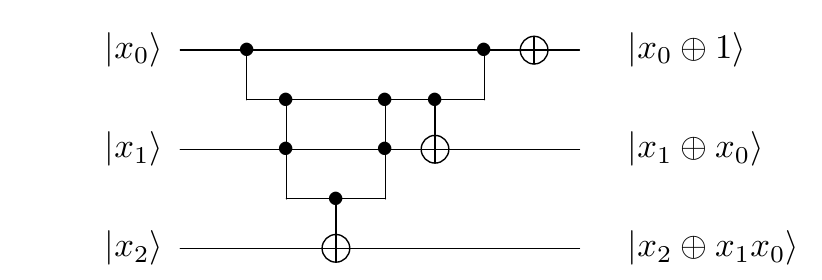}
\]
\caption{A circuit for incrementing the computational basis.
The wire corners denote logical ANDs with ancilla qubits, which are initialized at $\ket{0}$.}
\label{fig:increment}
\end{figure}

The $R_z$ rotations may be done using the Repeat-Until-Success method discussed in~\cite{bocharov2015efficient}, in which one repeatedly tries to implement the rotation with some fixed success probability until it has succeeded.
Using one ancilla and measurement, this method implements the one-qubit rotation to $\epsilon$ precision in the Frobenius norm (scaled by 1/2) such that for the approximation $R'$ of rotation $R$,
\begin{equation}
    \sqrt{1-\frac{|\mathrm{Tr}(R^\dagger R')|}{2}}= \frac{1}{2} ||R-R'||_F < \epsilon,
\end{equation}
and the expected number of $T$-gates is $1.15 \log(1/\epsilon)$.
Since $||R-R'||_S \leq ||R-R'||_F$, choosing a precision of $\delta = 2\epsilon$, the expected $T$ count is $1.15 \log(2/\delta)$ in terms of the spectral norm error.

Combining all this, we get the following result.

\begin{lemma}
\label{lem:kinetic}
Let $T^{(j)}$ be defined as in (\ref{eq:hopdefs}-\ref{eq:hopdeflast}).
For any $t\in \mathbb{R}$ (evolution time) and $\delta>0$ (circuit synthesis tolerance), there exists a unitary operation $\widetilde V_I(t)$ which can be implemented on a fault-tolerant quantum computer such that $\|\widetilde V_I(t) -\prod_{j=1}^4 e^{-itT^{(j)} }\|\le \delta$, where $\widetilde V_I(t)$ has an expected $T$-gate count of $8( \log \Lambda - 1 )+9.2 \log (16/\delta)$ and requires one ancilla.
\end{lemma}
\begin{proof}
\fig{kinetic} gives the existence, which is proven to implement $\prod_{j=1}^4 e^{-i t T^{(j)} }$ by the discussion of \sec{kineticimplementRC}.
As for its cost, note that if the construction in \fig{increment} is used then the two incrementers contribute $2(\eta-2)$ Toffolis with $\eta = \log( 2\Lambda)$.
Toffolis may be implemented using an additional ancilla qubit and four $T$-gates each~\cite{jones2013low}.
However, the single ancilla may be reused to implement all the Toffolis.
Thus the incrementers contribute $8(\eta-2)$ $T$-gates and $1$ ancilla barring any simplifications.

The only other contributions are the eight $z$-rotations, each performed to accuracy at least $\delta/8$, which from Box 4.1 of~\cite{nielsen2002quantum} guarantees a cumulative error of at most $\delta$ from synthesizing the rotations.
Thus their expected T-count is at most $9.2 \log(16/\delta)$.
Summing the two contributions in terms of $T$-gates gives the result.
\end{proof}

\subsection{Implementing (Diagonalized) Mass and Electric Energy Terms ($D$)}\label{sec:diagtermsFT}

Next, we turn to simulating the non-interaction terms in the Hamiltonian.

The easiest part of the second-order Trotterized time step $V(t)$ to implement  is the mass term propagators $D^{(M)}$, which are single-qubit $z$-rotations.

\begin{lemma}
\label{lem:zrots}
Let $D^{(M)}$ be defined as in \eq{terms}.
For any $t\in \mathbb{R}$ (evolution time) and $\delta>0$ (circuit synthesis tolerance), there exists a unitary operation $\widetilde V_{M}(t)$ which can be implemented on a quantum computer such that  $\|\widetilde V_{M}(t)-e^{-i D^{(M)} t}\|\le \delta$, which has an expected $T$-gate count of $1.15 \log (2/\delta)$ and requires one ancilla.
\end{lemma}
\begin{proof}
The exponential of $D^{(M)}$ is simply a $z$ rotation.
The cited complexity then follows from the repeat-until-success results given in \cite{bocharov2015efficient}.
\end{proof}

The most direct way to implement the electric propagators $e^{-i {D^{(E)}}t}=e^{-i{E}^2t}$ involves using a multiplication circuit to evaluate the value of $E^2$ on a given link and subsequently generating phase changes proportional to $E^2$.
Here we present an elementary implementation of a multiplication circuit which is based on the logarithmic-depth ripple-carry adder of~\cite{draper2004logarithmic}.
There are a host of multiplication algorithms that can be considered.
However, for simplicity, we will use ``grammar school'' multiplication here.
The circuit size for this method is expected to be in $O(\eta^2)$ for $\eta$ bits.
More advanced methods, based on Fourier-Transforms such as the Sch\"{o}nage-Strassen multiplication algorithm or its refinement in the form of F\"{u}rer's algorithm, can be used to solve the problem using a number of gates in $\widetilde{O}(\eta)$.
While these methods are asymptotically more efficient, they often prove to be less practical for modest-sized problems and, furthermore, have space complexity that is harder to bound than that of the comparably simpler grammar school multiplication algorithm.

We give below an implementation of a squaring algorithm based on this strategy and bound the number of non-Clifford operations as well as the number of ancillary qubits needed to compute the function out of place.
\begin{lemma}[Squaring Cost in Fault-Tolerant Model]
\label{lem:mult}
For any positive integer $\eta$ (link register size) there exists a positive integer $\alpha$ (number of ancillas) and unitary operation $U\in \mathbb{C}^{2^{\eta+\alpha} \times 2^{\eta+\alpha}}$ such that for all computational basis vectors $\ket{x} \in \mathbb{C}^{2^\eta}$, $U:\ket{x} \ket{0}^{\otimes \alpha} \mapsto \ket{x}\ket{x^2} \ket{0}^{\alpha-2\eta}$ where $\alpha \le 5\eta - \lfloor \log \eta \rfloor -1$, 
using a number of $T$-gates that is bounded above by
$$  4(\eta-1)(12\eta - 3\lfloor\log \eta \rfloor - 14).$$
\end{lemma}
\begin{proof}
We choose to base our multiplier on a logarithmic-depth adder proposed by Draper et al.\, in~\cite{draper2004logarithmic}. We use an in-place version of the adder that adds two $\eta$-bit integers using at most $10\eta -3\lfloor \log \eta \rfloor -13$ Toffoli gates by performing
\begin{equation}
  \ket{a} \ket{b} \ket{0}^{\otimes \beta} \mapsto \ket{a+b} \ket{b}   \ket{0}^{\otimes (\beta-1)},
\end{equation}
where $\ket{a+b}$ is $\eta+1$ bits and $\ket{0}^{\otimes \beta}$ consists of at most $2\eta - \lfloor \log \eta \rfloor -2$ ancillary qubits.

To initialize the grammar school multiplication algorithm, we copy the input $\ket{x}$ to an additional $\eta$-sized register conditioned on the zeroth bit of $x$.  This requires at most $(\eta-1)$ Toffoli gates and acts as
\begin{equation}\label{eq:squaremap}
  \ket{x}\ket{0}^{\otimes \eta} \ket{0}^{\otimes \eta} \ket{0}^{\otimes (\alpha-2\eta)} \mapsto \ket{x}\ket{x\cdot x_0} \ket{0}^{\otimes \eta} \ket{0}^{\otimes (\alpha-2\eta)}
\end{equation}
Additionally, append a blank ancilla as the most significant bit of the second register:
\begin{equation}
  \ket{x}\ket{x\cdot x_0} \ket{0}^{\otimes \eta} \ket{0}^{\otimes(\alpha-2\eta)} \mapsto \ket{x}(\ket{0}\ket{x\cdot x_0}) \ket{0}^{\otimes \eta} \ket{0}^{\otimes(\alpha-2\eta-1)}
\end{equation}
Next, the grammar school multiplication algorithm is done in rounds proceeding initialization. On the $j$th round ($j=1,2,\ldots,\eta-1$), we perform the sequence of steps:
\begin{enumerate}
    \item Copy input $\ket{x}$ to the third register conditioned on the $j$th bit of $x$. Requires at most $\eta-1$ Toffoli gates.
    \item Apply in-place addition algorithm to add the value in the third register (as ordered above) to the $\eta$ most-significant bits of the second register (i.e., the bits from the $2^j$'s place through the $2^{\eta -1 +j}$'s place). Requires at most $10\eta -3\lfloor \log \eta \rfloor -13$ Toffoli gates.
    \item Reset third register to $\ket{0}^\eta$. Requires at most $\eta-1$ Toffoli gates.
\end{enumerate}
To illustrate the effect of a round, the $j$th round takes
\begin{equation}
  \ket{x}\Ket{\sum_{k=0}^{j-1} (x\cdot x_k) 2^k} \ket{0}^{\otimes \eta} \ket{0}^{\otimes(\alpha-2\eta -j)} \mapsto \ket{x}\Ket{\sum_{k=0}^{j} (x\cdot x_k) 2^k} \ket{0}^{\otimes \eta} \ket{0}^{\otimes(\alpha-2\eta-j-1)}.
\end{equation}
Upon completing $\eta-1$ rounds, we will finish with
\begin{equation}
  \ket{x}\Ket{\sum_{k=0}^{(\eta-1)} (x\cdot x_k) 2^k} \ket{0}^{\otimes \eta} \ket{0}^{\otimes(\alpha-3\eta)} = \ket{x}\ket{x^2} \ket{0}^{\otimes \eta} \ket{0}^{\otimes(\alpha-3\eta)}.
\end{equation}
Counting the cost for the entire algorithm including initialization, we require
\begin{equation}
    (\eta-1) + (\eta-1)(10\eta - 3\lfloor\log \eta \rfloor - 13 + 2(\eta - 1)) = (\eta-1)(12\eta - 3\lfloor\log \eta \rfloor - 14)
\end{equation}
Toffoli gates. 

As used prior, Toffolis may be implemented using an additional ancilla qubit and four $T$-gates each~\cite{jones2013low}, implying a factor of 4 when converting this bound to $T$-gates. The single additional ancilla may be reused for every implementation. 

To conclude, the total number of ancillary qubits needed for this algorithm implies that
\begin{equation}
    \alpha \leq 2\eta + \eta + (2\eta - \lfloor \log \eta \rfloor -2) + 1 = 5\eta - \lfloor \log \eta \rfloor -1.
\end{equation}
\end{proof}

This multiplier can also be implemented in polylogarithmic depth by noting that each adder can be implemented in $O( \log \eta)$ depth and, by using a fan-out network\footnote{A fan-out network is a sequence of CNOT gates that copies the computational basis inputs according to a binary tree which has $2^d$ registers but requires a CNOT depth of $O(d)$.}, the results can be combined in depth $O( \log \eta)$ resulting in a depth of $O( \log^2 \eta)$.

Using the above squaring algorithm, we may implement the complex exponential of the electric energy operator.

\begin{corollary}[Electric Propagator Cost]
  \label{cor:square}
Let $E$ be a Hermitian operator in $\mathbb{C}^{2^\eta\times 2^\eta}$ such that $E \ket{j} = (j-2^{\eta-1})\ket{j}$.
For all $t\in \mathbb{R}$ (propagation time) and $\delta>0$ (circuit synthesis tolerance), there exists a positive integer $\alpha$ and a unitary operation $\widetilde V_{E}(t) \in \mathbb{C}^{2^{ \eta+\alpha} \times 2^{\eta +\alpha}}$ such that $\|\widetilde V_{E}(t)-e^{-i t {E}^2}\| \le \delta$, which can be implemented on a quantum computer using on average at most $$4.45 \eta\log( 3\eta/\delta) + 8(\eta-1)(12\eta - 3\lfloor\log \eta \rfloor - 14)$$ $T$-gates and a number of ancillas $$\alpha \le 5\eta - \lfloor \log \eta \rfloor -1. $$
\end{corollary}
\begin{proof}
To begin note that
\begin{equation}
    E^2 \ket{j} = ( j^2 - 2^\eta j + 2^{2 \eta-2}) \ket{j} \ .
\end{equation}
Therefore, up to a ($t$-dependent) phase that can be ignored, we need to implement
\begin{equation}
\ket{j} \mapsto e^{-i j^2 t} e^{i2^\eta jt}\ket{j}\label{eq:elltrans}
\end{equation}
The phase linear in $j$ can be acquired (up to an unimportant global phase) by applying $e^{-i2^\eta 2^m t Z / 2}$ to each qubit $m=0,1,\ldots ,(\eta-1)$ of $\ket{j}$.

The phase quadratic in $j$ can be similarly obtained by first computing $j^2$ into an ancilla register as described in \lem{mult}, then applying $e^{i t 2^m Z/2}$ to each qubit $m=0,1,\ldots , (2\eta-1)$ of $\ket{j^2}$, and finally uncomputing $j^2$ from the ancilla register.

As there are $2 \eta$ qubits used to store the value of $j^2$ and only $\eta$ bits needed to store $j$,~\eqref{eq:elltrans} implies that $3\eta$ single-qubit rotations are needed to perform the unitary transformation exactly.

The final step is to consider the cost of performing these rotations using the $H,T$ gate library.
Using Box 4.1 from~\cite{nielsen2002quantum} we have that it suffices to implement each rotation within error $\delta/ 3 \eta$ to ensure that the overall error in the circuit is at most $\delta$.
It then follows from~\cite{bocharov2015efficient} and the fact that the number of synthesis attempts made by each invocation of repeat-until success synthesis are independent random variables that the mean $T$ count of implementing all the $3\eta$ rotations is at most $4.45 \eta \log(3\eta /\delta)$ since the mean number of gates needed to synthesize a single-qubit rotation within operator error $\epsilon>0$ is at most $1.15\log(1/\epsilon)$.
The total $T$ cost is the sum of this cost and the cost of applying~\lem{mult} twice. The ancilla count is given by that of the squaring circuit, since it is the circuit element that requires the most ancilla, and its ancillas can be reused to perform the z rotations.

\end{proof}

By this point, we have given asymptotically-advantageous implementations of every constituent of the second-order Trotter approximation $V(t)$ of \eq{timestep}.

\subsection{Cost to Implement Approximate Time Step in Fault-Tolerant Setting}\label{sec:timecost}

We combine all the relevant lemmas above to determine the resource scaling for a single approximate time step.
For convenience, we state the upper bound on the cost after factoring out a sum of terms that show the possible asymptotic behaviors.
We use this format here and in subsequent results to emphasize possible asymptotic scalings.
\begin{theorem}[Schwinger Time Step Cost in Fault-Tolerant Setting]
\label{thm:ucost}
For any $t \in \mathbb{R}$ (propagation time) and $\delta_{\rm circ}>0$ (synthesis error) there exists a unitary operation $\widetilde V(t)$ which can be implemented on a quantum computer such that  $\|\widetilde V(t)-V(t)\|\le \delta_{\rm circ}$, where $V(t)$ (the second-order Trotterized time step) is defined in \eq{timestep}. Furthermore, $\widetilde V(t)$ requires on average a number of $T$-gates that is at most

$$ \left( N\eta^2 + N \eta \ln\left(\frac{6N-5}{\delta_{\rm circ}} \right) \right) \lambda(\delta_{\rm circ}),$$
where
\begin{align*} 
\lambda(\delta_{\rm circ}) &= \biggl( 2(N-1) \bigl(96\eta^2 + 24(1-\eta) \lfloor \log \eta \rfloor + 4.45 \eta \log (3\eta) +(10.35+4.45\eta) \log \left( \frac{6N-5}{\delta} \right) \\
&\quad -200\eta +133.95 \bigr ) + 1.15\log \left ( \frac{2(6N-5)}{\delta} \right) \biggr) \bigg / \left( N\eta^2 + N \eta \ln\left(\frac{6N-5}{\delta_{\rm circ}} \right) \right)
\end{align*}
and a number of qubits $$ N(\eta + 1) +4\eta - \lfloor \log \eta \rfloor -1. $$

\end{theorem}
\begin{proof}
As $V(t)$ is defined, we may implement it using $2(N-1)$ copies of each $\widetilde V_I$, $\widetilde V_{D^{(M)}}$, and $\widetilde V_{D^{(E)}}$, along with one additional $\widetilde V_{D^{(M)}}$ for site $N$.
The costs of these are given by \lem{kinetic}, \lem{zrots}, and \cor{square}.
Furthermore, by choosing each term to contribute at most $\delta_{\rm circ}/(6N-5)$ error, the total error will be bounded by $\delta_{\rm circ}$.
The result follows by plugging the requisite error into each lemma, summing, multiplying by $2(N-1)$, and finally adding the contribution of the extra $\widetilde V_{D^{(M)}}$.
We state the result after factoring out the leading order terms.
The number of ancillas is given by \cor{square}, as it is the circuit that requires the largest number of ancilla.
Adding the number of qubits needed to represent the fermionic and gauge fields, $\eta(N-1)+N$, we arrive at the claimed upper bound on the total number of qubits required for simulation.
\end{proof}

We put these results to use in the following section.

\section{Cost of Trotterized Time Step ($e^{-iHT}$)}\label{sec:faulttolerant}

In this section, we use the circuits given in \sec{nisqimplement} and \sec{circuiterr} to analyze implementation of the time evolution operator ($e^{-iHT}$) in the Fault-Tolerant and NEG settings. We emphasize that because we provide loose upper bounds, the results contained here give sufficient (but not necessary) conditions for simulation.

For convenience, we touch on the main points of our analysis with the following asymptotic scaling, an immediate consequence of the upper bounds we establish later on.
\begin{corollary}[Trotterized Time Evolution Cost in Fault-Tolerant Model]
\label{cor:stepsAss}
Consider any $\delta>0$ (Trotter error), $T\in \mathbb{R}$ (total evolution time), and let $V(t)$ be the second-order Trotter-Suzuki decomposition of $e^{-iHt}$ as defined as in \eq{timestep}. Additionally, let $\mu$ be a constant and $x$ be lower bounded by a constant (i.e. $x\in \Omega(1)$). 

Under these assumptions, there exists an $s\in \N$ (Trotter steps) such that $\| V(T/s)^s - e^{-iHT} \| \leq \delta$, where $V(T/s)^s$ consists of $N_{\exp}$ matrix exponentials and
\begin{equation}
    \label{eq:sboundAss}
N_{\exp} \in O \left(  \frac{ N^{3/2}T^{3/2}  \Lambda x^{1/2}}{\delta^{1/2}}\right) .
\end{equation}
Furthermore, there exists an implementation of the Trotter decomposition within spectral-norm error $\delta$ of $e^{-iHT}$ that requires a number of $T$-gates in
\begin{equation}
    \widetilde{O} \left(  \frac{ N^{3/2}T^{3/2}  \Lambda x^{1/2}}{\delta^{1/2}}\right),
\end{equation}
where $\widetilde{O}(\cdot)$ is equivalent to $O(\cdot)$ but with all non-dominant sub-polynomial factors in the scaling suppressed.
\end{corollary}

Proof of this corollary follows after the next section.

\subsection{Costing $e^{-iHT}$ Simulation in the Fault-Tolerant Model}\label{sec:troterr}

With a second-order approximation, we exploit the fact that many of the Hamiltonian summands commute and calculate an error bound through \eq{2ndordertrot}, restated as,
\begin{align}\label{eq:seconderror}
    \delta_{\rm Trot} &= \bigl\| V(t) - e^{i H t} \bigr\| \leq \frac{1}{12}\sum_{x, \ y>x} \|[[H_x,H_y],H_x]\|t^3 + \frac{1}{24}\sum_{x, \ y>x,\ z>x} \|[[H_x,H_y],H_z]\|  t^3 .
\end{align}

\noindent
We order the summands of $H$ as follows, recalling that the decomposition of $D_r$ into $D_r^{(M)}$ and $D_r^{(E)}$ contributes no error:
\begin{equation}
\label{eq:sumorder}
\{H_x\}_{x=1}^{5N-4} = \{
    D_1, T_1^{(1)},T_1^{(2)},T_1^{(3)},T_1^{(4)}, D_2, \ldots,T_{N-1}^{(1)},T_{N-1}^{(2)},T_{N-1}^{(3)},T_{N-1}^{(4)}, D_N \}
\end{equation}
This ordering gives some immediate simplifications, since it is clear that $[D_i,D_j] = [D_i, T_j^{(k)}]=0$ when $i<j$ and $[T_i^{(k)}, T_j^{(l)}] = [T_i^{(k)}, D_j]=0$ when $i+1<j$.
In other words, operations with at least one lattice site between them commute.
We calculate the right-hand side of \eq{seconderror} in \append{commutators} according to these simplifications, giving the bound of \cor{errboundcalced}, restated here:
\begin{equation}\label{eq:trotsteperr}
    \delta_{\rm Trot} \leq Nt^3 \biggl( \frac{2}{3}x \Lambda^2 + \biggl(2x^2 + \frac{5}{6}x\mu + \frac{2}{3}x \biggr) \Lambda + \frac{39}{8}x^3 + \frac{25}{12}x^2 \mu + x^2 + \frac{1}{3}x \mu^2 + \frac{5}{12}x\mu + \frac{1}{6}x \biggr).
\end{equation}
We recall the definitions of the parameters of the simulation, that $N$ is the number of lattice sites, $\Lambda$ is the electric cutoff, $t$ is the time step length, $x = 1/(ag)^2$ (where $a$ is the lattice spacing and $g$ is the coupling constant), and $\mu = 2m/ag^2$ with $m$ the mass of the particles on the lattice.

Ultimately, we have the following lemma.
\begin{lemma}[Trotter Step Count with Second-Order Product Formula]
\label{lem:steps}
Consider any $\delta>0$ (total Trotter error) and  $T\in \mathbb{R}$ (total evolution time), and let $V(t)$ be the second-order Trotter-Suzuki approximation to $e^{-iHt}$ as defined as in \eq{timestep}.
Then we have that $\| V(T/s)^s - e^{-iHT} \| \leq \delta$ provided that
\begin{equation}
    \label{eq:sbound}
s = \left\lfloor  \frac{N^{1/2} T^{3/2} \Lambda x^{1/2} \rho(\delta)}{\delta^{1/2}}\right\rfloor\ ,
\end{equation}
where
\begin{align}
    \rho(\delta) &= \frac{\delta^{1/2}}{N^{1/2} T^{3/2} \Lambda x^{1/2}} \\
    &\quad + \frac{\sqrt{ \frac{2}{3}x \Lambda^2 + \biggl(2x^2 + \frac{5}{6}x\mu + \frac{2}{3}x \biggr) \Lambda + \frac{39}{8}x^3 + \frac{25}{12}x^2 \mu + x^2 + \frac{1}{3}x \mu^2 + \frac{5}{12}x\mu + \frac{1}{6}x }}{\Lambda x^{1/2}}.
\end{align}
\end{lemma}
\begin{proof}
From \eq{nielsenChuang4.1} we have $\| V(T/s)^s - e^{-iHT} \| \leq s \delta_{\rm Trot}$,  where $\delta_{\rm Trot} = \| V(T/s) - e^{-iHT/s} \| $, so we seek the minimum $s$ such that $s\delta_{\rm Trot}\leq\delta$.
\cor{errboundcalced} gives a bound on $\delta_{\rm Trot}$, which we may use to choose $s$.
For convenience, let
\begin{equation}
    \chi := N \biggl( \frac{2}{3}x \Lambda^2 + \biggl(2x^2 + \frac{5}{6}x\mu + \frac{2}{3}x \biggr) \Lambda + \frac{39}{8}x^3 + \frac{25}{12}x^2 \mu + x^2 + \frac{1}{3}x \mu^2 + \frac{5}{12}x\mu + \frac{1}{6}x \biggr),
\end{equation}
so that
\begin{align}
    s \delta_{\rm Trot} &\leq s \left( \frac{T}{s} \right)^3  \chi = \frac{T^3\chi}{s^2}
\end{align}
implying that
\begin{equation}
    s \geq \frac{ T^{3/2} \chi^{1/2} }{\delta^{1/2}}
\end{equation}
to guarantee $\| V^s(T/s) - e^{-iHT} \|\leq \delta$.
Since the number of time steps needs to be an integer, it suffices to choose
\begin{align}
    s &= \left\lceil\frac{ T^{3/2} \chi^{1/2} }{\delta^{1/2}}\right\rceil\nonumber\\
    &=\left \lfloor \frac{N^{1/2} T^{3/2} \Lambda x^{1/2}}{\delta^{1/2}}\left( \frac{\delta^{1/2}}{N^{1/2} T^{3/2} \Lambda x^{1/2}} +\frac{\chi^{1/2}}{ N^{1/2} \Lambda x^{1/2}}\right) \right \rfloor \nonumber \\
    &= \left\lfloor  \frac{N^{1/2} T^{3/2} \Lambda x^{1/2} \rho(\delta)}{\delta^{1/2}}\right\rfloor.
\end{align}
\end{proof}

Combining this result with the cost to implement a single Trotter step, we can determine the number of $T$-gates required to construct $e^{-iHT}$ to arbitrary precision $\delta$.
Note that below we evenly distribute the errors due to circuit synthesis and due to Trotterization.
However, when we later incorporate measurement we will allow for different allocations of the errors in order to numerically optimize the $T$-gate count with respect to the error distribution.

\begin{corollary}[Upper Bound on Cost of Trotterized Time Evolution in Fault-Tolerant Model]\label{cor:faulttimeevolve}
Consider any $\delta>0$ (synthesis and Trotter error) and $T\in \mathbb{R}$ (total evolution time).
There exists an operation $W(T)$ that may be implemented on a quantum computer such that $\|W(T)-e^{-iHT}\| \leq \delta$, using on average at most
\begin{align*}
&\quad \frac{N^{3/2}T^{3/2}\Lambda \eta x^{1/2}}{(\delta/2)^{1/2}}\ln \left(\frac{2^{3/2}(6N-5)N^{1/2} T^{3/2}  \Lambda x^{1/2} \rho(\delta/2)}{\delta^{3/2}}\right) \nonumber \\
    &\quad \quad \quad \times \gamma'(\delta) \rho(\delta/2)\lambda \left(\frac{\delta^{3/2}}{2^{3/2}N^{1/2} T^{3/2}  \Lambda x^{1/2} \rho(\delta/2)} \right),
 \end{align*}
$T$-gates, where $\lambda$ and $\rho$ are as given in \thm{ucost} and \lem{steps},
\begin{equation}
    \gamma'(\delta):= \frac{\left(\eta + \ln \left(\frac{2^{3/2}(6N-5)N^{1/2} T^{3/2}  \Lambda x^{1/2} \rho(\delta/2)}{\delta^{3/2}}\right)\right)}{ \ln \left(\frac{2^{3/2}(6N-5)N^{1/2} T^{3/2}  \Lambda x^{1/2} \rho(\delta/2)}{\delta^{3/2}}\right)} \ ,
\end{equation}
and using at most $ N(\eta + 1) +4\eta - \lfloor \log \eta \rfloor -1$ qubits.
\end{corollary}
\begin{proof}
  Choose $s$ according to \lem{steps} such that $\|V(T/s)^s - e^{-iHT}\| \leq \delta/2$, i.e.
$$s = \left\lfloor  \frac{N^{1/2} T^{3/2} \Lambda x^{1/2} \rho(\delta/2)}{(\delta/2)^{1/2}}\right\rfloor.$$

Choose $\widetilde{V}(t)$ according to \thm{ucost}, such that $\|\widetilde{V}(T/s) - V(T/s) \| \leq \delta/(2s)$.
Here, we are evenly distributing the circuit synthesis and Trotter errors.
Letting $W(T) = \widetilde{V}(T/s)^s$,
\begin{equation}
    \|W(T) - e^{-iHT} \| \leq \|\widetilde V(T/s)^{s} - V(T/s)^s \| + \|V(T/s)^s - e^{-iHT}\| \leq s\|\widetilde{V}(T/s) - V(T/s) \| + \delta/2 \leq \delta.
\end{equation}

The $T$-gate cost of implementing $W(T)$ is given by $s C_{\mathrm{Trot}}$, where $C_{\mathrm{Trot}}$ is the $T$-gate cost of implementing a single Trotter step.
This is
\begin{align}
    sC_{\mathrm{Trot}} &\leq \frac{N^{1/2} T^{3/2} \Lambda x^{1/2} }{(\delta/2)^{1/2}}\left(N\eta^2 + N\eta \ln \left(\frac{6N-5}{\delta/2s}\right)\right) \lambda(\delta/2s) \rho(\delta/2) \\
    &=  \frac{N^{3/2}T^{3/2}\Lambda \eta x^{1/2}}{(\delta/2)^{1/2}}\ln \left(\frac{2^{3/2}(6N-5)N^{1/2} T^{3/2}  \Lambda x^{1/2} \rho(\delta/2)}{\delta^{3/2}}\right) \nonumber \\
    &\quad \times \gamma'(\delta) \rho(\delta/2)\lambda \left(\frac{\delta^{3/2}}{2^{3/2}N^{1/2} T^{3/2}  \Lambda x^{1/2} \rho(\delta/2)} \right),
\end{align}
where
\begin{equation}
    \gamma'(\delta):= \frac{\left(\eta + \ln \left(\frac{2^{3/2}(6N-5)N^{1/2} T^{3/2}  \Lambda x^{1/2} \rho(\delta/2)}{\delta^{3/2}}\right)\right)}{\ln \left(\frac{2^{3/2}(6N-5)N^{1/2} T^{3/2}  \Lambda x^{1/2} \rho(\delta/2)}{\delta^{3/2}}\right)}.
\end{equation}
The number of qubits are given by \thm{ucost}.
\end{proof}

We can now easily see \cor{stepsAss}, the asymptotic scaling of implementing $e^{-iHT}$ in the fault-tolerant setting.

\begin{proof}[Proof of Corollary \ref{cor:stepsAss}]
Proof of the scaling of the number of operator exponential directly follows from~\lem{steps} and~\eqref{eq:sumorder}.
The scalings for the cost of fault-tolerant simulation is given in~\cor{faulttimeevolve}.
\end{proof}

\subsection{Costing $e^{-iHT}$ Simulation in NEG Model Setting}\label{sec:nisqanalysis}

Here, we briefly discuss the cost of implementing the unitary $e^{-iHT}$ within the NEG model (see~\sec{NEG}), since more rigorous analysis follows once we introduce measurement in \sec{negsim}. 

The analysis required in this model is different than that required in the Fault-Tolerant model. This is because, due to the assumptions of the NEG model, we may not approximate $e^{-iHT}$ to arbitrarily small error. However, we know that
\begin{align}
    \|e^{-iHT}-\widetilde{V}^s(T/s)\|  &\leq s \|e^{-iHT/s}-\widetilde{V}(T/s)\| \\
    &\leq s(\|e^{-iHT/s} - V(T/s)\| + \|V(T/s) - \widetilde{V}(T/s)\|) \\
    &\leq \frac{K}{s^2} + s \|\widetilde{V}(T/s) - V(T/s)\|, \label{eq:serrbound}
\end{align}
where $K$ (propagation time cubed, along with the parameters of the Schwinger model) is defined according to \eq{trotsteperr}. If we may evaluate $\|\widetilde{V}(T/s) - V(T/s)\|$, we can find the $s$ for which the bound above is minimized, which is near
\begin{equation}
    s_{\mathrm{err}} = \lceil( (2K/\|\widetilde{V}(t) - V(t)\|)^{1/3})\rceil.
\end{equation}
The number of CNOTs required to approximate $e^{-iHT}$ to nearly-optimal worst-case error is then $N_g s_{\mathrm{err}}$, with $N_g$ the bound on the number of CNOTs per timestep given in \lem{nisqtimestep}.

Another possible scenario is if we are given a fixed number $N_{x}$ of noisy CNOTs as a resource. In this case, we would be limited to $\lfloor N_x / N_g \rfloor$ timesteps, and the only remaining free parameters in simulation are in $K$, since $\|\widetilde{V}(t) - V(t)\|$ is independent of $t$ within the NEG Model. The effect of tuning these parameters on the error bound \eq{serrbound} is implied within \eq{trotsteperr}. 

In \sec{negsim} we will introduce measurement and translate the diamond-norm error due to the noisy CNOT resource to error in estimating a time-evolved observable (as opposed to time-evolution operator spectral-norm error). In that section, we give more complete analysis that may be easily applied to other near-term Schwinger model simulations of interest.

\section{Estimating Mean Pair Density}\label{sec:measurement}

An important objective of simulation is to compute observables rather than naively try to learn the amplitudes of the final state wavefunction.
Therefore, in order to fully analyze the cost of simulation, we must choose some illustrative observable to examine.
For the following, we adapt the above discussions to determine a sufficient number of calls to our time-evolution algorithm to estimate the expectation value of a simple example observable.
We will estimate the mean pair density by way of the mean positron density $\hat N_p/N$, where $\hat N_p$ counts the total number of positrons on the lattice, i.e. $\hat N_p = \sum_{i=1}^{N/2} \hat n_{2(i-1)}$; the mean positron density and mean pair density are synonymous in the physical sector.
However, we emphasize that much of this analysis would be similar for other operators of interest such as fermionic correlation functions of the form~\cite{wecker2015solving} $C_{j}(t)=\bra{\psi_0}e^{-iHt} a_j^\dagger e^{iHt} a_j\ket{\psi_0}$.
The cost of such simulations is asymptotically the same as those of the mean-pair operator after applying the Jordan-Wigner transformation and a ``Hadamard Test'' circuit to learn the expectation value.
However, we focus on the mean pair density because of its relative simplicity compared to other physically meaningful observables.

The following two subsections give distinct schemes of estimation.
\sec{measureerr} uses straightforward sampling.
\sec{ampest} employs amplitude estimation to achieve improved asymptotic scaling of calls to our time-evolution algorithm (by a factor of $1/\tilde \epsilon$).

\subsection{Estimating Mean Pair Density using Sampling}\label{sec:measureerr}

The following lemma determines a sufficient number of calls to our time-evolution algorithm (or ``shots'')~$N_{shots}$ to estimate the mean pair density, to rms error $\tilde \epsilon$ via straightforward sampling.

We give this result in terms of a parameter $\kappa$, which will be used later to optimize the resource scaling over the allotment of error to each contribution: the Trotter-Suzuki decomposition, the approximate circuit synthesis and the measurement error.
As in previous sections, we factor out the dominant terms for convenience in subsequent discussions.

\begin{lemma}
\label{lem:epsilonbound}
Let $0<\kappa<1$, and let $\ket{\psi},\ket{\phi} \in \mathbb{C}^{2^{\eta (N-1) + N}}$ be normalized states for the electric and fermionic fields on the Schwinger model lattice.
If $| \ket{\psi} - \ket{\phi}| \leq \tilde \epsilon\sqrt{ \kappa} $ and the number of samples in a stochastic sampling of the number of positrons in the state $\ket{\phi}$ is given by $$N_{shots} = \floor{\nu(\tilde \epsilon, \kappa) /4\tilde \epsilon^2(1-\kappa)},$$ where $$\nu(\tilde \epsilon, \kappa) = 4\tilde \epsilon^2(1-\kappa) + 1,$$ then the empirically observed sample mean $\mathcal{N}$ satisfies
$$ \mathbb{E}\left(\left(\frac{\mathcal{N}-\bra{\psi}\hat N_p \ket{\psi}}{N}\right)^2\right) \leq \tilde \epsilon^2. $$
\end{lemma}
\begin{proof}
First, since
\begin{equation}
    |\ket{\psi} - \ket{\phi}| \leq \tilde \epsilon\sqrt{\kappa},
    \label{eq:bound}
\end{equation}
we may thus bound
\begin{align}
    |\mathbb{E}_\psi(\hat N_p) - \mathbb{E}_\phi(\hat N_p) | &= |\bra{\psi} \hat N_p \ket{\psi} - \bra{\phi} \hat N_p \ket{\phi} | \\
    &\leq |\bra{\psi} \hat N_p \ket{\psi} - \bra{\psi} \hat N_p \ket{\phi} | + |\bra{\phi} \hat N_p \ket{\psi} - \bra{\phi} \hat N_p \ket{\phi}| \\
    &\leq 2 || \hat N_p || \tilde \epsilon\sqrt{\kappa} = N\tilde \epsilon\sqrt{\kappa}.
\end{align}
\noindent
The second line is given through the triangle inequality.
The third line is given by factoring and applying \eq{bound} twice, noting that both states are normalized.
Also, we used the fact that $||\hat N_p|| = N/2$.

We are assuming that $\mathcal{N}$ is an unbiased estimator of a stochastic sampling of the mean number of positrons in the approximate state, using $N_{shots}$ samples.
We have that $\mathbb{E}(\mathcal{N}) \equiv \bra{\phi} \hat{N}_p \ket{\phi}$.
The variance of this estimator is
\begin{equation}
    \mathbb{V}(\mathcal{N}) = \frac{\bra{\phi}\hat{N}_p^2 \ket{\phi}-(\bra{\phi}\hat{N}_p \ket{\phi})^2}{N_{shots}}\le \frac{N^2}{4N_{shots}}.
\end{equation}
We then have that the rms error in the process is
\begin{align}
    \mathbb{E}\left(\left(\frac{\mathcal{N} - \bra{\psi} \hat N_p \ket{\psi}}{N}\right)^2 \right)&= \frac{1}{N^2} \left ( \mathbb{E}(\mathcal{N}^2) -(\mathbb{E}(\mathcal{N}))^2+ (\bra{\psi} \hat N_p \ket{\psi}- \mathbb{E}(\mathcal{N}))^2 \right).\nonumber\\
    &= \frac{1}{N^2} \left (\mathbb{V}(\mathcal{N}) + |\mathbb{E}_\psi(\hat N_p) - \mathbb{E}_\phi(\hat N_p) |^2 \right )\nonumber\\
    &\le \frac{1}{4N_{shots}} + \tilde \epsilon^2 \kappa \leq \frac{1}{4} \cdot 4 \tilde \epsilon^2 (1-\kappa) + \tilde \epsilon^2 \kappa = \tilde \epsilon^2.
    \end{align}
\end{proof}

\subsection{Estimating Mean Pair Density using Amplitude Estimation}\label{sec:ampest}

We assumed in the previous discussion that the mean number of positrons is gleaned by repeatedly measuring the system.
A drawback of this approach is that the number of measurements scales as $O(1/\tilde \epsilon^2)$ where $\tilde \epsilon$ is the target rms error in the mean positron density.
In contrast, using ideas related to Grover's search, the mean value can be learned using quadratically fewer queries.
We present a simple method for estimating this mean using linear combinations of unitaries~\cite{childs2012hamiltonian,berry2015simulating}.

{The idea behind our approach is to express the total positron number as a sum of unitary operations.  We then try to apply this sum to a quantum state using a linear-combination of unitaries circuit.  The probability of success for this operation gives, in essence, the mean value of the total positron number.  We then use amplitude estimation (which is a combination of amplitude amplification and phase estimation) to learn the number of excitations.  The specific details for amplitude amplification can be found in Theorem 12} of~\cite{brassard2002quantum}.

In particular, the total number of positrons is
\begin{equation}\label{eq:numop}
    \hat{N}_p - N/4 = \sum_{\mathrm{even }\ j} \rho_j - N/4 = \sum_{\mathrm{even }\ j} n_j - N/4 = - \sum_{\mathrm{even }\ j} Z_j/2 \ .
\end{equation}
We will refer to $N_p$ and the above quantity $N_p - N/4$ interchangeably as the positron number.

\begin{figure}
\[
\includegraphics[]{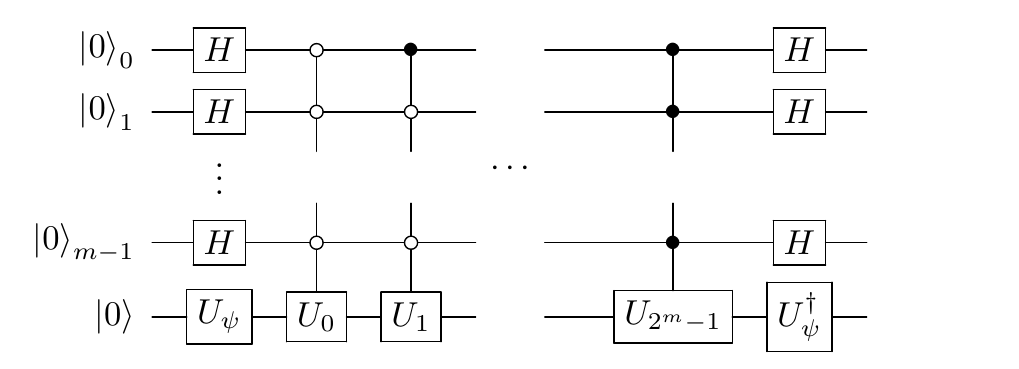}
\]

\caption{\label{fig:meancirc}Circuit for computing expectation value of a set of $2^m$ unitary operations, where each $U_j$ is a unitary operation.
The unitary operator $U_\psi$ is defined to prepare the state $\ket{\psi}$ i.e. $U_\psi: \ket{0} \mapsto \ket{\psi}$.
The expectation value is related to the probability of measuring all qubits to be in the zero state as shown in~\lem{genHadamard}.}
\end{figure}
\fig{meancirc} provides a way to estimate the mean positron number on a lattice with $2^{m}$ physical sites from the probability of measuring a control register to be $\ket{0}^{\otimes m}$.
Gauss' law equates the number of positrons with the number of electrons at all times in the physical sector dynamics, though the error budget allows for errors that cause the state to leak into the unphysical space.

\begin{lemma}[Generalized Hadamard Test]
\label{lem:genHadamard}
Let $U$ be the unitary transformation enacted in~\fig{meancirc},  then 
\begin{equation}
    {\rm Tr}(\left(U\ket{0}\!\bra{0} U^{\dagger}\right) \ketbra{0}{0})= \left|\sum_{j=0}^{2^m-1} \frac{\bra{\psi}U_{j}\ket{\psi}}{2^m} \right|^2. \nonumber
\end{equation}
Given access to controlled unitary implementations of $U_0,\ldots, U_{2^m-1}$ the circuit can also be implemented using $2m$ ancillary qubits within the Clifford+$T$ gate library.
\end{lemma}
\begin{proof}
The circuit implementation of $U$ performs
\begin{align}
    \ket{0} \ket{\psi} &\mapsto \frac{1}{\sqrt{2^m}} \sum_j \ket{j} \ket{\psi}\nonumber\\
    &\mapsto \frac{1}{\sqrt{2^m}} \sum_j \ket{j} U_j\ket{\psi}\nonumber\\
    &\mapsto \frac{1}{{2^m}} \sum_j \ket{0} U_j\ket{\psi} + C \ket{\phi}\nonumber\\
    &\mapsto \frac{1}{{2^m}} \sum_j \ket{0} U_{\psi}^\dagger U_j\ket{\psi} + C (I\otimes U_{\psi}^\dagger)\ket{\phi},
\end{align}
where $(\ketbra{0}{0}\otimes I)\ket{\phi}=0$ and $C$ is an irrelevant constant.
The probability of measuring all the qubits to be $0$ is therefore
\begin{equation}
    \left|(\bra{0}\bra{0})\frac{1}{{2^{m}}} \sum_j \ket{0} U_{\psi}^\dagger U_j\ket{\psi}\right|^2 = \left|\sum_{j=0}^{2^m-1} \frac{\bra{\psi}U_{j}\ket{\psi}}{2^m} \right|^2.
\end{equation}
The result then immediately follows.
\end{proof}

{We also have a slightly more general version of this result, which is appropriate in cases where a non-uniform sum of unitaries is desired (although such sums can be approximated closely using the prior result via the techniques shown} in~\cite{berry2015simulating}).
\begin{lemma}
\label{lem:genHadamard2}
Let $U_{\rm prep}: \ket{0} \mapsto \frac{1}{\sqrt{\|a\|_1}} \sum_k \sqrt{a_k} \ket{k} $ for $a_j>0$ and $\|a\|_1 = \sum_j |a_j|$ and let $U$ be the operation in \fig{meancirc}. Then
\begin{equation}
    {\rm Tr}(\left(U_{\rm prep} H^{\otimes m}U\ket{0}\!\bra{0} U^{\dagger} H^{\otimes m} U_{\rm prep}^\dagger\right) \ketbra{0}{0})= \left|\sum_{j=0}^{2^m-1} \frac{a_j\bra{\psi}U_{j}\ket{\psi}}{\|a\|_1} \right|^2. \nonumber
\end{equation}
Given access to controlled unitary implementations of $U_0,\ldots, U_{2^m-1}$ the circuit can also be implemented using $2m$ ancillary qubits within the Clifford+$T$ gate library.
\end{lemma}
\begin{proof}
Proof follows using the exact same reasoning as the above lemma.
\end{proof}

The above results may be applied to all sites of a Schwinger model lattice of $N = 2^{m+1}$ total fermionic sites.
For convenience, we shall restrict our attention to the reduced lattice of $2^m$ positronic sites when proving the following corollary.

\begin{corollary}[Expectation Value of Mean Pair Number]
\label{cor:fermmean}
Let $U$ be the unitary transformation enacted in~\fig{meancirc} on $m+2$ qubits where $$U_j=\begin{cases} -Z_j & j<2^{m}\\I & 2^{m+2}>j\ge 2^{m}\end{cases}$$ and $Z_j$ is the Pauli-$Z$ operator acting on the $j^{\rm th}$ qubit, and let $\Psi_j^\dagger$ be the creation operator acting on the positron mode $j\in[0,2^m-1]$.
Then $$\sum_{j=0}^{2^m-1} \bra{\psi}\Psi^\dagger_j \Psi_j\ket{\psi}=2^{m}\left( 2\sqrt{{\rm Tr}(\left(U\ket{0}\!\bra{0} U^{\dagger}\right) \ketbra{0}{0})}-1\right)$$
Further $U$ can be implemented using $2^m(m-2)$ Toffoli gates and $2(m+1)$ ancillas.
\end{corollary}
\begin{proof}
The proof follows from the generalized Hadamard test given in~\lem{genHadamard} after applying the Jordan-Wigner transformation.
The Jordan-Wigner representation the fermion number operator acting on site $j$ is $\Psi_j^\dagger\Psi_j=(I-Z)/2$.
Thus the mean positron number is
\begin{equation}
    \sum_{j=0}^{2^m-1} \bra{\psi}\Psi^\dagger_j \Psi_j\ket{\psi}=2^{m-1}\left( 1-  \sum_{j=0}^{2^m-1} \frac{\bra{\psi} Z_j \ket{\psi}}{2^m}\right)
\end{equation}
Next it follows from the definition of $U_j$ that
\begin{align}
    \left(\sum_{j=0}^{2^{m+2}-1} \frac{\bra{\psi}U_{j}\ket{\psi}}{2^{m+2}} \right)^2&=\left(\frac{3}{4}-\sum_{j=0}^{2^{m}-1} \frac{\bra{\psi}Z_{j}\ket{\psi}}{2^{m+2}} \right)^2\nonumber\\
    &=\left(\frac{1}{2} \left( 1 + \frac{1}{2^{m}}\sum_{j=0}^{2^{m}-1} \bra{\psi} \Psi_j^\dagger \Psi_j \ket{\psi} \right) \right)^2
\end{align}
We therefore have from~\lem{genHadamard} that
\begin{equation}
    \sqrt{{\rm Tr}(\left(U\ket{0}\!\bra{0} U^{\dagger}\right) \ketbra{0}{0})}=\left(\frac{1}{2} \left( 1 + \frac{1}{2^{m}}\sum_{j=0}^{2^{m}-1} \bra{\psi} \Psi_j^\dagger \Psi_j \ket{\psi} \right) \right)
\end{equation}

The bound on the number of ancillary qubits comes from noting that if $U_j$ has $m$ controls then we can implement this using a $m-1$ controlled not gate and a single controlled $U_j$ gate using an additional ancilla to store the output.
Using~\cite{barenco1995elementary}, such a gate can be implemented using at most $m-2$ Toffoli gates each requires an ancillary qubit.
Therefore the number of ancillas needed in this protocol is $m-1$.
However, if each Toffoli gate is implemented using~\cite{jones2013low} then an additional ancilla is required to store the output.
Thus the number of qubits needed to implement the controlled operations is at most $m$.
Since $m+2$ ancillas are required in in the generalized Hadamard test protocol on $m+2$ qubits, this shows that the total number of ancillas used (including those used in the generalized Hadamard test circuit)  is at most $2(m+1)$.
\end{proof}

\begin{theorem}[Estimating Mean Pair Density]\label{thm:ae}
Let $U_\psi$ be an oracle such that $U_{\psi} \ket{0} = \ket{\psi}$ and assume the user only has access to an oracle $\widetilde{U}_\psi$ such that $\|\widetilde{U}_{\psi} - U_{\psi}\| \le \frac{\tilde \epsilon}{8}$.
There exists a quantum circuit that outputs an estimate $\mathcal{N}$ of the mean positron number per site such that the rms error in the estimate satisfies $\mathbb{E}(|\mathcal{N} - \frac{1}{2^m}\sum_{j=0}^{2^m-1} \bra{\psi}\Psi^\dagger_j \Psi_j\ket{\psi}|^2)\le \tilde \epsilon^2<1$ using a number of queries to $\widetilde{U}_{\psi}$ that is bounded above by
$$
    N_{\rm query}\le \frac{128\pi}{16-\pi^2}\left(\frac{\pi }{\tilde \epsilon}+{2} \right)\ln\left( \frac{5}{\tilde \epsilon^2}\right),
$$
and a number of auxillary $T$-gates that are bounded above by
$$
    N_{\rm T}\le \left(\frac{32\pi(2^{m+4}+8)(m-1)}{16-\pi^2}\left(\frac{\pi }{\tilde \epsilon}+{2} \right)+\frac{7}{3} \log^2\left(\frac{2\sqrt{2}\pi}{\tilde \epsilon}+4 \right)\log\left( \frac{16\log^2\left(\frac{2\sqrt{2}\pi}{\tilde \epsilon}+4 \right)}{\tilde \epsilon^2} \right)\right)\ln\left( \frac{5}{\tilde \epsilon^2}\right)
$$
\end{theorem}
\begin{proof}
The proof follows directly from prior results, the Amplitude Estimation Lemma (as given in Lemma 5 of~\cite{brassard2002quantum}) and the Chernoff bound.

For brevity, let us define $\hat{N}_p=\sum_{j=0}^{2^{m}-1} \bra{\psi} \Psi_j^\dagger\Psi_j \ket{\psi}$.
First note that assuming that $\mathcal{N}$ is an unbiased estimator of $\bra{\widetilde{\psi}}\hat{N}_p \ket{\widetilde{\psi}}$ (which amplitude estimation yields) we have that
\begin{align}
    \mathbb{E}\left(\left(\mathcal{N} - \bra{\psi} \hat{N}_p \ket{\psi} \right)^2\right)&=\mathbb{E}\left(\mathcal{N}^2 -2\bra{\widetilde{\psi}}\hat{N}_p \ket{\widetilde{\psi}}\bra{\psi} \hat{N}_p \ket{\psi}+ \bra{\psi} \hat{N}_p \ket{\psi}^2 \right)\nonumber\\
    &=\mathbb{E}\left(\mathcal{N}^2 -\bra{\widetilde{\psi}}\hat{N}_p \ket{\widetilde{\psi}}^2 + \left(\bra{\widetilde{\psi}}\hat{N}_p \ket{\widetilde{\psi}}- \bra{\psi} \hat{N}_p \ket{\psi}\right)^2 \right)\nonumber\\
    &= \mathbb{E}\left(\mathcal{N}^2 -\bra{\widetilde{\psi}}\hat{N}_p \ket{\widetilde{\psi}}^2\right) + \left\|\bra{\widetilde{\psi}}\hat{N}_p \ket{\widetilde{\psi}}- \bra{\psi} \hat{N}_p \ket{\psi}\right\|^2 
\end{align}
Thus the mean square error is at most the sum of the variance of the estimate from amplitude estimation using a faulty state and the maximum square error in the mean due to simulation and synthesis.

We have from~\cor{fermmean} that if the expectation value of the probability deviates by $\delta_L$ from the true expectation value $\left(\frac{1}{2}\left(1+\frac{1}{2^m}\sum_{j=0}^{2^m-1} \bra{\psi}\Psi^\dagger_j \Psi_j\ket{\psi}\right)\right)^2$ then the deviation in the inferred mean positron number is (provided $\delta_L \le 3{{\rm Tr}(\left(U\ket{0}\!\bra{0} U^{\dagger}\right) \ketbra{0}{0})}/4$) at most 
\begin{align}
    & \left|\sqrt{{\rm Tr}(\left(U\ket{0}\!\bra{0} U^{\dagger}\right) \ketbra{0}{0})} - \sqrt{{\rm Tr}(\left(U\ket{0}\!\bra{0} U^{\dagger}\right) \ketbra{0}{0})+\delta_L} \right|\nonumber\\
    &\qquad\le \max \left|\partial_{\delta_L}\sqrt{{\rm Tr}(\left(U\ket{0}\!\bra{0} U^{\dagger}\right) \ketbra{0}{0})+\delta_L}\right| |\delta_L|\nonumber\\
    &\qquad\le\frac{|\delta_L|}{\sqrt{{\rm Tr}(\left(U\ket{0}\!\bra{0} U^{\dagger}\right) \ketbra{0}{0})}}\nonumber\\
    &\qquad=\frac{2|\delta_L|}{1+\frac{1}{2^m}\sum_{j=0}^{2^m-1} \bra{\psi}\Psi^\dagger_j \Psi_j\ket{\psi}}.
\end{align}

Next we need to relate $\delta_L$ to the simulation error.
First note that from the triangle inequality we have that operator $O$ of norm $1$
\begin{equation}
    \|\bra{0} U_\psi^\dagger O U_{\psi} \ket{0} - \bra{0} \widetilde{U}_\psi^\dagger O \widetilde{U}_{\psi} \ket{0}\|\le \|\widetilde{U}_{\psi}^\dagger - U_{\psi}^\dagger\|\|OU_{\psi}^\dagger\| + \|\widetilde{U}_{\psi} - U_{\psi}\| \|OU_{\psi}\|\le 2\|\widetilde{U}_{\psi} - U_{\psi}\|.\label{eq:trianglebound}
\end{equation}

We then have that if $\| U_{\psi}- \widetilde{U}_\psi\| = \delta$ and $\hat{N}_p$ is Hermitian then it follows from~\eqref{eq:trianglebound} and the fact that $\|\hat{N}_p/2^{m}\|=1$ that
\begin{align}
    |\delta_L|\le 2^{-m}\left\|\bra{\widetilde{\psi}}\hat{N}_p \ket{\widetilde{\psi}}- \bra{\psi} \hat{N}_p \ket{\psi}\right\|\le 2 \delta
\end{align}

In turn we have that since $\hat{N}_p \succeq 0$
\begin{equation}
    2^{-m}\left\|\bra{\widetilde{\psi}}\hat{N}_p \ket{\widetilde{\psi}}- \bra{\psi} \hat{N}_p \ket{\psi}\right\|\le \frac{2|\delta_L|}{1+\frac{1}{2^m}\sum_{j=0}^{2^m-1} \bra{\psi}\Psi^\dagger_j \Psi_j\ket{\psi}} \le 4\|U_{\psi} - \widetilde{U}_{\psi}\|
\end{equation}
This implies that $2^{-2m}\left\|\bra{\widetilde{\psi}}\hat{N}_p \ket{\widetilde{\psi}}- \bra{\psi} \hat{N}_p \ket{\psi}\right\|^2 \le \tilde \epsilon^2/4$ if

\begin{equation}
\|U_{\psi}-\widetilde{U}_{\psi}\|\le \frac{\tilde \epsilon}{8}
\end{equation}
as claimed.

Now let us turn our attention to estimating the mean positron number using amplitude estimation.
The Amplitude Estimation Theorem (stated as Theorem 12 in~\cite{brassard2002quantum}) states that there exists a protocol that yields an estimate of the probability using $L$ iterations of the Grover operator that is defined by the reflections $I-2\ketbra{0}{0}$ and $I-2U\ketbra{0}{0}U^\dagger$.
The error in the estimate yielded in a bit-string by amplitude estimation can be made at most $\epsilon_L$ with probability at least $8/\pi^2$, by Lemma 5 of~\cite{brassard2002quantum}, by choosing $L$ to satisfy
\begin{equation}
    \frac{\sqrt{2}\pi}{L}+ \left(\frac{\pi}{L} \right)^2 = \epsilon_L.
\end{equation}
Solving the resultant quadratic equation for $L> 0$ yields 
\begin{align}\label{eq:Lbd}
    L&= \frac{ \pi}{\sqrt{2}\epsilon_L}\left(1+\sqrt{1+2\epsilon_L} \right)\le \left\lceil \frac{2 \pi}{\sqrt{2}\epsilon_L} +\frac{\pi}{\sqrt{2}} \right\rceil \le \frac{\sqrt{2}\pi}{\epsilon_L}+4.
\end{align}

On the other hand, if an error occurs in the amplitude estimation routine then the error in the inferred positron number is at most $2^m$, which means the error per site is at most $1$.
Thus if the success probability is at least $P$ and if we choose $\delta_L \le (1+\frac{1}{2^m}\sum_{j=0}^{2^m-1} \bra{\psi}\Psi^\dagger_j \Psi_j\ket{\psi})/4$ then the mean-square error in the procedure is at most
\begin{equation}
    P\left(\left(\frac{2|\delta_L|}{1+\frac{1}{2^m}\sum_{j=0}^{2^m-1} \bra{\psi}\Psi^\dagger_j \Psi_j\ket{\psi}}\right)^2+\epsilon_L^2\right)+ \left(1-P\right)\le P\left( \frac{\epsilon^2}{4}+\epsilon_L^2\right)+ \left(1-P\right).
\end{equation}
Next let us choose and $\epsilon_L={\epsilon}/{2}$, where $\epsilon^2$ is the target upper bound in the mean-square error for the estimate.
We then have that
\begin{equation}\label{eq:Peqthingy}
    \frac{P\epsilon^2}{2} +(1-P)=\epsilon^2
\end{equation}
In order to satisfy~\eqref{eq:Peqthingy} it suffices to aim for a success probability that satisfies
\begin{equation}\label{eq:epscond}
    P\ge \frac{2(1-\epsilon^2)}{2-\epsilon^2}.
\end{equation}
As the probability of success is at least $8/\pi^2$~\cite{brassard2002quantum}, we cannot directly guarantee that we can achieve this accuracy using a single run.

In fact the probability of success from amplitude estimation is in practice slightly lower than the $8/\pi^2$ lower bound given in~\cite{brassard2002quantum}.
This is because finite error is introduced in implementing the quantum Fourier transform using $H$, $T$ and CNOT gates.
For any projector $P$, and unitary operations $V$ and $\widetilde{V}$ we have that
\begin{align}
    \left|{\rm Tr}(P V\ketbra{0}{0}V^\dagger) - {\rm Tr}(P \widetilde{V}\ketbra{0}{0}\widetilde{V}^\dagger)\right| &\le \left|{\rm Tr}(P V\ketbra{0}{0}V^\dagger) - {\rm Tr}(P \widetilde{V}\ketbra{0}{0}{V}^\dagger)\right| \nonumber\\
    &\qquad+ \left|{\rm Tr}(P \widetilde{V}\ketbra{0}{0}V^\dagger) - {\rm Tr}(P \widetilde{V}\ketbra{0}{0}\widetilde{V}^\dagger)\right|\nonumber\\
    &\le 2\|V -\widetilde{V}\|.
\end{align}
Thus if the error in the Fourier transform is at most $\epsilon_{QFT}$ then from the union bound\footnote{For a countable set of events, the union bound (or Boole's inequality) states that the probability at least one event happens is at most the sum of the probabilities of each individual event.} we have that it suffices to take
\begin{equation}
P= P_0 - 2\epsilon_{QFT} \ge \frac{2(1-\epsilon^2)}{2-\epsilon^2}
\end{equation}

The success probability can, however, be amplified to this level at logarithmic cost using the Chernoff bound.
Specifically, if we take the median of $R$ different runs of the algorithm then the probability that the median constitutes a success is at least
\begin{equation}
    P_0\ge 1-e^{-\frac{R}{2}\left(\frac{\sqrt{8}}{\pi} - \frac{\pi}{2\sqrt{8}} \right)}
\end{equation}
This implies that~\eqref{eq:epscond} can be satisfied if we take $\epsilon \le 1$, $\epsilon_{QFT} \le \epsilon/8$ and
\begin{align}
    R&=\left\lceil \frac{8\sqrt{2}\pi \ln\left( \frac{2-\epsilon^2}{\epsilon^2+2\epsilon_{QFT}\epsilon^2 -4\epsilon_{QFT}}\right)}{16-\pi^2}\right\rceil \le \frac{8\sqrt{2}\pi}{16-\pi^2}\left(\ln\left( \frac{2-\epsilon}{\epsilon^2+2\epsilon_{QFT}\epsilon -4\epsilon_{QFT}}\right)+\frac{\sqrt{2}}{8} \right) \nonumber\\
    &\le  \frac{8\sqrt{2}\pi}{16-\pi^2}\left(\ln\left( \frac{4}{\epsilon^2}\right)+\frac{\sqrt{2}}{8} \right)\le\frac{8\sqrt{2}\pi}{16-\pi^2}\ln\left( \frac{5}{\epsilon^2}\right)\label{eq:Rbd}
\end{align}
Next let us count the total number of operations.
Every application of $U$ involves two queries to $U_\psi$.
Thus since two applications of $U$ are required in the Grover operator $I-2U\ketbra{0}{0}U^\dagger$ we then have that four queries of $U_\psi$ are required per Grover iteration.
Thus it follows from~\eqref{eq:Lbd} and \eqref{eq:Rbd} that the total number of queries needed obeys
\begin{equation}
    N_{\rm query} \le 4LR \le \frac{128\pi}{16-\pi^2}\left(\frac{\pi }{{\epsilon}}+{2} \right)\ln\left( \frac{5}{\epsilon^2}\right).
\end{equation}

The number of $T$-gates required in addition to the oracle queries is slightly harder to find.
First, note from~\cite{barenco1995elementary} and~\cite{jones2013low} that an $(m-1)$-controlled Toffoli gate requires $2(m-2)+1$ standard (two-controlled) Toffoli gates to implement, which in turn requires at most $4(2m-1)$ $T$-gates if measurement is used.
Further using the result of~\cite{gidney2018halving}, half of these Toffolis can be negated using measurement and post-selected Clifford operations, as used prior in \sec{kineticimplementFT}.
Thus the minimal cost of each $m$-controlled Toffoli gate is at most $4(m-1)$ $T$-gates.
There are at most $2^{m+1}$ such gates so the number of auxiliary operations is $2^{m+3}(m-1)$ to implement $U$ and thus $(2^{m+4}+4)(m-1)$ auxiliary $T$-gates suffice to implement $I-2U\ketbra{0}{0}U^\dagger$.
Similarly, at most $4(m-1)$ $T$-gates are needed to implement the controlled gate in $I-2\ketbra{0}{0}$.
Each Grover iteration therefore requires at most $(2^{m+4}+8)(m-1)$ auxiliary $T$-gates.
Therefore
\begin{align}
    N_T&\le ((2^{m+4}+8)(m-1)L+N_{T,QFT})R\nonumber\\
    &\le \left(\frac{32\pi(2^{m+4}+8)(m-1)}{16-\pi^2}\left(\frac{\pi }{{\epsilon}}+{2} \right)+N_{T,QFT}\right)\ln\left( \frac{5}{\epsilon^2}\right) \label{eq:NTAE}
\end{align}

Next we need to bound the complexity of performing the quantum Fourier transform as part of the phase estimation algorithm used above.
As discussed previously, we have that $\lceil \log(L)\rceil =\left\lceil \log\left(\frac{2\sqrt{2}\pi}{{\epsilon}}+4 \right) \right\rceil$ additional qubits are needed to perform the phase estimation.
Thus the number of single qubit rotations needed is at most
\begin{equation}
    \lceil \log(L)\rceil (\lceil \log(L)\rceil-1) \le \log(L)(\log(L)+1)\le 2\log^2\left(\frac{2\sqrt{2}\pi}{{\epsilon}}+4 \right),
\end{equation}
We need to ensure that the error is at most $\epsilon^2/8$ from these rotations.
From Box 4.1 of~\cite{nielsen2002quantum} we then have that the error per single qubit rotation can be as high as $\epsilon^2/(8\log(L)(\log(L)+1))$.
The average gate complexity for implementing these rotations using RUS synthesis~\cite{bocharov2015efficient} is at most
\begin{align}
    N_{T,QFT}&\le 2.28 \log^2\left(\frac{2\sqrt{2}\pi}{\epsilon}+4 \right)\log\left( \frac{16\log^2\left(\frac{2\sqrt{2}\pi}{{\epsilon}}+4 \right)}{\epsilon^2} \right)\nonumber\\
    &\le \frac{7}{3} \log^2\left(\frac{2\sqrt{2}\pi}{{\epsilon}}+4 \right)\log\left( \frac{16\log^2\left(\frac{2\sqrt{2}\pi}{{\epsilon}}+4 \right)}{\epsilon^2} \right).\label{eq:NTQFT}
\end{align}
The claim involving the number of auxillary $T$-gates then follows from~\eqref{eq:NTAE} and~\eqref{eq:NTQFT} after taking $\epsilon = \widetilde{\epsilon}$.
\end{proof}

\subsection{Importance Sampling}
{In the above discussion we assumed that we wish to use an unbiased estimator of the mean fermion number, which is} appropriate to use if we have no a priori knowledge of the underlying distributions of fermions at each site.  However, in practice this is unlikely to be true.
For example, the likelihood of measuring a result that is many excitations away from the initial state should be low.  This is because mean energy is a constant of motion and Markov's inequality states that the probability of observing a state with energy $a$ times the mean is at most $1/a$.  
In such cases, we can therefore reduce the variance of our estimate by appropriately re-weighting the sum using importance sampling.  The use of importance sampling, or alternatively prior information, can be used to dramatically reduce the number of classical measurements needed and in particular such knowledge can be used to outperform the na\"{i}ve quantum algorithm despite the quantum speedup.

Importance sampling works as follows.  If we are to compute an expectation value of a function $f$ given an underlying probability distribution $P(x)$ and a proposal distribution $Q(x)>0$ then we can write
\begin{equation}
    \mathbb{E}(f) = \sum_x f(x) P(x) = \sum_x \frac{f(x) P(x)}{Q(x)} Q(x),
\end{equation}
If we define $g(x) := f(x)P(x)/Q(x)$ then we can view this as an expectation value of $g(x)$ over the proposal distribution $Q(x)$.  While nothing has changed in the underlying mean by this substitution, the variance of the result can be lowered by an appropriate choice of $Q(x)$.  This means that fewer samples are needed if such a technique is employed.  While the scaling with $\epsilon$ will not approach that of Heisenberg-limited quantum approaches, if $H=\sum_j a_j U_j$ (for Hamiltonian coefficients $a_j$ and unitary operators $U_j$), the scaling with respect to the $a_j$ may be reduced below that of the quantum approaches if importance sampling lowers the variance to $o((\sum_j|a_j|)^2)$.\footnote{Here we use ``little-o'' notation to mean that for any functions $f$ and $g$, $f(n)\in o(g(n))$ if and only if $\lim_{n\rightarrow \infty} f(n)/g(n)=0$.} This begs the question of whether quantum techniques can also employ a similar strategy to reduce the query complexity given prior knowledge of the distribution.

The quantum approach that we propose to address this issue  is based on similar reasoning to importance sampling and can, in some cases, reduce the scaling with the number of terms.  We state the first result formally in the following proposition.

\begin{proposition}
Let $\ket{\psi} = \sum_j \alpha_j \ket{\psi_j}$ for $\alpha_j>0$ and $\ket{\psi_j} = P_j \ket{\psi}/ |P_j \ket{\psi}|$ 
where $P_j$ is the projector onto all states that have $j$ positronic modes occupied.  Further, let $ \beta_j \ge\beta_{\min}>0$ and $\alpha_j\le \beta_j$ for an efficiently computable sequence $\beta_j$.  
Then with high probability, for any $\tilde \epsilon \le (\frac{1}{2}) \sum_k |\alpha_k|^2/|\beta_k|^2$, the number of queries to $U_\psi$ needed to learn an estimate, $\mathcal{N}$, of the mean positron number per site such that the rms error in the estimate satisfies $\mathbb{E}(|\mathcal{N} - \frac{1}{2^m}\sum_{j=0}^{2^m-1} \bra{\psi}\Psi^\dagger_j \Psi_j\ket{\psi}|^2)\le \tilde \epsilon^2<1$ is in
$$
N_{\rm query}\in \widetilde{O}\left( \frac{\max(\sum_j |\beta_j|^2 j,1)}{2^{m/2}\beta_{\min} \widetilde{\epsilon}}\right).
$$
\end{proposition}
\begin{proof}
First, note that the mean positron density can be expressed as $\mathcal{M}:=\sum_{j=0}^{2^m-1} j P_j/2^m$.  Next note that under our assumptions $\bra{\psi}\mathcal{M}\ket{\psi} = \sum_j |\alpha_j|^2 j/2^m$.  It then follows that for any sequence $\beta_j>0$,
\begin{equation}
    \bra{\psi}\mathcal{M}\ket{\psi} = \sum_j \frac{|\alpha_j|^2}{|\beta_j|^2} \left(\frac{|\beta_j|^2 j}{2^m}\right).
\end{equation}
If $0\le \alpha_j / \beta_j \le 1$ then this can  be thought of as an expectation of the new observable
\begin{equation}
    \big(\sum_k |\alpha_k|^2/ |\beta_k|^2\big)\sum_j|\beta_j|^2 j P_j/2^m
\end{equation} 
with respect to $\sum_j \alpha_j/\beta_j \ket{\psi_j} /\sqrt{\sum_j |\alpha_j|^2/|\beta_j|^2}$. The strategy is to prepare this state and estimate both the overall factor, $\sum_k |\alpha_k|^2/ |\beta_k|^2$, and the new observable in order to reconstruct $\mathbb{E}(M)$.

Next, using a single query to $U_\psi$ we can prepare the state $\sum_j \alpha_j \ket{\psi_j}$.  Thus by computing the function $\beta_j$ on the input value of $j$ (which can efficiently be computed from the inputs using a poly-sized adder circuit), we can perform the following transformation using $O(1)$ queries, the calculation of an inverse trig function and a sequence of controlled rotations:
\begin{equation}
\ket{0} \mapsto \sum_j \alpha_j \ket{\psi_j} \ket{\beta_j} \ket{\sin^{-1}(\beta_{\min}/\beta_j)}\left(\frac{\beta_{\min}}{\beta_j} \ket{0} + \sqrt{1 - \frac{\beta_{\min}^2}{\beta_j^2}} \ket{1} \right).
\end{equation}
If this ancillary qubit is measured to be zero (and the output from the $\beta_j$ oracle query and the arcsine are inverted) then this allows us to prepare the state $\sum_j \alpha_j/\beta_j \ket{\psi_j} /\sqrt{\sum_j |\alpha_j|^2/|\beta_j|^2}$ using a form of quantum rejection sampling~\cite{ozols2013quantum} and is exactly the state that we need to compute our expectation value for this version of quantum importance sampling.
The probability of measuring the ancilla qubit to be zero, which we will hitherto consider to be ``success,'' is
\begin{align}
    \sum_j \beta_{\min}^2|\alpha_j|^2 / |\beta_j|^2 \le 2^{m} \beta_{\min}^2.
\end{align}
Thus if we use amplitude amplification to boost this probability we can, based on these assumptions, transform the success probability using $O(1/(\beta_{\min}2^{m/2}))$ queries to $U_{\psi}$ to
\begin{align}
    P'(0):=&\sin^2\left( \left(2\left\lfloor \frac{\pi}{4\sin^{-1}\left( \beta_{\min} 2^{m/2}\right)} -\frac{1}{2} \right\rfloor+1 \right) \sin^{-1}\left(\sqrt{\sum_j \beta_{\min}^2|\alpha_j|^2 / |\beta_j|^2}\right) \right)\nonumber\\
    &\in O\left(\sin^2\left( \frac{\pi}{2} \sqrt{\frac{1}{2^{m}}\sum_j\frac{|\alpha_j|^2}{|\beta_j|^2}} \right) \right).\label{eq:preamp}
\end{align}
Let $\Delta := 1 - \sqrt{\frac{1}{2^m}\sum_j\frac{|\alpha_j|^2}{|\beta_j|^2}}\ge 0$, then we have from Taylor's theorem that the success probability now scales as
\begin{equation}
    P'(0) \in O\left(\sin^2\left(\frac{\pi}{2} \left(1-\Delta \right) \right) \right) \in O\left(\frac{1}{2^m} \sum_{j} \frac{|\alpha_j|^2}{|\beta_j|^2} \right),
\end{equation}
and $O(1/(\beta_{\min} 2^{m/2}))$ queries to $U_\psi$ are needed to amplify the probability of $0$ to this level.  Thus by using amplitude estimation on the probability in~\eqref{eq:preamp} we can learn for any $\tilde{\delta}>0$ an estimate $E'$ with high probability such that~\cite{brassard2002quantum} 
\begin{equation}
    |E'- P'(0)| \le 2^{-m}\widetilde{\delta},
\end{equation}
using $O(2^{m} \widetilde{\delta}^{-1})$ applications of the underlying unitary circuit~\cite{brassard2002quantum}, which in this case consists of $O(1/(\beta_{\min} 2^{m/2}))$ oracle queries.
Therefore we can learn $2^{-m}\sum_j |\alpha_j|^2/|\beta_j|^2$ within relative error $2^{-m} \widetilde{\delta} = \widetilde\epsilon/  \left( {2^m}{ \max(\sum_j |\alpha_j|^2 j,1)}\right)$ using
\begin{equation}\label{eq:epsilon0}
    N_{\rm queries} \in O\left(\frac{1}{\beta_{\min} \widetilde{\epsilon}}  \left(\frac{\max(\sum_j|\alpha_j|^2 j,1)}{2^{m/2}}\right) \right).
    \end{equation}
Equivalently, we now know the value of $\sum_j |\alpha_j|^2/|\beta_j|^2$ within error $\widetilde{\epsilon} / \max(\sum_j (|\alpha_j|^2 j,1))$, with high probability, using the number of queries given in~\eqref{eq:epsilon0}.  In order to ensure that the error is non-zero here we will take $\widetilde{\epsilon} \le \frac{1}{2}\sum_j |\alpha_j|^2/|\beta_j|^2$ to ensure that the resulting estimate is greater than zero.

We can then employ our knowledge of $\sum_j |\alpha_j|^2/|\beta_j|^2$ to amplify the probability of success that we would have seen from the above quantum rejection sampling step to near unit probability.  Specifically, we use amplitude amplification~\cite{brassard2002quantum,yoder2014fixed} to achieve a unitary transformation that maps, for any $\widetilde{\delta'}>0$ 
\begin{equation}
\sum_j \alpha_j \ket{\psi_j}\left(\frac{\beta_{\min}}{\beta_j} \ket{0} + \sqrt{1 - \frac{\beta_{\min}^2}{\beta_j^2}} \ket{1} \right) \mapsto  \frac{1}{\sqrt{\sum_k \frac{|\alpha_k|^2}{|\beta_k|^2}}}\sum_j \frac{\alpha_j}{\beta_j} \ket{\psi_j} + \widetilde{\delta'} \ket{\phi}:=\ket{\chi}
\end{equation}
for some unit-vector $\ket{\phi}$ (which is not necessarily orthogonal to $\ket{\psi_j}$)
using 
\begin{equation}
    O\left(\log (1/\widetilde \delta')/\left(\sqrt{\sum_k \frac{|\alpha_k|^2}{|\beta_k|^2}}\beta_{\min} \right) \right)
\end{equation} 
queries to $U_\psi$.

Next using the above trick we can decompose each $P_j$ into a sum of unitary operations as follows:
\begin{equation}
    P_j = \frac{I - (I - 2 P_j)}{2}.
\end{equation}
Thus we can express $\sum_j\left(\frac{|\beta_j|^2 jP_j}{2^m}\right) = \sum_j a_j U_j$, for unitary matrices $U_j$ with $\|a\|_1 \in O\left(\frac{\sum_j|\beta_j|^2 j}{2^m}\right)$.  
It further follows that
\begin{align}\label{eq:notrescaled}
    \bra{\chi} \sum_j a_j U_j \ket{\chi} = \left(\sum_k \frac{|\alpha_k|^2}{|\beta_k|^2}\right)^{-1} \bra{\psi} \mathcal{M} \ket{\psi} + O\left(\frac{\widetilde{\delta'}}{\sqrt{\sum_k \frac{|\alpha_k|^2}{|\beta_k|^2}}} \right).
\end{align}
It then follows that
\begin{equation}\label{eq:rescaled_bd}
    \left|\sum_k\frac{|\alpha_k|^2}{|\beta_k|^2}\bra{\chi} \sum_j a_j U_j \ket{\chi} - \bra{\psi} \mathcal{M} \ket{\psi} \right| \in O\left(\widetilde{\delta}'\sqrt{\sum_k\frac{|\alpha_k|^2}{|\beta_k|^2}} \right).
\end{equation}

Thus using \lem{genHadamard2} we can learn an estimate $E$ such that 
\begin{equation}
    |E-\bra{\chi} \sum_j a_j U_j \ket{\chi}| \le \widetilde{\epsilon} \left(\sum_k \frac{|\alpha_k|^2}{|\beta_k|^2}\right)^{-1},
\end{equation}
with high probability, using
\begin{equation}
    N_{\rm queries } \in {O}\left(\sum_k \frac{|\alpha_k|^2}{|\beta_k|^2}\frac{\sum_j |\beta_j|^2 j}{2^m  \widetilde{\epsilon}}\left( \frac{\log(1/\widetilde{\delta}')}{\sqrt{\sum_k \frac{|\alpha_k|^2}{|\beta_k|^2}}\beta_{\min}}\right)  \right).
\end{equation}
We then see from~\eqref{eq:notrescaled} and~\eqref{eq:rescaled_bd} that with high probability
\begin{align}
    &\left|\sum_k\frac{|\alpha_k|^2}{|\beta_k|^2} E - \bra{\psi} \mathcal{M} \ket{\psi} \right| \nonumber\\
    &\qquad\le \left|\sum_k\frac{|\alpha_k|^2}{|\beta_k|^2} \left(E-\bra{\chi} \sum_j a_j U_j \ket{\chi}\right) \right| + \left|\sum_k\frac{|\alpha_k|^2}{|\beta_k|^2} \bra{\chi} \sum_j a_j U_j \ket{\chi} - \bra{\psi} \mathcal{M} \ket{\psi}\right|\nonumber\\
    &\qquad \in O\left( \widetilde{\epsilon} + \widetilde{\delta}'\sqrt{\sum_k\frac{|\alpha_k|^2}{|\beta_k|^2}}\right).
\end{align}

The result then follows by taking $\widetilde{\delta}' \in O(\widetilde{\epsilon} / \sqrt{\sum_k |\alpha_k|^2/|\beta_k|^2})$ and noting that the total number of queries made to estimate $E$ is in
\begin{equation}\label{eq:epsilon1}
    N_{\rm queries}\in \widetilde{O}\left( \sqrt{\sum_k \frac{|\alpha_k|^2}{|\beta_k|^2}}\frac{\sum_j |\beta_j|^2 j}{2^m \beta_{\min}  \widetilde{\epsilon}} \right).
\end{equation}
The total complexity is this added to the cost of learning $\sum_k |\alpha_k|^2/|\beta_k|^2$ within relative error $\widetilde{\epsilon} / \bra{\psi} \mathcal{M} \ket{\psi} $ which is given by~\eqref{eq:epsilon0} and~\eqref{eq:epsilon1} to be
\begin{align}
    N_{\rm queries}&\in \widetilde{O}\left( \frac{1}{\beta_{\min} \widetilde{\epsilon}} \left(\sqrt{\sum_k \frac{|\alpha_k|^2}{|\beta_k|^2}}\frac{\sum_j |\beta_j|^2 j}{2^m } + \left(\frac{\max(\sum_j|\alpha_j|^2j,1)}{2^{m/2}}\right)\right) \right)\nonumber\\
    &\in \widetilde{O}\left( \frac{\max(\sum_j |\beta_j|^2 j,1)}{2^{m/2}\beta_{\min} \widetilde{\epsilon}}\right).
\end{align}
\end{proof}

It follows that this importance sampling inspired approach will potentially have an advantage over using importance sampling on the outputs if $\sum_j |\beta_j|^2 j / \beta_{\min} \in o(2^{m/2})$.  This will occur when the function $\beta_j$ has support only over a small fraction of the allowable modes.

Additionally prior information can be used in a more sophisticated manner using Bayesian methods for phase estimation in place of the Fourier-based PE that we currently use in phase estimation~\cite{svore2013faster,wiebe2016efficient}.  This allows any prior information about the mean-particle number to be used in a well principled way (at the price of increased classical computation).  However, we leave detailed exploration of these methods for subsequent work because in order to show the {benefits from using the prior we need to assume a well motivated prior distribution, which is beyond this paper's scope.}

\section{Simulation including Estimation of Mean Pair Density}\label{sec:fullsim}

While many choices for an observable could have been made, we have focused on a particular observable: the mean positron density, which in turn gives us the number of pairs produced (on average) in the quantum dynamics.
For simulations optimized for fault-tolerant quantum devices, the number of $T$ gates needed to simulate this quantity within fixed rms error is described in the following corollary; this corollary is the immediate consequence of upper bounds we will give for fault-tolerant simulations involving measurement via sampling or amplitude estimation.

\begin{corollary}[Real-Time Evolution of Mean Pair Density in the Schwinger Model]

\label{cor:trotresultcorb}
Let $\tilde \epsilon > 0$ (rms error tolerance) and $T\in \mathbb{R}$ (total evolution time).
Let $\ket{\psi(T)}$ be the state $\ket{\psi(0)}$ evolved under the Schwinger model Hamiltonian $H$ of \eq{ham} for time $T$.
Denote $\hat N_p$ as the operator that counts all positrons in a state of the lattice.
Then $\mathbb{E}\left( \left(\frac{\mathcal{N} - \bra{\psi(T)} \hat N_p \ket{\psi(T)}}{{N}}\right)^2\right) \leq \tilde \epsilon^2$
can be achieved via direct sampling using a number of $T$-gates that are, for constant $\mu$ and $x$ that is lower bounded by a constant (i.e. $x\in \Omega(1)$), in
$$ O\biggl( \frac{(NT)^{3/2} x^{1/2} \Lambda \log \Lambda }{\tilde \epsilon^{5/2}} \log \biggl (\frac{NT x\Lambda}{\tilde \epsilon} \biggr) \biggr )$$  and a number of qubits that is in $O(N\log \Lambda)$ .

Similarly, the number of $T$-gates needed to simulate the evolution of the mean pair density using amplitude estimation is in
$$
O\left(\frac{ (NT)^{3/2} x^{1/2}\Lambda\log \Lambda }{\tilde \epsilon^{3/2}}\log \left(\frac{ NTx\Lambda}{\tilde \epsilon} \right)\log\left(\frac{1}{\tilde \epsilon}\right)\right)
$$
and requires $O(N\log(\Lambda) + \log(1/\tilde \epsilon))$ qubits.
\end{corollary}

The following section will end with the proof of this corollary.

\subsection{Cost Analysis in Fault-Tolerant Model}\label{sec:upperbounds}

We combine the results of prior sections to determine a sufficient number of $T$-gates to simulate the pair density evolution for arbitrary simulation parameters to $\tilde \epsilon$ precision in the rms error of the mean pair density.
First, we consider estimation via sampling, incorporating parameters $\tau$ and $\kappa$ to encode the distribution of error (in order to optimize this distribution with respect to $T$-gate cost).

\begin{theorem}[Real-Time Evolution of Mean Pair Density via Direct Sampling]
\label{thm:onedtrotresults}
Consider any $\tilde \epsilon > 0$ (rms error tolerance) and $T\in \mathbb{R}$ (total evolution time).
Let $\ket{\psi(T)}$ be the state $\ket{\psi(0)}$ evolved under the Schwinger model Hamiltonian $H$ of \eq{ham} for time $T$.
Denote $\hat N_p$ as the operator that counts all positrons in a state of the lattice.
There exists an operation $W(T)$ that may be implemented on a quantum computer such that, if $\mathcal{N}$ denotes an empirically observed sample mean of the number of positrons in state $W(T)\ket{\psi(0)}$, then 
\begin{equation}
    \mathbb{E}\left( \left(\frac{\mathcal{N} - \bra{\psi(T)} \hat N_p \ket{\psi(T)}}{{N}}\right)^2\right) \leq \tilde \epsilon^2.
\end{equation}
Furthermore, the entire process requires on average at most
\begin{align*}
\frac{  (N T)^{3/2}  \Lambda \log(2\Lambda) x^{1/2} }{2^{1/4} \tilde \epsilon^{5/2}} \ln\left(\frac{2^{9/4} N^{1/2}(6N-5) T^{3/2}  \Lambda x^{1/2} \rho(\tilde \epsilon/\sqrt{8})}{\tilde \epsilon^{3/2}} \right) \\
    \times \gamma(\tilde \epsilon) \rho(\tilde \epsilon/\sqrt{8})\lambda\left(\frac{\tilde \epsilon ^{3/2}}{2^{9/4}N^{1/2} T^{3/2}  \Lambda x^{1/2} \rho(\tilde \epsilon/\sqrt{8})}\right)\nu(\tilde \epsilon,1/2)
 \end{align*}
$T$-gates, where the factors are $\rho$ as in \lem{steps}, $\lambda$ as in \thm{ucost}, $\nu$ as in \lem{epsilonbound}, and
\begin{equation*}
\gamma(\tilde \epsilon) = \frac{\left( \eta + \ln\left(2^{9/4} N^{1/2}(6N-5) T^{3/2}  \Lambda x^{1/2} \rho(\tilde\epsilon/\sqrt{8})/\tilde \epsilon^{3/2} \right) \right)}{ \ln\left(2^{9/4} N^{1/2}(6N-5) T^{3/2}  \Lambda x^{1/2} \rho(\tilde \epsilon/\sqrt{8})/\tilde \epsilon^{3/2} \right)},
\end{equation*}
and $ N(\eta + 1) +4\eta - \lfloor \log \eta \rfloor -1$ qubits are required.
\end{theorem}

\begin{proof}
Assume $0<\kappa<1$, which we include in order to numerically optimize the error distribution.
If we have $ \| W(T) - e^{-iHT} \| \leq \tilde\epsilon\sqrt{ \kappa}$ and choose $N_{shots} = \ceil{1/4\tilde \epsilon^2(1-\kappa)}$, then \lem{epsilonbound} implies the bound on the expectation value.
We therefore take $W(T)=\widetilde V^s(t)$ where we choose $s$ according to \lem{steps}, so that $\| V^s(t) - e^{-iHT} \| \leq \tau \tilde\epsilon\sqrt{ \kappa} $.
Here, $\tau$ is a parameter with $0<\tau < 1$ that expresses the distribution of error between $\delta_{\rm circ}$ and $\delta_{\rm Trot}$.
Assume that $\delta_{\rm circ} =  (1-\tau)\tilde\epsilon \sqrt{ \kappa}/s$.

We get that 
\begin{align}
    \| W(T) - e^{-iHT} \| = \|\widetilde V^s(t) - e^{-iHT} \| &\leq \|\widetilde V^s(t) - V^s(t) \| + \| V^s(t) - e^{-iHT} \| \nonumber \\
    &\leq s \delta_{\rm circ} + \tau \tilde\epsilon\sqrt{ \kappa} = (1- \tau) \tilde \epsilon \sqrt{ \kappa} + \tau \tilde \epsilon\sqrt{ \kappa} = \tilde \epsilon\sqrt{\kappa}.
\end{align}
By \lem{epsilonbound}, this implies that
\begin{equation}
    \mathbb{E}\left( \left(\frac{\mathcal{N} - \bra{\psi(T)} \hat N_p \ket{\psi(T)}}{N}\right )^2  \right) \leq \tilde \epsilon^2.
\end{equation}

\thm{ucost} gives the cost of a single Trotter step as a function of $\delta_{\rm circ}$, $\Lambda = 2^{\eta-1}$, and $N$.
Denote this cost by $C_{\mathrm{Trot}}$.
The total cost of the entire process is thus $C:=s  C_{\mathrm{Trot}} N_{shots}$, which is, with fixed lattice parameters, only a function of the parameters $\kappa$ and $\tau$ that encode the error distribution.

To achieve the cost bounds in the theorem statement, we set $\kappa = \tau = 1/2$.
We are therefore choosing $s$ such that $\| V^s(t) - e^{-iHT} \| \leq  \tilde \epsilon/\sqrt{8} $, i.e.
$$s = \left\lfloor  \frac{N^{1/2} T^{3/2}  \Lambda x^{1/2} \rho(\tilde \epsilon/\sqrt{8})}{(\tilde \epsilon^2 /8)^{1/4}}\right\rfloor$$ 
and $N_{shots} = \floor{\nu(\tilde \epsilon,1/2)/2\tilde \epsilon^2}$.
Furthermore, since we are taking $\delta_{\rm circ} = \tilde \epsilon/(2^{3/2}s)$, we have
$$C_{\mathrm{Trot}} = \left( N\eta^2 + N \eta \ln\left(\frac{6N-5}{\tilde \epsilon/2^{3/2}s} \right) \right) \lambda(\tilde \epsilon/2^{3/2}s).$$
Bounding $C$, we get
\begin{align}
    C &= s C_{\mathrm{Trot}} N_{shots} \nonumber \\
    &\leq  \frac{N^{1/2} T^{3/2}  \Lambda x^{1/2} \rho(\tilde \epsilon/\sqrt{8})}{\tilde \epsilon^{1/2} /8^{1/4}} \cdot \left( N\eta^2 + N \eta \ln\left(\frac{6N-5}{\tilde \epsilon/2^{3/2}s} \right) \right) \lambda(\tilde \epsilon/2^{3/2}s) \cdot \frac{\nu(\tilde \epsilon,1/2)}{2\tilde \epsilon^2} \nonumber \\
    &= \frac{  N^{3/2} T^{3/2}  \Lambda x^{1/2} }{2^{1/4} \tilde \epsilon^{5/2}} \left( \eta^2 + \eta \ln\left(\frac{6N-5}{\tilde \epsilon/2^{3/2}s} \right) \right)  \rho(\tilde \epsilon/\sqrt{8})\lambda(\tilde \epsilon/2^{3/2}s)\nu(\tilde \epsilon,1/2)  \nonumber \\
    &= \frac{  N^{3/2} T^{3/2}  \Lambda \eta x^{1/2} }{2^{1/4} \tilde \epsilon^{5/2}} \ln\left(\frac{2^{9/4} N^{1/2}(6N-5) T^{3/2}  \Lambda x^{1/2} \rho(\tilde \epsilon/\sqrt{8})}{\tilde \epsilon^{3/2}} \right) \nonumber \\
    &\quad \times \gamma(\tilde \epsilon) \rho(\tilde \epsilon/\sqrt{8})\lambda\left(\frac{\tilde \epsilon ^{3/2}}{2^{9/4}N^{1/2} T^{3/2}  \Lambda x^{1/2} \rho(\tilde \epsilon/\sqrt{8})}\right)\nu(\tilde \epsilon,1/2),
 \end{align}
where
\begin{equation}
    \gamma(\tilde \epsilon) = \frac{\left( \eta + \ln\left(2^{9/4} N^{1/2}(6N-5) T^{3/2}  \Lambda x^{1/2} \rho(\tilde \epsilon/\sqrt{8})/\tilde \epsilon^{3/2} \right) \right)}{\ln\left(2^{9/4} N^{1/2}(6N-5) T^{3/2}  \Lambda x^{1/2} \rho(\tilde \epsilon/\sqrt{8})/\tilde \epsilon^{3/2} \right)}.
\end{equation}

The total number of required qubits is given by \thm{ucost}.
\end{proof}

We numerically minimize the total cost, $C(\kappa, \tau)$, over $\kappa$ and $\tau$ for given lattice parameters using Mathematica, reporting the results in \append{costnumeric}.
Unable to produce an analytical minimum that is both lucid and avoids approximation, we have stated the result with $\kappa = \tau = 1/2$ and note that, over the parameter range in \append{costnumeric}, the upper bound on $C$ is no more than twice the optimal T gate bound found numerically.

Next, we consider a simulation incorporating the alternative measurement scheme using amplitude estimation.

\begin{theorem}[Real-Time Evolution of Mean Pair Density via Amplitude Estimation]
\label{thm:onedtrotresults2}
Under the assumptions of~\thm{onedtrotresults}, the condition $
    \mathbb{E}\left( \left(\frac{\mathcal{N} - \bra{\psi(T)} \hat N_p \ket{\psi(T)}}{{N}}\right)^2\right) \leq \tilde \epsilon^2$
can be satisfied using amplitude amplification with a quantum computation that uses a number of $T$-gates bounded above by
$$
 \frac{168\pi^2  x^{1/2} (NT)^{3/2} \Lambda\eta\ln(5/\tilde \epsilon^2)}{\tilde \epsilon^{3/2}\rho^{-1}(\tilde \epsilon/16)\alpha^{-1}(\tilde \epsilon)\zeta^{-1}(\tilde \epsilon)}  \ln \left(\frac{384(Nx)^{1/2}T^{3/2} \Lambda }{\tilde \epsilon^{3/2}\rho^{-1}(\tilde \epsilon/16)}\right) \lambda\left(\frac{\tilde \epsilon^{5/2}\rho^{-1}(\tilde \epsilon/16)}{64 (Nx)^{1/2} T^{3/2} \Lambda }\right)+\mathscr{T}
$$
and at most $N(\eta+1)+3\eta - \lfloor \log(2\eta -1)\rfloor +1 + 2\lceil \log(N/2)\rceil +\left\lceil \log\left(\frac{2\sqrt{2}\pi}{\tilde \epsilon}+4 \right) \right\rceil$ qubits, where
\begin{align*}
\alpha(\tilde \epsilon) &= \left(1 + \frac{2\tilde \epsilon}{\pi} \right),\qquad
\zeta(\tilde \epsilon) = \left(1 +\frac{\eta}{\ln \left(\frac{384 (Nx)^{1/2}T^{3/2} \Lambda \rho(\tilde \epsilon/16)}{\tilde \epsilon^{3/2}}\right)} \right), \ \text{and}\\
\mathscr{T}  &= \left(\frac{32\pi^2(16N+8)\log(N/2)}{(16-\pi^2)\tilde \epsilon}\alpha(\tilde \epsilon)+\frac{7}{3} \log^2\left(\frac{2\sqrt{2}\pi}{\tilde \epsilon}+4 \right)\log\left( \frac{16\log^2\left(\frac{2\sqrt{2}\pi}{\tilde \epsilon}+4 \right)}{\tilde \epsilon^2} \right)\right)\ln\left( \frac{5}{\tilde \epsilon^2}\right).
\end{align*}

\end{theorem}
\begin{proof}
In order to prove the bound on T gates we need to apply the result of~\thm{ae} but provide a concrete quantum circuit implementation of the oracle.
In particular, we require that this simulation satisfy $\|\widetilde{U}_{\psi} -U_\psi\| \le \frac{\tilde \epsilon}{8}$ in total.
This error needs to be distributed over two sources of error: Trotter-Suzuki error and circuit synthesis error.
For simplicity, we choose both sources of error to be at most $\frac{\tilde \epsilon}{16}$.
We then have from~\lem{steps} that the number of Trotter steps $s$ needed to simulate the pair density within an error tolerance of $\tilde \epsilon/16$ is at most
\begin{equation}
    s\le   \frac{4 N^{1/2}T^{3/2} x^{1/2} \Lambda }{\tilde \epsilon^{1/2}}\rho(\tilde \epsilon/16).
\end{equation}
This process needs to be repeated $N_{\rm query}$ times to achieve the target accuracy, where from~\thm{ae} we have that
\begin{align}
    N_{\rm query}&\le \frac{128\pi}{16-\pi^2}\left(\frac{\pi }{\tilde \epsilon}+2 \right)\left(\ln\left( \frac{5}{\tilde \epsilon^2}\right)\right)\nonumber\\
    &\le \frac{21\pi^2 \ln(5/\tilde \epsilon^2)}{\tilde \epsilon}\alpha(\tilde \epsilon).\label{eq:NqueryAE}
\end{align}
Next note that a single query to a controlled Trotter step requires two repetitions of each rotation to implement the controlled version of the operator.
Note that this is in contrast with static simulation wherein doubling the number of rotations is not needed~\cite{wecker2014gate}.
Therefore the total number of Trotter steps in the amplitude amplification protocol obeys
\begin{equation}
    2N_{\rm query} s \le \frac{168\pi^2 (Nx)^{1/2} T^{3/2} \Lambda\ln(5/\tilde \epsilon^2)}{\tilde \epsilon^{3/2}} \rho(\tilde \epsilon/16) \alpha(\tilde \epsilon).\label{eq:trotterRepsAE}
\end{equation}
We then have from~\thm{ucost} that the cost of synthesis with $s$ steps is
\begin{align}
    &N \eta (\eta+\ln(96 s /\tilde \epsilon)) \lambda(\tilde \epsilon/16s)\nonumber\\
    &\qquad\le N\eta\left(\eta + \ln \left(\frac{384(Nx)^{1/2}T^{3/2} \Lambda \rho(\tilde \epsilon/16)}{\tilde \epsilon^{3/2}} \right)\right)\lambda\left(\frac{\tilde \epsilon^{5/2}}{64(Nx)^{1/2} T^{3/2} \Lambda \rho(\tilde \epsilon/16)} \right)\nonumber\\
    &\qquad\le N\eta\ln \left(\frac{384(Nx)^{1/2}T^{3/2} \Lambda \rho(\tilde \epsilon/16)}{\tilde \epsilon^{3/2}}\right) \zeta(\tilde \epsilon)\lambda\left(\frac{\tilde \epsilon^{5/2}}{64(Nx)^{1/2} T^{3/2} \Lambda \rho(\tilde \epsilon/16)}\right).\label{eq:synthcostAE}
\end{align}
Thus the contribution to the total $T$-count from the Trotter-Suzuki operations in the simulation is at most the product of~\eqref{eq:trotterRepsAE} and~\eqref{eq:synthcostAE} which is
\begin{equation}
 \frac{168\pi^2  x^{1/2} (NT)^{3/2} \Lambda\eta\ln(5/\tilde \epsilon^2)}{\tilde \epsilon^{3/2}\rho^{-1}(\tilde \epsilon/16)\alpha^{-1}(\tilde \epsilon)\zeta^{-1}(\tilde \epsilon)}  \ln \left(\frac{384(Nx)^{1/2}T^{3/2} \Lambda }{\tilde \epsilon^{3/2}\rho^{-1}(\tilde \epsilon/16)}\right) \lambda\left(\frac{\tilde \epsilon^{5/2}\rho^{-1}(\tilde \epsilon/16)}{64(Nx)^{1/2} T^{3/2} \Lambda }\right).
\end{equation}
The auxillary $T$-gates needed to perform the Toffoli operations used in the amplitude estimation routine is then simply given by~\eqref{eq:NTAE}.
The result then follows from summing these two results and using the definition of $\alpha(\tilde \epsilon)$ to simplify the result.

The total number of ancillas needed for this procedure will be greater than that required by the sampling-based approach because of the additional qubits needed for the generalized Hadamard test as well as those needed to perform the phase estimation within amplitude estimation.
The number of ancillary qubits required within amplitude estimation is at most
\begin{equation}
    \left\lceil \log(L)\right\rceil =\left\lceil \log\left(\frac{2\sqrt{2}\pi}{\tilde \epsilon}+4 \right) \right\rceil.
\end{equation}
However, the circuit for phase estimation requires adding an additional $\lceil \log(L) \rceil$ controls to the operations in the generalized Hadamard test.
Any qubits used in synthesis of Toffoli gates can be absorbed into those accounted for in other steps.
Thus from~\lem{genHadamard} and~\thm{onedtrotresults} the number of qubits is at most $N(\eta + 1) +4\eta - \lfloor \log \eta \rfloor -1 + 2m +\left\lceil \log\left(\frac{2\sqrt{2}\pi}{\tilde \epsilon}+4 \right) \right\rceil$ as claimed.
\end{proof}

The asyptotic scalings for simulating the mean pair density, stated in \cor{trotresultcorb}, now follow from the above theorems.

\begin{proof}[Proof of Corollary \ref{cor:trotresultcorb}]
By construction, all $\alpha$, $\rho$, $\lambda$, $\gamma$, $\zeta$, $\nu \in \Theta(1)$ under the assumptions, which results in the claimed asymptotic complexities.
\end{proof}

\subsection{Cost Analysis in NEG Model}\label{sec:negsim}
In this section, we would like to determine sufficient resources to simulate real-time dynamics of the mean-pair density in the NEG setting (defined in \sec{NEG}); assume that the user has access to a quantum computer that is equipped with CNOT and single qubit gates wherein the single qubit gates, preparation and measurement are error-free and CNOT can be implemented within diamond distance $\delta_g$ of the ideal CNOT channel. Additionally, as before, $\tilde \epsilon$ is the rms deviation from the true expectation value of the mean positron number per lattice site.

More explicitly, CNOT is the only gate with non-negligible error and, furthermore, if $C_X$ is the ideal channel that carries out the CNOT gate, then we assume that the user has access to $\widetilde{C}_X$ such that $\|C_X - \widetilde{C}_X\|_{\diamond}\le \delta_g$.

\begin{theorem}[Real-Time Evolution of Mean-Pair Density in NEG Setting]
Assume a quantum computer is provided with a universal set of single qubit gates and measurements that are error free.  Further, assume that for any pair of qubits $i,j$ the computer is equipped with a quantum channel $\widetilde C_x$ such that $\|\widetilde C_x - {\rm CNOT}_{i,j}\|_\diamond \le \delta_g$.  We then have that the smallest achievable root-mean-square error in positron number per lattice site evaluated at time $T$---for a Schwinger model simulation in $1+1$D using second-order Trotter formulas---is at most
$$
\frac{3T}{2}\left(\delta_g\Gamma N^{1/2}(N-1)\Lambda x^{1/2}(9\eta^2 -7\eta +34)\right)^{2/3}+\delta_g(N-1)\Gamma,
$$
where $\Gamma = \frac{\biggl( \frac{2}{3}x \Lambda^2 + \biggl(2x^2 + \frac{5}{6}x\mu + \frac{2}{3}x \biggr) \Lambda + \frac{39}{8}x^3 + \frac{25}{12}x^2 \mu + x^2 + \frac{1}{3}x \mu^2 + \frac{5}{12}x\mu + \frac{1}{6}x \biggr)^{1/2}}{\Lambda x^{1/2}}$.  The optimal error target for the error in the Trotter-Suzuki formula is similarly at least
$$\delta_{\rm Trot, opt} = T\left(\delta_g\Gamma N^{1/2}(N-1)\Lambda x^{1/2}(9\eta^2 -7\eta +34)\right)^{2/3}.$$
\end{theorem}
\begin{proof}
Let us take $C_x^p$ to be the $p$-fold composition of the quantum channel $C_x$, which is the ideal $\mathrm{CNOT}_{i,j}$ channel.
We then have from the definition of $\|\cdot\|_{\diamond}$ that for all $p\ge 1$
\begin{equation}
    \|C_x^p\|_{\diamond} = \|\widetilde{C}_x^p\|_{\diamond}=1.
\end{equation}
Let us assume that for some $p\ge 1$ we have that $\|C_x^p-\widetilde{C}_x^p\|_{\diamond}\le p\delta_g$.
We then have from the sub-multiplicativity of the diamond norm that
\begin{align}
    \|C_x^{p+1}-\widetilde{C}_x^{p+1}\|_{\diamond}&\le \|C_x^{p+1}-C_x\widetilde{C}_x^{p}\|_{\diamond}+\|C_x\widetilde{C}_x^{p}-\widetilde{C}_x^{p+1}\|_{\diamond}\nonumber\\
    &\le p\delta_g + \delta_g = (p+1)\delta_g.
\end{align}
Since the result holds by assumption for $p=1$ it follows by induction that it also holds for all $p\ge 1$.
Thus the error from composing the channel is at most additive.

It is shown in~\cite{aharonov1998quantum} that for any operator $M$ and channels $C$ and $\widetilde{C}$
\begin{equation}
    \max_{\rho}\left|{\rm Tr}\left(MC(\rho) \right)-{\rm Tr}\left(M\widetilde{C}(\rho) \right)\right|\le 2 \|M\| \|C-\widetilde{C}\|_{\diamond}.
\end{equation}
In our case, the observable $M$ is the mean positron number per lattice site which is at most $1/2$.
We therefore have that if our simulation circuit contains $N_g$ CNOT gates with errors per gate of $\delta_g$, then 
\begin{equation}\label{eq:Ngchannel}
    \max_{\rho}\left(\frac{1}{N}\left|{\rm{Tr}}(\hat{N}_p\widetilde{C}_{\rm tot}(\rho))-{\rm{Tr}}(\hat{N}_p{C}_{\rm tot}(\rho)) \right|\right)\le  N_g \delta_g,
\end{equation}
where $C_{\rm tot}$ and $\widetilde{C}_{\rm tot}$ denote the ideal total channel and the noisy total channel, respectively.

We can bound $N_g$ in terms of the target Trotterization error using prior results. From \lem{nisqtimestep}, a complete Trotter step for the lattice using the second-order formula costs at most $(N-1)(9\eta^2 -7\eta +34)$ CNOTs.
If the total number of Trotter steps used is at most $s$, then the number of CNOT gates is bounded by
\begin{equation}
    N_g\le s(N-1)(9\eta^2 -7\eta +34)
\end{equation}
We can use \lem{steps} to choose $s$ such that our Trotterization error is at most $\delta_{\rm{Trot}}$. This implies that
\begin{equation}\label{eq:Ng}
N_g\le \frac{N^{1/2}T^{3/2}\Lambda x^{1/2}(N-1)(9\eta^2 -7\eta +34)\rho(\delta_{\rm{Trot}})}{\delta_{\rm{Trot}}^{1/2}}.
\end{equation}
It then follows from~\lem{epsilonbound}, ~\eqref{eq:Ngchannel}, and~\eqref{eq:Ng} that if $\tilde \epsilon$ is the rms error in the true expectation value of the observable, and if we choose $s$ such that the Trotter-Suzuki error is at most $\delta_{\rm Trot}$, then we have
\begin{align}
    \tilde \epsilon^2 \le \left(\frac{\delta_gN^{1/2}(N-1)T^{3/2}\Lambda x^{1/2}(9\eta^2 -7\eta +34)\rho(\delta_{\rm{Trot}})}{\delta_{\rm{Trot}}^{1/2}}+\frac{1}{2} \delta_{\rm Trot}\right)^2+\left(\frac{1}{N_{\rm shots}} \right)^2.
\end{align}
Here, $N_{\rm shots}$ is the number of samples used in the estimate of the mean. The errors add in quadrature as demonstrated prior in the proof of \lem{epsilonbound}. Next let us take $\widetilde\epsilon_{\min}^2:= \lim_{{\rm N_{shots}} \rightarrow \infty} \widetilde\epsilon^2$.
Therefore we have that
\begin{equation}
    \widetilde{\epsilon}_{\min} \le \frac{\delta_gN^{1/2}(N-1)T^{3/2}\Lambda x^{1/2}(9\eta^2 -7\eta +34)\rho(\delta_{\rm{Trot}})}{\delta_{\rm{Trot}}^{1/2}}+\frac{1}{2} \delta_{\rm Trot}
\end{equation}

We then can find an upper bound on the minimal achievable value $\widetilde{\epsilon}_{\min}$ through calculus.  Specifically, we set
\begin{equation}\label{eq:epsmin}
    \frac{\partial}{\partial_{\delta_{\rm Trot}}}\left(\frac{\delta_gN^{1/2}(N-1)T^{3/2}\Lambda x^{1/2}(9\eta^2 -7\eta +34)\rho(\delta_{\rm{Trot}})}{\delta_{\rm{Trot}}^{1/2}}+\frac{1}{2} \delta_{\rm Trot}\right)=0
\end{equation}
Solving for this using computer algebra, we find that if we take the constant part of $\rho$ to be $\Gamma :=\rho(\delta)- \frac{\delta^{1/2}}{N^{1/2} T^{3/2} \Lambda x^{1/2}}$ then the optimal Trotter error to choose is
\begin{equation}\label{eq:deltaopt}
    \delta_{\rm Trot, opt} = T\left(\delta_g\Gamma N^{1/2}(N-1)\Lambda x^{1/2}(9\eta^2 -7\eta +34)\right)^{2/3}.
\end{equation}
Next, substituting~\eqref{eq:deltaopt} into~\eqref{eq:epsmin} yields the least upper bound on $\widetilde\epsilon_{\min}$ is
\begin{equation}
    \frac{3T}{2}\left(\delta_g\Gamma N^{1/2}(N-1)\Lambda x^{1/2}(9\eta^2 -7\eta +34)\right)^{2/3}+\delta_g(N-1)\Gamma
\end{equation}
\end{proof}

\tab{errbd} shows that highly accurate CNOT gates (diamond distance at most $10^{-5}$) will be sufficient to simulate short evolutions in the weak coupling limit.
Even smaller infidelities given by diamond distances of $10^{-7}$ are sufficient to demonstrate a small simulation at strong coupling ($x\le 0.1$).
Owing to the looseness of bounds such as this that have been seen in simulations of fermionic systems~\cite{reiher2017elucidating}, the data in~\tab{errbd} suggests that numerical studies may be needed to assess whether existing quantum computers will be able simulate the Schwinger model in the strong coupling regime or whether further algorithmic optimizations will be needed to do so.

\begin{table}[t!]
\centering
\includegraphics[]{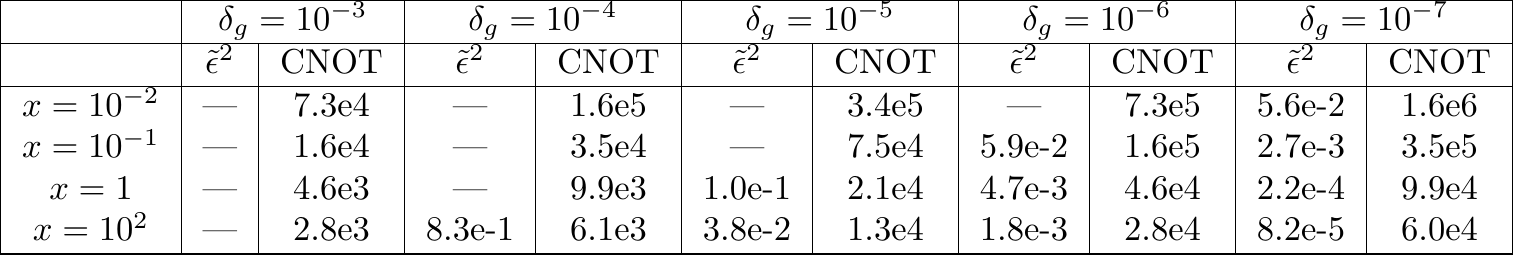}
\caption{Upper bound on the worst-case mean-square error in sampling the mean positron number density ($\tilde \epsilon^2$) in a six qubit Schwinger model simulation, along with a sufficient number of CNOTs per shot taken to ensure that the bias in the estimate is appropriately small. Note that the range of the square of the mean positron density is $[0,1/4]$, and ``---'' represents when the error bound falls outside this range. This is tabulated over $x=(ag)^{-2}$ (where $a$ is the lattice spacing and $g$ is the coupling constant) and the diamond-distance for the CNOT channel ($\delta_g$), with fixed $\eta=N=2$, $T=10/x$ and $\mu=1$.
\label{tab:errbd}}
\end{table}

\section{Incorporating Locality Constraints via the Lieb-Robinson Bound}\label{sec:locality}
It is apparent, by examining (\ref{eq:terms}-\ref{eq:Sumterms}) that the interaction between electron and positron sites is geometrically local with nearest neighbor interactions on the line.
The locality of interactions here actually can potentially yield benefits for simulation as noted in~\cite{haah2018quantum,kliesch2014lieb}.
This advantage stems from the use of Lieb-Robinson bounds which show the existence of an effective light cone for the evolution of local observables.
Observables that have lightcones that do not intersect can be shown to have commutators that are exponentially small and thus Lieb-Robinson bounds provide a mechanism for dividing the time evolution into quasi-independent pieces that can be evolved separately with far less error than would be suggested from bounds that do not exploit these light cones that emerge from the locality of the Hamiltonian.

This locality can be made manifest by rewriting the Hamiltonian as $H=\sum\limits_{\mathcal{X}}h_{\mathcal{X}}+\sum\limits_{\mathcal{Y}}h_{\mathcal{Y}}$ where $\mathcal{X}$ are all two-index sets $\mathcal{X}=\{r,r+1\}; 1\le r\le N-1$ and $\mathcal{Y}$ are all single-index sets $\mathcal{Y}=\{r\}; 1\le r\le N-1$.
Then $h_{\mathcal{X}}=T_r$ and $h_{\mathcal{Y}}=D_r^{(M)}+D_r^{(E)}$.
Lemma 6 in \cite{haah2018quantum} can be applied to $H$, i.e. for any disjoint regions $A, B$, and $C$ on the lattice there are constants $\bar{\mu}, v\ge 0$ such that,
\begin{align}\label{eq:HaahLemma6}
\| U^{H_{A\cup B\cup C}}_{T} - U^{H_{A\cup B}}_{T}(U^{H_{B}}_{T})^{\dagger}U^{H_{B\cup C}}_{T}\| \in O\left(e^{vT-\bar{\mu} {\rm dist}(A,C)}\sum\limits_{\mathcal{Z}:bd(A\cup B,C)}\|h_{\mathcal{Z}}\|\right)\ \ \ ,
\end{align}
where $bd(A\cup B,C)$ denotes the index sets that span the boundary between $A\cup B$ and $C$.
The constant $v$ is often called the Lieb-Robinson velocity and, for the systems with strictly local interactions, it can be estimated as (see Apendix C.6 in \cite{haah2018quantum}) $v\in O(d\sqrt{K})$.
For the Schwinger model the graph degree $d=2$ and $K$ is an upper bound on the commutator $\|[h_{\mathcal{X}},h_{\mathcal{Y}}]\|$.
We can directly compute $K=x(\mu+4\Lambda)$ and thus $v\in O(x^{1/2}(\mu+\Lambda)^{1/2})$.
Also, using the definition of $H_{bd} = H_{A\cup B\cup C} - H_{A\cup B} - H_{C}$ we can compute $\sum\limits_{\mathcal{Z}:bd(A\cup B,C)}\|h_{\mathcal{Z}}\|\le 2x$.
If we denote $\delta$ as the unitary operator approximation error then by using (\ref{eq:HaahLemma6}) we can estimate the partitioning block length $l={\rm dist}(A,C)\in O(\log(\frac{x}{\delta})+vT)$ where we have set $\bar{\mu}=1$.
For the lattice of a given size $N$ we will need to recursively apply $N/l$ partitions with the total unitary decomposition error $O(\delta l/N)$.
After each partitioning step we will be left with a Schwinger sublattice of the size $l$ for which we apply the second order symmetric Trotter decomposition discussed in~\sec{faulttolerant}.
From \lem{steps} we can estimate the number of $T$ gates required to implement a single block of length $l$ on the lattice with $N$ sites $N/l$ times as 
\begin{align}\label{eq:LRcount}
\widetilde{O}\left(\left(\left(\frac{N}{l} \right)\frac{ (lT)^{3/2} x^{1/2}\Lambda }{(\delta l/N)^{1/2}}\right)\right)= \widetilde{O}\left( \frac{(NT)^{3/2} x^{1/2} \Lambda}{\sqrt{\delta}}\right).
\end{align}
On the other hand, the $T$-gate gate count in our Trotter-based algorithm, that does not explicitly use locality of the interaction, can be readily computed by multiplying the number of Trotter steps $s$ in \lem{steps} with the $T$-gate count per step in \thm{ucost}.
Comparing the resulting gate count with that of (\ref{eq:LRcount}) we conclude that the asymptotic scaling of our Trotter-based algorithm and the algorithm of Haah et al. \cite{haah2018quantum} is identical.

So far we have established that there is no explicit asymptotic cost improvement from Lieb-Robinson bounds to the Schwinger model time evolution algorithm.
Recall, however, that ultimately we are interested in estimating the expectation value of local observables such as $\hat{N}_p$.
Can we use the locality of the observable and the Hamiltonian to reduce the computational cost? Intuition suggests that because correlations propagate at a finite velocity $v$ the effect of distant lattice sites on the dynamics of local observables can be sufficiently small beyond certain characteristic radius.
Thus the effective dynamics can be approximated by the Schwinger Hamiltonian defined on a sublattice around a site of interest.
Let us first bound the size of such a sublattice by using the Lieb-Robinson bound for exponentially decaying interactions (see Appendix C.4 Eq.(40) in~\cite{haah2018quantum}).
Consider a local observable $Z_{k}$ with support on a single lattice site $k$.
If $H$ is the Schwinger model Hamiltonian for the entire $N$-site lattice $\Omega$ and $H_{\Omega_{k}}$ is the Schwinger model Hamiltonian defined for a sublattice $\Omega_{k}$ ($|\Omega_{k}|\le N$) around the site $k$ then the following upper bound holds:
\begin{align}\label{eq:HaahEq40}
\| e^{itH}Z_{k}e^{-itH} - e^{itH_{\Omega_{k}}}Z_{k}e^{-itH_{\Omega_{k}}}\| & \le \frac{2\xi|t|}{\sqrt{\eta}}(\exp(\xi |t|\sqrt{8\eta})-1)e^{-{\rm dist}(k,\Omega^{c}_{k})}, 
\end{align}
where $\xi = (8x+\frac{\mu}{2}+\Lambda^2)e$, $\eta=K/(\|h_{\mathcal{X}}\|\|h_{\mathcal{Y}}\|)=(\mu+4\Lambda)/(\mu+2\Lambda^2)$, and $\Omega^{c}_{k}$ is the compliment to the set $\Omega_{k}$ (here we are using equations 15 and 40 in \cite{haah2018quantum}).
Given time $t$ and arbitrary $\delta \ge 0$, choosing the cardinality of $\Omega_{k}$ to be $l =|\Omega_{k}| \ge |\ln \sqrt{\delta} - \ln(\frac{2\xi|t|}{\sqrt{\eta}})- \xi |t|\sqrt{8\eta}|$ guarantees that $\| e^{itH}Z_{k}e^{-itH} - e^{itH_{\Omega_{k}}}Z_{k}e^{-itH_{\Omega_{k}}}\| \le \sqrt{\delta}$.
The knowledge of $l$ enables us to estimate the cost of quantum simulation of the mean number of positrons $\hat{N}_p$ using the techniques of \thm{onedtrotresults}.

Let us denote $W_{k}(t)$ the second order Trotter-Suzuki approximation of $e^{-itH_{\Omega_{k}}}$ and introduce the following expectation values
\begin{eqnarray}
\mathbb{E}_{\psi}(\hat{N}_{p})  & = & \langle \psi_{0}|e^{itH}\hat{N}_{p}e^{-itH}|\psi_{0} \rangle, \label{eq:ExpctVal1} \\
\mathbb{E}_{\phi}(\hat{N}_{p})  & = & \langle \psi_{0}|\sum\limits_{k=1}^{N/2}e^{itH_{\Omega_{k}}}Z_{2(k-1)}e^{-itH_{\Omega_{k}}}|\psi_{0} \rangle, \label{eq:ExpctVal2} \\
\mathbb{E}_{\chi}(\hat{N}_{p}) & = & \langle \psi_{0}|\sum\limits_{k=1}^{N/2}W_{k}(t)^{\dagger}Z_{2(k-1)}W_{k}(t)|\psi_{0} \rangle, \label{eq:ExpctVal3}
\end{eqnarray}
where $|\psi_0\rangle$ is an arbitrary initial state.
We can estimate the number of shots $N_{shots}$ needed to bound the rms error in computing the expectation value of $\hat{N}_{p}$ to $\delta$ using \lem{epsilonbound}.
The key insight here is that
\begin{eqnarray}\label{eq:MeanError}
| \mathbb{E}_{\psi}(\hat{N}_{p}) - \mathbb{E}_{\chi}(\hat{N}_{p}) |  & = & | \mathbb{E}_{\psi}(\hat{N}_{p}) - \mathbb{E}_{\phi}(\hat{N}_{p}) + \mathbb{E}_{\phi}(\hat{N}_{p}) - \mathbb{E}_{\chi}(\hat{N}_{p}) | \nonumber \\
 & \le & \|\mathbb{E}_{\psi}(\hat{N}_{p}) - \mathbb{E}_{\phi}(\hat{N}_{p})\| + \|\mathbb{E}_{\phi}(\hat{N}_{p}) - \mathbb{E}_{\chi}(\hat{N}_{p})\| \nonumber \\
 & \le & \frac{N\sqrt{\delta}}{2} + \frac{N\sqrt{\delta}}{2} = N\sqrt{\delta}, 
\end{eqnarray}
where we used (\ref{eq:HaahEq40}) (setting $|\Omega_{k}| = |\ln \sqrt{\delta} - \ln(\frac{2\xi|t|}{\sqrt{\eta}})- \xi |t|\sqrt{8\eta}|$ for all $k$ ) to upperbound  $\|\mathbb{E}_{\psi}(\hat{N}_{p}) - \mathbb{E}_{\phi}(\hat{N}_{p})\|$ and the second order Trotter-Suzuki bound for $\|\mathbb{E}_{\phi}(\hat{N}_{p}) - \mathbb{E}_{\chi}(\hat{N}_{p})\|$.
Since the variance $\mathbb{V}(\mathcal{N})$ of the empirically observed sample mean $\mathcal{N}$ can still be bounded by $\frac{N^2}{4N_{shots}}$ we conclude that the number of samples $N_{shots}$ we need to take to satisfy $ \mathbb{E}\left(\left(\frac{\mathcal{N}-\bra{\psi}e^{itH}\hat{N}_{p} e^{-itH}\ket{\psi_{0}}}{N}\right)^2\right) \in O(\tilde \epsilon^2), $ is given by \lem{epsilonbound} with $\tilde \epsilon^2\kappa\in O(\delta)$.
Therefore, we can use \thm{onedtrotresults} to estimate the simulation cost of $\hat{N}_{p}$ by replacing the lattice size $N$ with $N_{\rm eff}=\left\lceil|\ln \tilde \epsilon - \ln(\frac{2\xi T}{\sqrt{\eta}})- \xi T\sqrt{8\eta}| \right\rceil$.
Thus we have from~\cor{trotresultcorb} that the number of $T$-gates needed to perform the simulation using amplitude amplification is in
\begin{equation}
    \widetilde{O}\left(\frac{ (N_{\rm eff}T)^{3/2} x^{1/2}\Lambda}{\tilde \epsilon^{3/2}}\right)=\widetilde{O}\left(\frac{ (\xi \sqrt{\eta})^{3/2}T^3 x^{1/2}\Lambda}{\tilde \epsilon^{3/2}}\right)
\end{equation}

We note that a particularly interesting simulations regime is the one for very large lattice sizes $N$ and fixed cutoff $\Lambda$.
In this regime the Lieb-Robinson argument would allow to significantly improve the cost of simulation by removing the dependence on $N$ at the expense of increasing the dependence on $T$ quadratically.
Recent work has further shown that if high-order product formulas are used in the place then the temporal scaling of high-order Trotter-Suzuki approximations to $e^{-iHt}$ can be improved from $\widetilde{O}(T^{3}/\delta_{\rm Trot}^{1/2+o(1)})$ to $T^{2}(T/\delta_{\rm Trot})^{o(1)}$~\cite{Yuan}, where $(T/\delta_{\rm Trot})^{o(1)}$ is the sub-polynomial contribution to the scaling that comes from adaptively choosing an increasingly high-order Trotter formula as $T$ increases.
While the nested commutators in such an expansion are harder to evaluate to determine the scaling of the cost with $\zeta, \eta, x$ and $\Lambda$, this other work nonetheless shows that in some regimes our approach can be improved by going to higher order Trotter formulas.

\section{Conclusion}
% Deleted "the first" from "the first explicit quantum..." below, as according to the second reviewer "explicit" is a somewhat vague term (ironically) and I'm fine downplaying the claim this late into the paper. - A.S.
{The foremost contribution of this article is having detailed explicit quantum algorithms to simulate the lattice Schwinger model on a quantum computer and proven them to be scalable in the computational models considered.}
As the Schwinger model is a testbed for the study of gauge field theories underpinning the Standard Model, the algorithms of this paper are an important stepping stone towards simulating quantum electrodynamics in higher dimensions as well as more general quantum field theories both in the NISQ era and beyond.

Further significance of this work is seen in its comparison to state-of-the-art quantum simulation algorithms.
As discussed in \sec{compare}, our Trotter-Suzuki-based algorithms scale only as $\tilde O (T^{3/2} \Lambda x^{1/2})$.
This is quadratically better scaling with the electric field cutoff $\Lambda$ than would be expected from other simulation algorithms such as qubitization or QDRIFT.

We also analyze simulation of the Schwinger model in a simplified cost model that is appropriate for near-term simulations wherein qubits and entangling gates are regarded as the costly resources.
A similar $O(T^{3/2} \Lambda x^{1/2})$ scaling for the number of two-qubit gates in this setting has been shown, while also calling for only one ancillary qubit.
This scaling by itself, while technically true, neglects the effects of gate error on the simulation;
we address this by providing an analysis of the minimal achievable root-mean-square error as a function of the diamond distance between the two-qubit gates and the target, and providing sufficient conditions on the diamond norm to allow a given simulation.
These contributions are significant since it allows subsequent NISQ-era algorithms to be explicitly compared to our approach.
Subsequent optimizations can therefore rigorously quantify their improvements over the initial simulation strategies set forth herein.

There are several important avenues of investigation to pursue beyond what we have considered.
While we chose Trotter-Suzuki formulas in our simulation algorithms because of their potential advantages over other state-of-the-art methods, it is possible that explicit implementation of these other algorithms may yet yield unforeseen benefit, so we leave it to subsequent work to examine their performance and compare the effectiveness of each algorithm in simulating the lattice Schwinger model.
Analysis of qubitization, the interaction picture, and multi-product formulas, as well as numerical studies of the Trotter error, will be of particular interest for subsequent studies.
In actual quantum simulations of the Schwinger model, one might also wish to add in a background field or periodic conditions, the algorithms for both of these being straightforward extensions of what has been introduced in this article.

Importantly, while this work has focused on establishing rigorous upper bounds on the resources needed in both NISQ-era quantum computing and beyond, the bounds may prove to be loose compared to the empirically-observed resources needed.
A further numerical study would be useful to understand the gaps between our bounds and the actual performance in small-scale simulations.

Perhaps the ultimate goal of subsequent work is to generalize the analysis in this paper to the simulation of gauge field theories in higher dimensions and for non-Abelian gauge groups.
These innovations will allow us to not only simulate quantum electrodynamics in higher dimensions, but will also provide valuable insights into how to simulate quantum field theories that better represent what we know about Nature on quantum computers.

\section{Acknowledgements}
The authors would like to thank Natalie Klco for conceptual discussions and contributions to Sections 2 and the circuits of Section 3 of this work and M. Savage for insightful discussions about the Schwinger model.
NW would like to thank Y. Su, A. Childs, M. Tran and S. Zhu for useful discussions surrounding product formulas and Lieb-Robinson bounds.

This work was performed in part at Oak Ridge National Laboratory, operated by UT-Battelle for the U.S. Department of Energy under Contract No. DEAC05-00OR22725.
AS and PL are supported by the U.S. Department of Energy, Office of Science, Office of Advanced Scientific Computing Research (ASCR) Quantum Algorithm Teams, Quantum Computing Application Teams, and Accelerated Research in Quantum Computing programs, under field work proposal numbers ERKJ333, ERKJ347, and ERKJ354.
JRS was supported by DOE Grant No. DE-FG02-00ER41132, and by the National Science Foundation Graduate Research Fellowship under Grant No. 1256082.
NW is supported by the Google Quantum Research Award, the Embedding Quantum Computing into Many-body Frameworks for Strongly Correlated Molecular and Materials Project which is funded by the US department of Energy, Office of Science, Office of Basic Energy Sciences, the Division of Chemical Sciences, Geosciences and Biosciences as well as the QUASAR initiative at Pacific Northwest National Laboratory. AS was further supported by the US Department of Energy Office of Science Award DE-SC0020271.

\section{Notations}\label{sec:notation}
\begin{table}[H] \label{tab:vartable}
\centering
\includegraphics[]{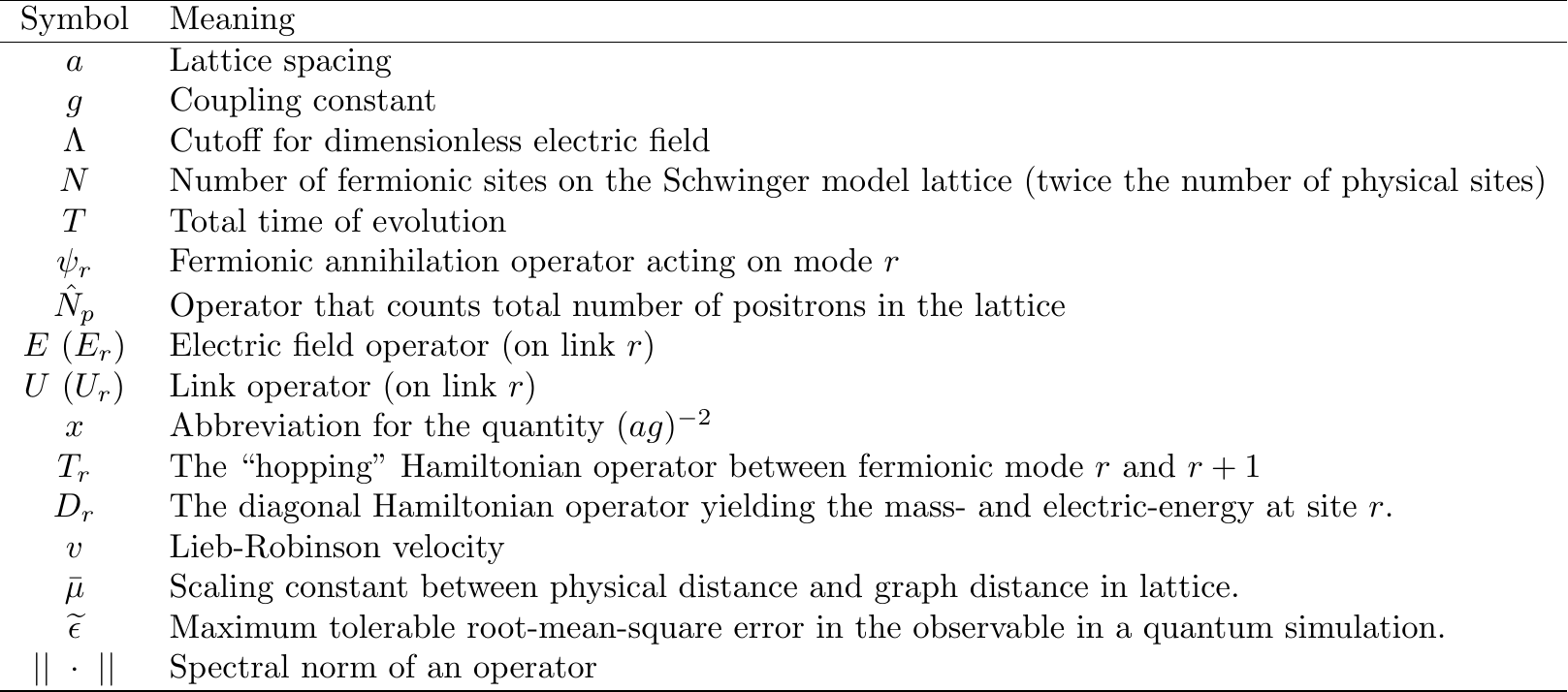}
\caption{Summary of important variables used in this paper.}
\end{table}

\bibliographystyle{plainnat}
\bibliography{biblioDOI}

\appendix
\section{Computing Commutators for Second-Order Error Bound}
\label{append:commutators}

With the ordering of \eq{sumorder}, through examining cases and applying the mentioned simplifications it is straightforward to show the following implication of \eq{seconderror}.
To aide in understanding, the large parentheses indicate which of the two sums we are plugging our ordering into.
Recall that $T_r^{(i)}$ are given in terms of the block-diagonalized matrices $A_r$, $\tilde{A}_r$, $B$, and $\tilde{B}$ in \sec{kineticimplementRC}.
And for $r \geq N$, $T_r^{(j)} = 0$ while $D_N = D_N^{(M)}$ and $D_r = 0$ for $r \geq N+1$.
To avoid superfluous case analysis at the lattice boundary, we will ignore these boundary conditions and overcount slightly when $r = N-2, N-1$ or $N$.
\begin{align}
    \delta_{\rm Trot} &\leq  t^3 \sum_{r=1}^{N} \biggl ( \frac{1}{12}\biggl (  \sum_{i=1}^4 \| [[D_r,T_r^{(i)}], D_r] \|   +  \sum_{i=1}^4 \| [[T_r^{(i)}, D_{r+1}], T_r^{(i)}] \|  \nonumber \\
    &\quad + \sum_{i=1}^4 \sum_{j>i}^4 \| [[T_r^{(i)}, T_{r}^{(j)}], T_r^{(i)}] \|  + \sum_{i=1}^4 \sum_{j=1}^4 \| [[T_r^{(i)}, T_{r+1}^{(j)}], T_r^{(i)}] \|  \biggr ) \nonumber \\  
    &\quad + \frac{1}{24}\biggl ( \sum_{i=1}^4 \sum_{j=1}^4\| [[D_r,T_r^{(i)}],T_{r}^{(j)}] \|  + \sum_{i=1}^4  \| [[D_r,T_r^{(i)}],D_{r+1}] \|  \nonumber \\
    &\quad + \sum_{i=1}^4 \sum_{j=1}^{4} \| [[D_r,T_r^{(i)}], T_{r+1}^{(j)}] \|  + \sum_{i=1}^4 \sum_{j>i}^4 \sum_{k>i}^4 \| [[T_r^{(i)},T_r^{(j)}],T_r^{(k)}] \| \nonumber \\
    &\quad + \sum_{i=1}^4 \sum_{j>i}^4 \| [[T_r^{(i)},T_r^{(j)}],D_{r+1}] \|  + \sum_{i=1}^4 \sum_{j>i}^4 \sum_{k=1}^4 \| [[T_r^{(i)},T_r^{(j)}],T_{r+1}^{(k)}] \| \nonumber \\
    &\quad +  \sum_{i=1}^4 \sum_{j>i}^4 \| [[T_r^{(i)},D_{r+1}],T_{r}^{(j)}] \|  +   \sum_{i=1}^4 \| [[T_r^{(i)},D_{r+1}],D_{r+1}]\|   \nonumber \\
    &\quad +  \sum_{i=1}^4 \sum_{j=1}^4 \| [[T_r^{(i)},D_{r+1}],T_{r+1}^{(j)}] \| + \sum_{i=1}^4 \sum_{j=1}^4 \sum_{k>i}^4 \| [[T_r^{(i)},T_{r+1}^{(j)}],T_r^{(k)}] \|  \nonumber \\
    &\quad + \sum_{i=1}^4 \sum_{j=1}^4 \| [[T_r^{(i)},T_{r+1}^{(j)}],D_{r+1}] \| + \sum_{i=1}^4 \sum_{j=1}^4 \sum_{k=1}^4 \| [[T_r^{(i)},T_{r+1}^{(j)}],T_{r+1}^{(k)}] \| \nonumber \\
    &\quad + \sum_{i=1}^4 \sum_{j=1}^4 \| [[T_r^{(i)},T_{r+1}^{(j)}],D_{r+2}] \|+ \sum_{i=1}^4 \sum_{j=1}^4 \sum_{k=1}^4 \| [[T_r^{(i)},T_{r+1}^{(j)}],T_{r+2}^{(k)}] \| \biggr ) \biggr ).
    \label{eq:bigcombound}
\end{align}

We bound each commutator using the lemma below.
Regarding influence in scaling, the major culprit is the first term in the sum above, for within it alone appear terms like $\| [[D_r^{(E)}, T_r^{(1)}], D_r^{(E)}] \| = \| [[E^2, x (A\otimes G) /4],E^2] \|$, which may be crudely bounded by $x\|E^2\|^2 \|A\otimes G\| = 2 x\Lambda^4$ through matrix norm inequalities.
The lemma below includes a reduction of this bound to one quadratic in $\Lambda$.
Furthermore, it gives a bound on each commutator norm above.

\begin{lemma}
Let $T_r^{(i)}$ and $D_r$ be defined as in Section \ref{sec:faulttolerant}, with $\| \cdot \|$ the spectral norm.
The following statements are true for all lattice sites $r,s,t$ and indices $i,j, k$,
\begin{enumerate}
    \item $\| [[D_r, T_r^{(i)}], D_r] \| \leq  x(2\Lambda^2+ (2+2\mu) \Lambda +\frac{\mu^2}{2}+\mu+ \frac{1}{2}) $
    \item If $s\neq r$, then  $\| [[T_r^{(i)},D_s], T_t^{(j)}] \| \leq \frac{x^2\mu}{2} $
    \item If $s\neq r, t$, then $\| [[T_r^{(i)}, T_t^{(j)}], D_s] \| \leq \frac{x^2 \mu}{2}$ 
    \item $\| [[T_r^{(i)}, T_{s}^{(j)}], T_t^{(k)}] \| \leq \frac{x^3}{2}$ 
    \item $\| [[D_r,T_r^{(i)}],T_{s}^{(j)}] \| \leq x^2(\frac{\mu}{2} + \Lambda + \frac{1}{2}) $
    \item $\| [[D_r,T_r^{(i)}],D_{r+1}] \| \leq x\mu(\frac{\mu}{2} + \Lambda+\frac{1}{2})$
    \item $\| [[T_r^{(i)},D_{r+1}],D_{r+1}]\| \leq \frac{x\mu^2}{2} $
    \item $\| [[T_r^{(i)},T_{r+1}^{(j)}],D_{r+1}] \|  \leq x^2 (\frac{\mu}{2} + \Lambda +\frac{1}{2})$
\end{enumerate}
\label{lem:combound}
\end{lemma}
\begin{proof}
First note that from the sub-multiplicative property of the spectral norm and the triangle inequality we have that
\begin{equation}
\label{eq:roughbound}
    \| [[A,B],C] \| \leq 4\|A\| \|B\| \|C\|,
\end{equation}
which will be key to demonstrating many of the above claims.
\begin{enumerate}

    \item Initially, we have
    \begin{equation}\label{eq:lemsum1}
         \|[[D_r,T_r^{(j)}],D_r] \| \leq \sum_{i,k \in \{ M, E \} } \| [[D_r^{(i)}, T_r^{(j)}],D_r^{(k)} ]\|.
    \end{equation}
    We will first focus on the case where $i=k=2$, i.e. bounding $\| [[E_r^2, T_r^{(j)}], E_r^2] \|$. We determine this commutator bound by understanding how the commutator acts on a basis state $\ket{n}\otimes \ket{m}$ where $\ket{n} \in \mathcal{H}_g$, where $\mathcal{H}_g$ is the gauge Hilbert space, and $\ket{m}\in \mathcal{H}_f$, where $\mathcal{H}_f$ is the fermion Hilbert space, corresponding to the $r$th lattice site.
    
    \indent
    First, see that the operation of $T_r^{(j)}$ on a computational basis state in the fermionic space will produce a multiple of another basis state for all $j$. Explicitly, $T_r^{(j)}$ takes $\Ket{n} \to c\ket{n \pm 1}$, where the arithmetic runs $-\Lambda \to \Lambda - 1$. The sign and $c$ depend on $j$, and are ultimately inconsequential in bounding the spectral norm. This can be quickly verified using the form of $A,\ \widetilde{A},\ B$ and $\widetilde{B}$.
    
    Therefore, we have that
    \begin{align}
        [E_r^2 , T_r^{(j)}] \Ket{n}\otimes \Ket{m} &= (E_r^2 - n^2) T_r^{(j)} \Ket{n}\otimes \Ket{m} \nonumber \\
        &\to ((n\pm 1)^2 - n^2) T_r^{(j)} \Ket{n}\otimes \Ket{m} \nonumber \\
        &= (\pm 2n +1) T_r^{(j)} \Ket{n}\otimes \Ket{m},
    \end{align}
    implying that $\|[E_r^2 , T_r^{(j)}]\| \leq (2\Lambda + 1)\|T_r^{(j)}\| \leq x(\Lambda + 1/2)$, since $\|T_r^{(j)}\| = x/2$. (Note that this holds for all $n$, even at the cutoffs $-\Lambda$ and $\Lambda - 1$.)

We similarly have that
\begin{align}
    [[E_r^2, T_r^{(j)}],E_r^2] \Ket{n} \otimes \Ket{m} &= ( n^2 - E_r^2) [E_r^2, T_r^{(j)}] \Ket{n} \otimes \Ket{m} \nonumber \\
    &\to (n^2 - (n\pm 1)^2) [E_r^2, T_r^{(j)}] \Ket{n} \otimes \Ket{m} \nonumber \\
    &= (\mp 2n -1)[E_r^2, T_r^{(j)}] \Ket{n} \otimes \Ket{m}.
\end{align}
This implies
\begin{equation}
    \| [[E_r^2, T_r^{(j)}], E_r^2] \| \leq (2\Lambda +1) \|[E_r^2, T_r^{(j)}]\| \leq 2x(\Lambda + \frac{1}{2})^2 = 2x(\Lambda^2 + \Lambda + \frac{1}{4}).
\end{equation}
    
    Next, bounding the term in the sum \eq{lemsum1} with $i=1$ and $k=2$,
    \begin{align}
    [[\frac{\mu}{2} Z_r, T_r^{(j)}],E^2_r]  \Ket{n} \otimes \ket{m} &= (n^2 - E^2) [\frac{\mu}{2} Z_r, T_r^{(j)}] \Ket{n}\otimes\Ket{m} \nonumber \\
    &\to (n^2 - (n\pm 1)^2) [\frac{\mu}{2} Z_r, T_r^{(j)}] \Ket{n}\otimes\Ket{m} \nonumber \\
    &= (\mp 2n -1) [\frac{\mu}{2} Z_r, T_r^{(j)}] \Ket{n}\otimes\Ket{m},
    \end{align}
    implying that $\|[[\mu Z_r/2, T_r^{(j)}],E_r^2]\| \leq x\mu (\Lambda+1/2)$. Therefore, we get
    \begin{align}
         \|[[D_r,T_r^{(j)}],D_r] \| &\leq \| [[E_r^2, T_r^{(j)}], E_r^2] \| + \| [[ \frac{\mu}{2} Z_r, T_r^{(j)}],\frac{\mu}{2} Z_r] \| \nonumber \\
         &\quad + \| [[E^2,T_r^{(j)}], \frac{\mu}{2} Z_r] \| + \| [[\frac{\mu}{2} Z_r ,T_r^{(j)}], E^2] \| \nonumber \\
         &\leq 2x(\Lambda^2 + \Lambda + \frac{1}{4}) + \frac{1}{2}x\mu^2 + x\mu(\Lambda + \frac{1}{2}) + x\mu(\Lambda+\frac{1}{2}) \nonumber \\
         &\leq x(2\Lambda^2+ (2+2\mu) \Lambda +\frac{\mu^2}{2}+\mu+ \frac{1}{2}).
    \end{align}

    \item The assumption ensures that the $D_{s}^{(E)}$ summand of $D_{s}$ will disappear in the internal commutator. Thus applying \eq{roughbound} gives
    \begin{equation}
       \| [[T_r^{(i)},D_s] T_t^{(j)}] \| \leq 4\| T_r^{(i)}\| \|\mu Z_{s}/2\| \| T_t^{(j)} \| = x^2\mu/2.
    \end{equation}
    \item The assumption ensures that the $D_{s}^{(E)}$ summand of $D_{s}$ will disappear in the external commutator, giving the same bound as No. 2.
    
    \item Using \eq{roughbound} and the fact that $\| T_r^{(j)} \| = x/2$ for all $i, j$, the result is immediate.

    \item The proof of No. 1 showed that $\| [D_r^{(E)}, T_r^{(j)}] \| \leq  x(\Lambda + 1/2 )$, and so
    \begin{equation}
        \| [[D_r,T_r^{(j)}],T_{s}^{(k)}] \| \leq 2 (\|[D_r^{(M)}, T_r^{(j)}] \| + \|[D_r^{(E)}, T_r^{(j)}] \|)\|T_s^{(k)}\| \leq x^2(\frac{\mu}{2} + \Lambda + \frac{1}{2}).
    \end{equation}
    \item The $D_{r+1}^{(E)}$ term commutes, so we may ignore them and apply similar reasoning to that of No. 5:
    \begin{align}
        \| [[D_r,T_r^{(j)}],D_{r+1}] \| &\leq  \| [[\frac{\mu}{2} Z_r ,T_r^{(j)}], \frac{\mu}{2} Z_{r+1}] \| +\| [[ D_r^{(E)} ,T_r^{(j)}], \frac{\mu}{2} Z_{r+1}] \| \nonumber \\
        &\leq x(\frac{\mu^2}{2} + \mu(\Lambda+\frac{1}{2})).
    \end{align}

    \item The first $D_{r+1}^{(E)}$ term commutes, giving
    \begin{align}
        \| [[T_r^{(j)},D_{r+1}],D_{r+1}]\| &= \| [[T_r^{(j)},D_{r+1}^{(M)}],D_{r+1}]\| \\ \nonumber
        &\leq \| [[T_r^{(j)},\frac{\mu}{2} Z_{r+1}],\frac{\mu}{2} Z_{r+1}]\| + \| [[T_r^{(j)},\frac{\mu}{2} Z_{r+1}], E^2_{r+1}]\|.
    \end{align}
Since $T_r^{(j)}$ does nothing to the gauge register on the $(r+1)$th link, we have that
    \begin{align}
        [[T_r^{(j)},\frac{\mu}{2} Z_{r+1}], E^2_{r+1}] \Ket{n}_{r+1} \Ket{m}_{r+1} &= (E^2_{r+1} - n^2)[T_r^{(j)},\frac{\mu}{2} Z_{r+1}] \Ket{n}_{r+1} \Ket{m}_{r+1} \nonumber \\
        &= (n^2 - n^2)[T_r^{(j)},\frac{\mu}{2} Z_{r+1}] \Ket{n}_{r+1} \Ket{m}_{r+1} = 0,
    \end{align}
    implying that 
    \begin{equation}
         \| [[T_r^{(j)},D_{r+1}],D_{r+1}]\| \leq \| [[T_r^{(j)},\frac{\mu}{2} Z_{r+1}],\frac{\mu}{2} Z_{r+1}]\| \leq x\mu^2/2.
    \end{equation}

    \item Using \eq{roughbound} and prior reasoning,
    \begin{align}
        [[T_r^{(i)},T_{r+1}^{(j)}],E^2_{r+1}] \Ket{n}_{r+1} \Ket{m}_{r+1} &= (n^2 - E^2_{r+1}) [T_r^{(i)},T_{r+1}^{(j)}] \Ket{n}_{r+1} \Ket{m}_{r+1} \\ \nonumber 
        &\to (n^2 - (n\pm 1)^2) [T_r^{(i)},T_{r+1}^{(j)}] \Ket{n}_{r+1} \Ket{m}_{r+1} \\ \nonumber 
        &= (\mp 2n -  1) [T_r^{(i)},T_{r+1}^{(j)}] \Ket{n}_{r+1} \Ket{m}_{r+1} 
    \end{align}
    implying that
    \begin{align}
        \| [[T_r^{(i)},T_{r+1}^{(j)}],D_{r+1}] \| &\leq \| [[T_r^{(i)},T_{r+1}^{(j)}],E^2_{r+1}] \| + \| [[T_r^{(i)},T_{r+1}^{(j)}],\frac{\mu}{2} Z_{r+1}] \| \\ \nonumber 
        &\leq x^2 (\Lambda +\frac{1}{2}) + \frac{x^2 \mu}{2}.
    \end{align}

    \end{enumerate}
\end{proof}

We compile all these bounds to evaluate \eq{bigcombound} in the following Corollary,

\begin{corollary}
\label{cor:errboundcalced}
Let $t \in \mathbb{R}$, $V(t)$ be defined as in \eq{timestep}, and $H$ be the Schwinger model Hamiltonian as defined in \sec{hamrep}.
If  $\delta_{\rm Trot} = \| V(t)-e^{-iHt}\|$ (where $\|\cdot \|$ is the spectral norm), then
\begin{equation}
    \delta_{\rm Trot} \leq  Nt^3 \biggl( \frac{2}{3}x \Lambda^2 + \biggl(2x^2 + \frac{5}{6}x\mu + \frac{2}{3}x \biggr) \Lambda + \frac{39}{8}x^3 + \frac{25}{12}x^2 \mu + x^2 + \frac{1}{3}x \mu^2 + \frac{5}{12}x\mu + \frac{1}{6}x \biggr) 
\end{equation}
\end{corollary}
\begin{proof}
We plug in the results of Lemma \ref{lem:combound} to Equation \ref{eq:bigcombound}, and the case of the lemma that each summand satisfies is quick to verify from the summands' definitions.
We get that
\begin{alignat}
    \epsilon \delta_{\rm Trot} \leq  N t^3\biggl (& \frac{1}{12}\biggl( 4 x(2\Lambda^2+ (2+2\mu) \Lambda +\frac{\mu^2}{2}+\mu+ \frac{1}{2}) &&+ 4\frac{x^2\mu}{2} \nonumber\\
    &+ 6x^3/2 &&+ 16x^3/2 \biggr ) \nonumber \\ 
    &+ \frac{1}{24}\biggl(16x^2(\frac{\mu}{2} + \Lambda + \frac{1}{2}) &&+ 4x\mu(\frac{\mu}{2} + \Lambda+\frac{1}{2}) \nonumber\\
    &+ 16x^2(\frac{\mu}{2} + \Lambda + \frac{1}{2}) &&+ 14x^3/2 \nonumber \\ 
    &+ 6 \frac{x^2\mu}{2} &&+ 24x^3/2 \nonumber\\
     &+ 6\frac{x^2\mu}{2} &&+ 4\frac{x\mu^2}{2} \nonumber\\
    &+ 16\frac{x^2\mu}{2} &&+ 24x^3/2 \nonumber \\
    &+ 16x^2 (\frac{\mu}{2} + \Lambda +\frac{1}{2}) &&+ 64x^3/2 \nonumber\\
    &+ 16\frac{x^2\mu}{2} &&+64x^3/2 \biggr )  \biggr )\\
    &= Nt^3 \biggl( \frac{2}{3}x \Lambda^2 + \biggl(2x^2 + \frac{5}{6}x\mu + \frac{2}{3}x \biggr) \Lambda + \frac{39}{8}&&x^3 + \frac{25}{12}x^2 \mu + x^2 + \frac{1}{3}x \mu^2 + \frac{5}{12}x\mu + \frac{1}{6}x \biggr)
\end{alignat}

\end{proof}

\section{Numerical Evaluation and Analysis of $T$-Count Upper Bounds}\label{append:costnumeric}
This appendix reports the upper bound of ``Sampling," \thm{onedtrotresults}, and the (unoptimized) upper bound from ``Estimating," \thm{onedtrotresults2}, evaluated over the indicated lattice parameters.
The cost bound of \thm{onedtrotresults} is numerically minimized with respect to $\tau$ and $\kappa$ using \textit{Mathematica}.

We express the time of the following simulation scenarios in terms of $1/x$ since the dimensionless quantity $xT$ describes the coupling that emerges between subsystems in time-dependent perturbation theory.
This allows comparison to the capabilities of classical algorithms, which are anticipated to not have capabilities that extend far into the ``weak coupling'' and ``long time'' regimes.

\begin{table}
\begin{center}
\includegraphics[]{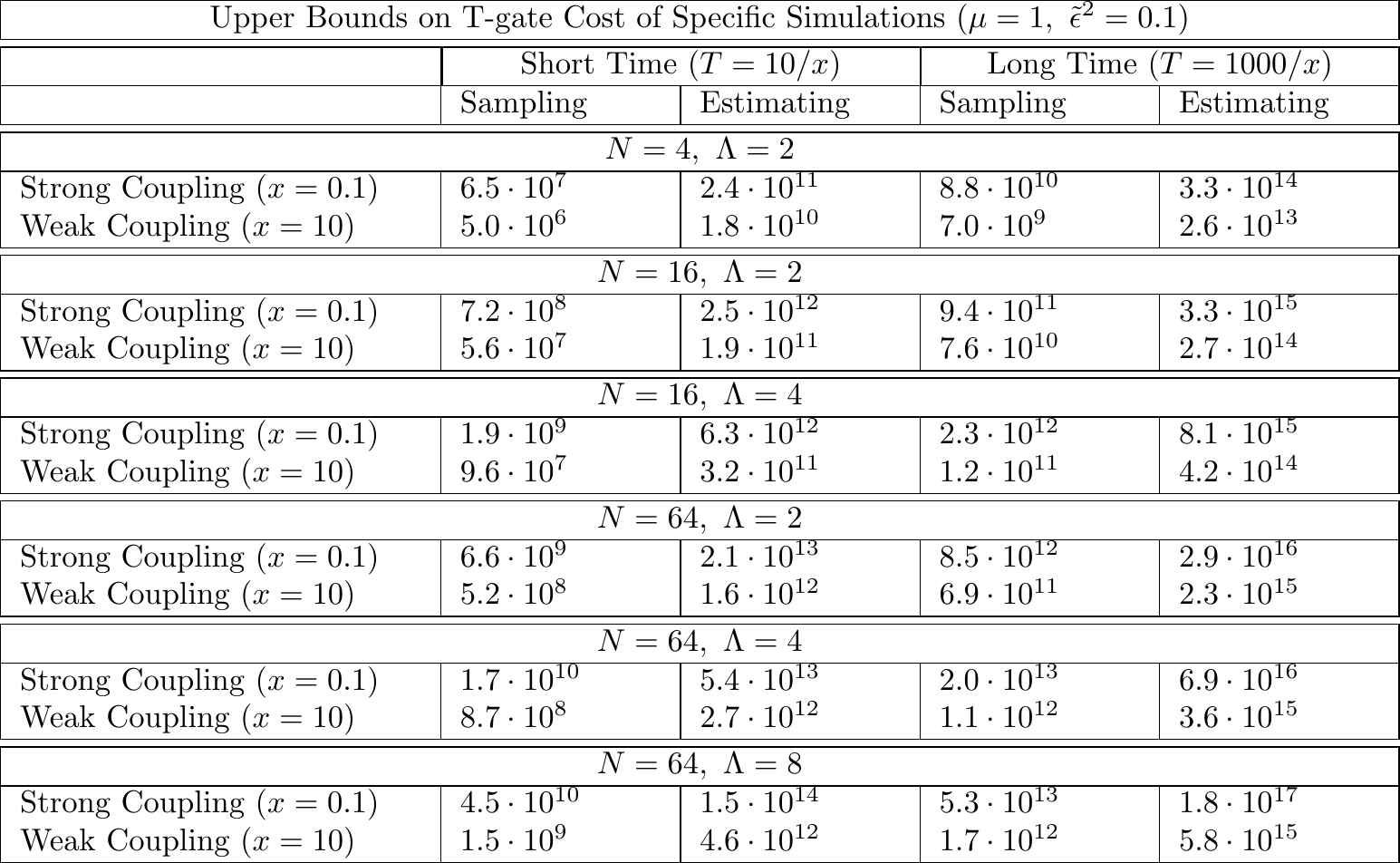}
\end{center}
\caption{The $T$-gate upper bound on the cost of simulating the lattice Schwinger Model using different measurement schemes: ``Sampling," from \thm{onedtrotresults}, and the (unoptimized) upper bound of ``Estimating," from \thm{onedtrotresults2}, evaluated over the indicated lattice parameters.
The cost bound of \thm{onedtrotresults} is numerically minimized with respect to $\tau$ and $\kappa$ using \textit{Mathematica}.}

\end{table}

In speculating where these two upper bounds may cross over, we plot the ratio ``Estimating/Sampling'' of the upper bounds over varied parameters in \fig{ratio}, using the unoptimized bounds for each; Specifically, we are setting $\kappa=\tau=1/2$ in the sampling bound and assuming $N=64$, $\Lambda = 8$, $T=1$, $x=1$, and $\mu=1$, except for the varying parameter.

Note in \fig{ratio} we see that, under our assumptions, the upper bound ratio is consistently greater than one in reasonable parameter ranges. This hints that a simulation using amplitude estimation may be beneficial only at more extreme lattice parameters.
However, since we are considering upper bounds on the cost of simulation, this conclusion is uncertain.

\begin{figure}[t!]
\centering
\includegraphics[]{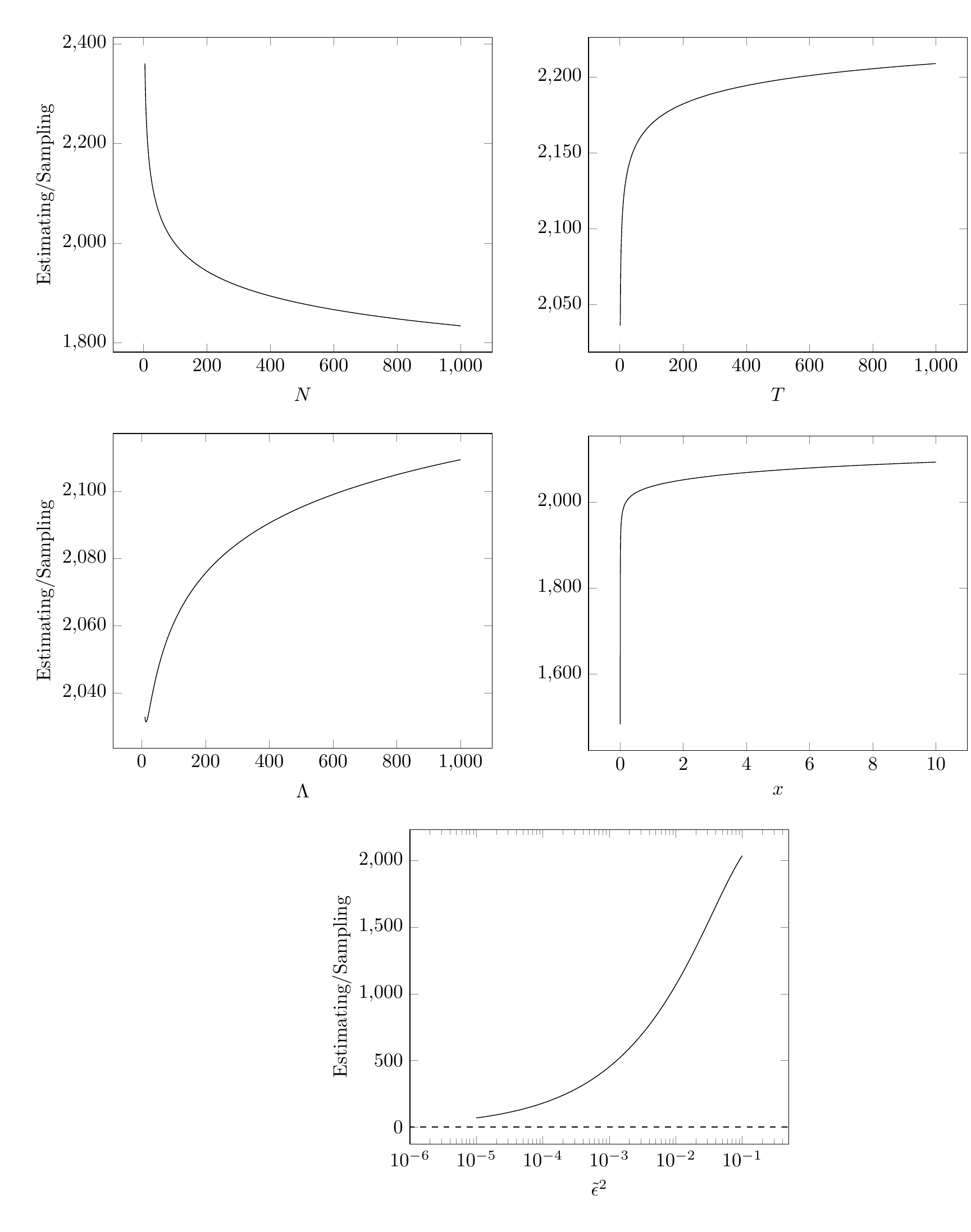}
\caption{Plots of ``Estimating/Sampling" ratio, the ratio of the unoptimized T-gate upper bound of \thm{onedtrotresults} (measurement via amplitude estimation) to the corresponding bound of \thm{onedtrotresults2} (measurement via sampling), when varying given parameters.
We are assuming $N=64$, $\Lambda = 8$, $T=1$, $x=1$, and $\mu=1$, except for the varying parameter.}
\label{fig:ratio}
\end{figure}

\end{document}